\documentclass[11pt,letterpaper]{article}

\usepackage[utf8]{inputenc}

\usepackage{amsmath, amsthm}
\usepackage{MnSymbol} 
\usepackage{graphicx,color}
\usepackage[hyphens]{url}
\usepackage{dsfont}
\usepackage{booktabs}
\usepackage[normalem]{ulem}
\usepackage{mathtools}
\usepackage{nicefrac}

\usepackage{hyperref}
\usepackage[nameinlink,capitalize]{cleveref}

\usepackage{multirow}
\usepackage{algorithm}
\usepackage[noend]{algpseudocode}
\usepackage{tikz} %

\usepackage{soul} %
\usepackage{amsmath}
\usepackage[T1]{fontenc}
\DeclareMathAlphabet{\mathpzc}{T1}{pzc}{m}{it}

\usepackage[giveninits=true, maxnames=10, style=alphabetic, defernumbers=true]{biblatex}
\AtBeginRefsection{\GenRefcontextData{sorting=ynt}}
\AtEveryCite{\localrefcontext[sorting=ynt]}

\usepackage{subcaption}

\usepackage{tabularx}
\newcolumntype{L}{>{\raggedright\arraybackslash}X}

\DeclareMathAlphabet\mathbb{U}{msb}{m}{n}

\numberwithin{equation}{section}
\newtheorem{theorem}{Theorem}[section]
\newtheorem{cor}[theorem]{Corollary}
\newtheorem{lemma}[theorem]{Lemma}
\newtheorem{remark}[theorem]{Remark}
\newtheorem{prop}[theorem]{Proposition}

\newtheorem{defin}[theorem]{Definition}
\newtheorem{assumption}[theorem]{Assumption}

\newenvironment{proof-of-lemma}[1][{}]{\noindent\textbf{Proof of {#1}}
  \hspace*{1em}}{\qed\\}

\newcommand{\nfrac}[2]{\nicefrac{#1}{#2}}

\newcommand{\cN}{\mathcal{N}}

\newcommand{\R}{\mathbb{R}}

\newcommand{\N}{\mathbb{N}}

\newcommand*{\E}{\mathbb{E}}

\DeclareMathOperator*{\Tr}{tr}
\providecommand{\tr}{\operatorname{tr}}

\newcommand*{\Otilde}{\widetilde{O}}

\newcommand*{\poly}{\mathrm{poly}}
\newcommand*{\polylog}{\mathrm{polylog}}

\newcommand*{\eps}{\varepsilon}

\DeclareMathOperator*{\argmin}{arg\,min}

\renewcommand{\le}{\leqslant}
\renewcommand{\ge}{\geqslant}
\renewcommand{\leq}{\leqslant}
\renewcommand{\geq}{\geqslant}

\providecommand{\abs}[1]{\lvert{#1}\rvert}

\providecommand{\norm}[1]{\lVert{#1}\rVert}

\DeclareMathOperator{\KL}{\mathsf{KL}}

\newcommand{\TV}{\mathsf{TV}}

\newcommand{\Ren}{\mathsf{R}}

\usepackage{blkarray}
\usepackage{cancel}
\usepackage{enumitem}
\usepackage{float}
\usepackage{mathrsfs}
\usepackage[framemethod=tikz]{mdframed}
\usepackage{microtype}
\usepackage{ragged2e}
\usepackage{sectsty}
\usepackage{thmtools}
\usepackage{thm-restate}
\usepackage[Symbolsmallscale]{upgreek}
\usepackage{array}

\usetikzlibrary{cd}

\hypersetup{citecolor=violet, colorlinks=true, linkcolor=blue}

\makeatletter
\patchcmd{\@counteralias}
{\@ifdefinable{c@#1}}
{\expandafter\@ifdefinable\csname c@#1\endcsname}
{}{}
\makeatother

\makeatletter
\newcommand{\opnorm}{\@ifstar\@opnorms\@opnorm}
\newcommand{\@opnorms}[1]{%
	\left|\mkern-1.5mu\left|\mkern-1.5mu\left|
	#1
	\right|\mkern-1.5mu\right|\mkern-1.5mu\right|
}
\newcommand{\@opnorm}[2][]{%
	\mathopen{#1|\mkern-1.5mu#1|\mkern-1.5mu#1|}
	#2
	\mathclose{#1|\mkern-1.5mu#1|\mkern-1.5mu#1|}
}
\makeatother

\newcommand{\PreserveBackslash}[1]{\let\temp=\\#1\let\\=\temp}
\newcolumntype{C}[1]{>{\PreserveBackslash\centering}p{#1}}

\newcommand\bs[1]{\boldsymbol{#1}}
\newcommand\mb[1]{\mathbf{#1}}

\newcommand\mc[1]{\mathcal{#1}}

\newcommand\ms[1]{\mathscr{#1}}

\newcommand\msf[1]{\mathsf{#1}}

\definecolor{MITBrown}{RGB}{164, 31, 50}

\DeclareMathOperator\ent{ent}

\DeclareMathOperator\id{id}

\DeclareMathOperator\law{law}

\DeclareMathOperator\one{\mathbf{1}}

\DeclareMathOperator\var{var}

\renewcommand{\Pr}{\mathbb{P}}

\newcommand{\D}{\mathrm{d}}

\newcommand{\Coup}{\ms C}

\newcommand\deq{\coloneqq}

\newcommand\mmid{\mathbin{\|}}

\newcounter{dummy}
\makeatletter
\newcommand\myitem[1][]{\item[#1]\refstepcounter{dummy}\def\@currentlabel{#1}}
\makeatother

\tikzset{
    shadedNode/.style={rectangle, draw=none, fill=blue!20, inner sep=1mm}
}

\global\long\def\inner#1{\left\langle #1\right\rangle }%

\newcommand{\sh}{\msf{sh}}
\newcommand{\aux}{\msf{aux}}
\newcommand{\msx}{\msf{x}}
\newcommand{\msp}{\msf{p}}

\newcommand{\msO}{\msf{O}}
\newcommand{\msOH}{\msf{OH}}
\newcommand{\msOHO}{\msf{OHO}}
\newcommand{\msOB}{\msf{OB}}
\newcommand{\msOBA}{\msf{OBA}}
\newcommand{\msOBAB}{\msf{OBAB}}
\newcommand{\msOBABC}{\msf{OBABC}}
\newcommand{\msOBABCO}{\msf{OBABCO}}

\newcommand{\eff}{{\rm H_2}}

\newcommand{\Renyi}{\msf{R}}

\newcommand{\target}{{\operatorname{target}}}

\newcommand{\op}{\operatorname{op}}
\newcommand{\betH}{\beta_{\rm H_1}}
\newcommand{\kapH}{\kappa_{\rm H_2}}
\usepackage[margin=1in]{geometry}

\makeatletter
\def\blfootnote{\gdef\@thefnmark{}\@footnotetext}
\makeatother

\allowdisplaybreaks

\begin{document}

    \title{Algorithmic warm starts for Hamiltonian Monte Carlo}
 	 \author{
         Matthew S.\ Zhang \\
         \small UToronto \\
         \texttt{\small matthew.zhang@mail.utoronto.ca} \and
     		Jason M.\ Altschuler\\
		 \small UPenn\\
		 \texttt{\small alts@upenn.edu}
		 \and
		 Sinho Chewi \\
		 \small Yale \\
		 \texttt{\small sinho.chewi@yale.edu}
     }  
    \date{}
	\maketitle

    \begin{abstract}
        Generating samples from a continuous probability density is a central algorithmic problem across statistics, engineering, and the sciences. For high-dimensional settings, Hamiltonian Monte Carlo (HMC) is the default algorithm across mainstream software packages. However, despite the extensive line of work on HMC and its widespread empirical success, it remains unclear how many iterations of HMC are required as a function of the dimension $d$. On one hand, a variety of results
        show that Metropolized HMC converges in $O(d^{1/4})$ iterations from a \emph{warm start} close to stationarity. On the other hand, Metropolized HMC is significantly slower without a warm start, e.g., requiring  $\Omega(d^{1/2})$ iterations even for simple target distributions such as isotropic Gaussians. 
        Finding a warm start is therefore the computational bottleneck for HMC.
        
        \par We resolve this issue for the well-studied setting of
        sampling from a probability distribution satisfying 
        strong log-concavity (or 
        isoperimetry) and third-order derivative bounds. We prove that \emph{non-Metropolized} HMC generates a warm start in $\Otilde(d^{1/4})$ iterations, after which we can exploit the warm start using Metropolized HMC. 
        Our final complexity of $\Otilde(d^{1/4})$ is the fastest algorithm for high-accuracy sampling under these assumptions, improving over the prior best of $\Otilde(d^{1/2})$. This closes the long line of work on the dimensional complexity of MHMC for such settings, and also provides a simple warm-start prescription for practical implementations.
    \end{abstract}
    
    \newpage

    \small
	\setcounter{tocdepth}{2}
	\tableofcontents	
	\normalsize
	\newpage

        \section{Introduction}

Generating samples from a high-dimensional probability distribution $\pi$ on $\R^d$ is a central algorithmic challenge in statistics, machine learning, computational mathematics, physics, chemistry, and more~\cite{james1980monte, jerrum1996markov, robert1999monte, liu2001monte, andrieu2003introduction, brooks2011handbook, dwork2014algorithmic,chewibook}. The de facto algorithmic approach is Markov chain Monte Carlo (MCMC): simulate a Markov chain that is designed to have stationary distribution (approximately) equal to $\pi$. 
The basic idea is that after simulating the Markov chain for a sufficiently long period---the \emph{mixing time}---then the final iterate has law close to $\pi$. The effectiveness of this idea has led to MCMC being listed as one of the top ten algorithms of the 20th century~\cite{cipra2000best,dongarra2000guest}. 

\par Among the extensive suite of MCMC algorithms, Metropolized Hamiltonian Monte Carlo (MHMC\footnote{In the literature, Metropolized HMC is often simply called HMC. To avoid confusion, throughout we refer to Metropolized HMC as MHMC, and we refer to (unadjusted, non-Metropolized) HMC as HMC.}) 
has long been the practitioners' algorithm of choice 
for sampling from high-dimensional distributions with tractable gradients. Today MHMC is the workhorse algorithm in mainstream software packages such as NumPyro~\cite{phan2019composable}, PyMC~\cite{abril2023pymc}, Stan~\cite{gelleeguo15stan}, and TensorFlow~\cite{dillon2017tensorflow}.

The overarching idea behind MHMC is to give the iterates of the Markov chain ``momentum'' so that they can accelerate in important directions. Concretely, HMC augments the state space by introducing an auxiliary momentum variable and then proposes long-distance moves by approximately following Hamiltonian dynamics for time $T \asymp 1$ in the resulting phase space (see \S\ref{sec:prelim} for details). Because the Hamiltonian ODE does not have a closed form solution, it is integrated numerically---most commonly with the St\"ormer--Verlet ``leapfrog'' scheme---using a step size $h$ for $T/h$ integration steps. Each integration step requires one gradient query to the negative log-density $-\log \pi$; these queries are the computational bottleneck and are therefore the standard measure of computational complexity. Finally, MHMC builds upon HMC by applying a Metropolis--Hastings accept/reject step to the HMC proposal; 
this ensures that the stationary distribution is exactly equal to $\pi$ despite discretization error in the HMC step. Empirically, these momentum-based trajectories mitigate random-walk behavior intrinsic to other MCMC algorithms, leading to fast mixing even in high dimension~\cite{neal2011mcmc}.

However, despite the widespread practical usage of MHMC for high-dimensional sampling, it remains unclear how many iterations are required as a function of the dimension $d$. The current understanding is surprisingly nuanced due to a phase transition: the convergence rate of MHMC depends on whether the current iterate is already close to stationarity, as we detail below.

\paragraph*{Predictions of $d^{1/4}$ assuming a warm start.} 
On one hand, an influential line of work dating back to the 1980s has shown that, if initialized sufficiently close to stationarity, MHMC converges in $\sim d^{1/4}$ iterations. 
Traditionally, these results were of the following flavor (since acceptance probability is simpler to analyze than mixing time):
in order for the acceptance probability to not vanish in the limit $d \to \infty$, then the step size of HMC (for the standard leapfrog discretization) should be of size $h \asymp d^{-1/4}$, hence HMC takes roughly $d^{1/4}$ iterations to move $O(1)$ distance in the state space. Such results were first shown via scaling limit predictions based on statistical physics~\cite{creutz88global,gupta88tuning}, then later shown rigorously for Gaussian $\pi$~\cite{kennedy91acceptances} and more general families of i.i.d.\ product distributions satisfying fourth-order derivative bounds~\cite{Beskos13}. The $d^{1/4}$ scaling of MHMC is appealing since it is better than other MCMC algorithms such as the Metropolized Random Walk (MRW) and the Metropolis-Adjusted Langevin Algorithm (MALA); those respectively require $O(d)$ and $O(d^{1/3})$ scalings to move $O(1)$ distance in the state space, even with warm starts~\cite{gelman1997weak, roberts1998optimal}.

\par Of course, traversing the state space does not imply mixing; the recent breakthrough of~\cite{chengatmiry23} upgraded the aforementioned results to a bona fide $O(d^{1/4})$ mixing time for MHMC---again assuming a sufficiently warm start.
This result also extended beyond i.i.d.\ product distributions: similarly to other non-asymptotic results of this flavor, their results hold when $\pi$ enjoys third-order derivative bounds and isoperimetric inequalities (or strong log-concavity).

\paragraph*{Significantly slower convergence without a warm start.}  On the other hand, the convergence of MHMC is dramatically slower if not initialized at a warm start. This is the relevant regime in practice, since warm starts are not available a priori.
 The worsened dimension dependence occurs because, at a cold start, MHMC requires significantly smaller step sizes to ensure a non-vanishing acceptance probability. These small step sizes lead to slow movement. Hence, many iterations are required even to traverse the state space, let alone to mix. This failure is not a limitation of existing analysis techniques: it is a fundamental phenomenon that occurs even for simple target distributions such as isotropic Gaussian $\pi$ and initialization at its mode, see Figure~\ref{fig:accept}. In this case, even the classical prediction based on non-vanishing acceptance probability degrades from $d^{1/4}$ to $d^{1/2}$ without a warm start~\cite{lee2021lower,lee25thesis}, and leads to $\widetilde{\Omega}(d^{1/2})$ lower bounds on the mixing time of MHMC from a cold start~\cite[Theorem 4]{lee2021lower}.

\begin{figure}[t]
    \begin{subfigure}[b]{0.51\textwidth}
        \centering
        \includegraphics[width=0.95\textwidth]{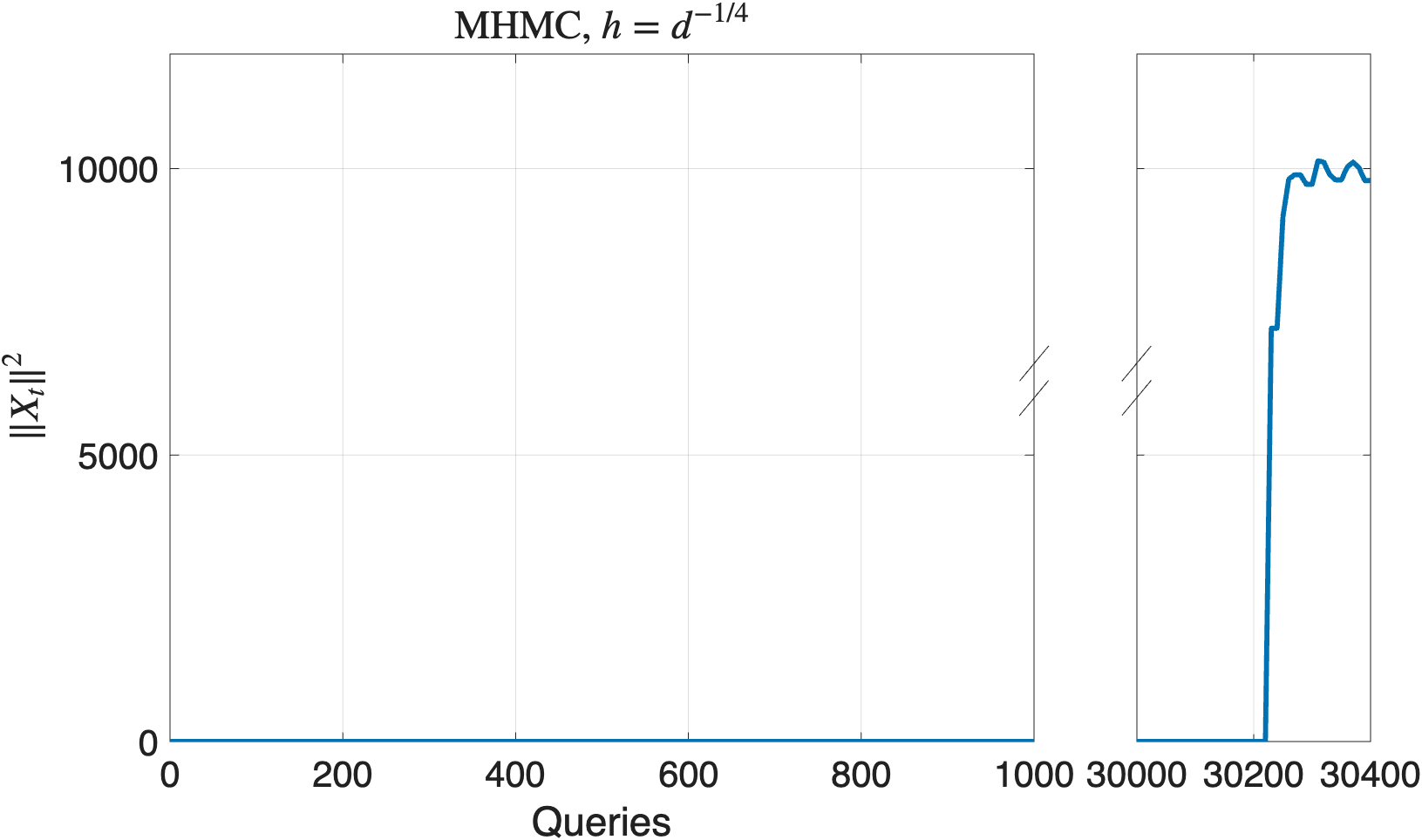}
        \caption{\footnotesize With large steps $h \asymp d^{-1/4}$, MHMC gets stuck in cold starts due to very low acceptance probability of the Metropolis filter. Here it is slow to escape the origin, let alone mix. 
        }
        \label{fig:norm-large}
    \end{subfigure}
    \hfill
    \centering
    \begin{subfigure}[b]{0.46\textwidth}
        \centering
        \includegraphics[width=0.95\textwidth]{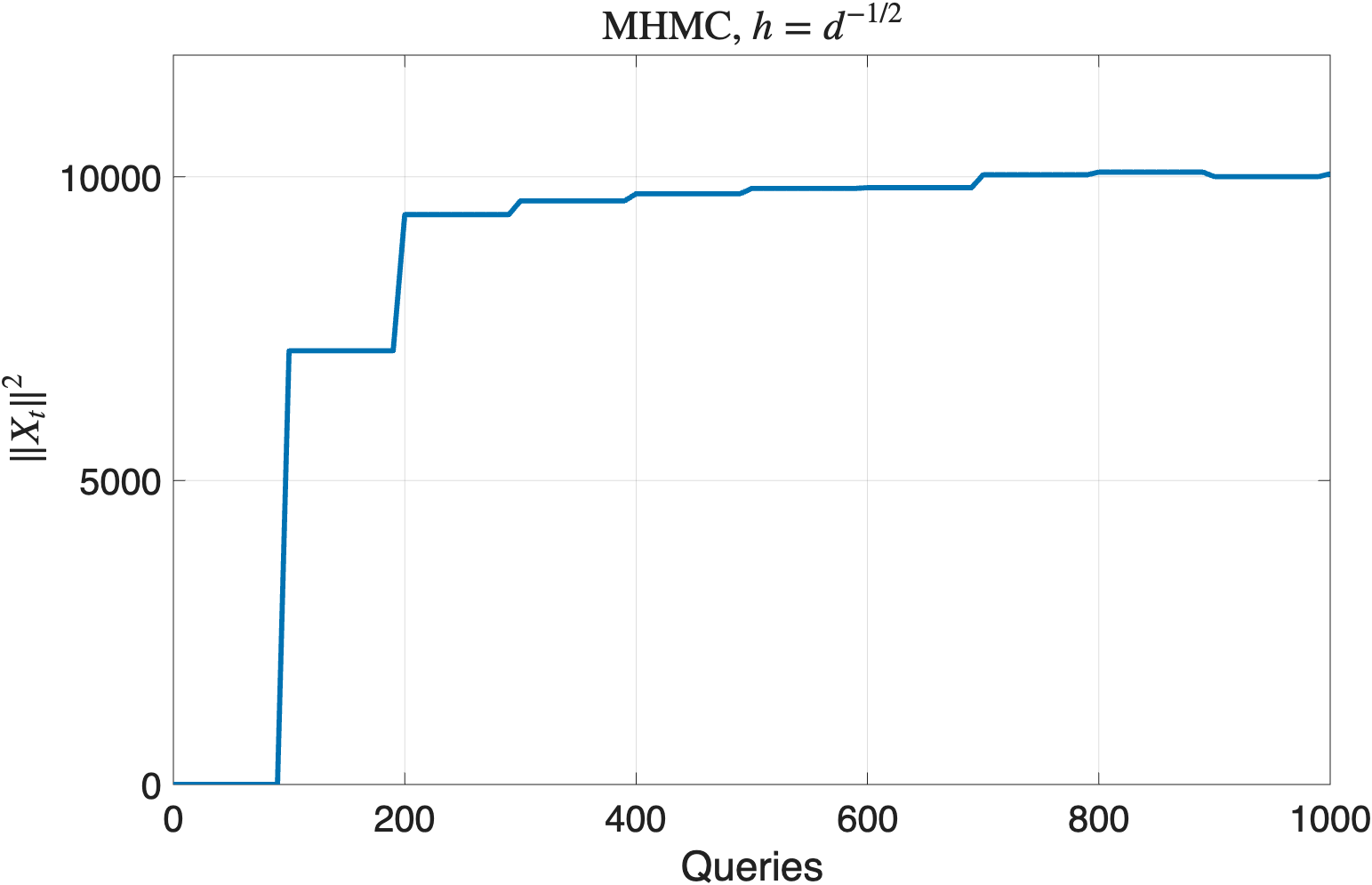}
        \caption{\footnotesize With small steps $h \asymp d^{-1/2}$, MHMC has high acceptance probability but moves slowly, requiring at least $1/h \asymp d^{1/2}$ steps to traverse the space, let alone mix.}
        \label{fig:norm-small}
    \end{subfigure}
    \caption{\footnotesize
     The convergence of MHMC is heavily dependent on the step size $h$. Large step sizes $h \asymp d^{-1/4}$ classically lead to fast convergence from a warm start, but can get stuck in cold starts due to very low acceptance probability (\textbf{left}). Small step sizes $h \asymp d^{-1/2}$ fix that issue but lead to slow movement, requiring at least $1/h \asymp d^{1/2}$ steps to traverse the space, let alone mix (\textbf{right}). Illustrated for the simple target $\pi = \cN(0, I)$ in dimension $d = 10^4$, with ``cold start'' initialization at the mode $x_0 = 0$. Reproducibility details: MHMC is repeatedly integrated for $T=1$ unit of continuous time via $1/h$ leapfrog steps of size $h$, as is standard. Similar qualitative phenomena are observed for other settings. 
    }
    \label{fig:norm}
\end{figure}

\paragraph*{Algorithmic framework for generating a warm start.} Summarizing, the convergence rate of MHMC depends critically on the initialization.
Given a warm start, MHMC can use large step sizes while ensuring large acceptance probability, leading to rapid mixing. But without a warm start, MHMC requires much smaller stepizes to ensure reasonable acceptance probability, leading to much slower mixing. This suggests a natural opportunity: can one efficiently compute a warm start for MHMC? A positive answer would finally make the elegant mathematical predictions of $d^{1/4}$ into a truly algorithmic result, by removing the conditional warm start assumption that is currently unverifiable and unimplementable in practice. 

This warm start question also arises for other Metropolized algorithms, since the acceptance probability of the Metropolis--Hastings filter can be highly sensitive to initialization~\cite{jourdain15optimal,kuntz18}. A notable example is the Metropolis-adjusted Langevin algorithm (MALA), which also exhibits a mixing time dichotomy: $d$ from a cold start versus $d^{1/2}$ from a warm start~\cite{chewi2021optimal,lee2021lower,wuschche2022minimaxmala} in the standard strongly log-concave and smooth setting, without higher-order smoothness; see the related work in \S\ref{ssec:intro:rel} for details. Motivated by this gap,~\cite{AltChe24Warm} initiated a program of computing warm starts for sampling algorithms. The main result was an algorithm that computes a warm start in $d^{1/2}$ iterations, thereby removing the complexity gap for MALA.

The purpose of this paper is to continue this program, now for MHMC. The key question can be stated precisely as follows:
\begin{align*}
	\text{\emph{Is there an algorithm that uses }} \widetilde{O}&(d^{1/4}) \text{\emph{ gradient queries}} \\ \text{\emph{ to output a measure }} \mu_0 \text{\emph{ with }} &\chi^2(\mu_0 \mmid \pi) \leq O(1)\text{\emph{?}}
\end{align*}
The basic idea is that if yes, then one should instead run a \emph{two-phase procedure}: (1) use this algorithm to compute a warm start; (2) run MHMC from that warm start. It is crucial that the warm start take only $d^{1/4}$ iterations, otherwise the cost of computing the warm start dominates the subsequent cost of running MHMC. Unfortunately this is the current situation: the fastest known algorithms take significantly longer to produce a warm start than to use one. The exception is the Gaussian setting~\cite{apers2024hamiltonian}, where a $d^{1/4}$ warm start was shown; however, this analysis was tailored for Gaussians and does not extend to a more general framework.

\par What algorithm should be used to warm start MHMC? A key insight in~\cite{AltChe24Warm,apers2024hamiltonian} is that warm starts for Metropolized algorithms can be obtained via \emph{non-Metropolized} algorithms. Indeed, the purpose of the Metropolis filter is to eliminate discretization bias so that $\pi$ is exactly stationary; however, warm starts only require getting moderately close to $\pi$ and therefore can afford some discretization bias. 
This suggests removing the Metropolis filter during the initial warm start phase, to avoid being throttled by low acceptance probability at a cold start. 
Such non-Metropolized methods are often called \textit{unadjusted} algorithms or \textit{low-accuracy} algorithms---the latter name because their iteration complexity typically scales as $\poly(1/\varepsilon)$ in the target accuracy $\varepsilon$ due to the asymptotic bias, in contrast to the $\polylog(1/\varepsilon)$ \textit{high-accuracy} dependence that Metropolis-adjusted methods can achieve once they are in their stable regime. Note that this low-accuracy scaling is not problematic since warm starts only require moderate error $\varepsilon = \Theta(1)$; what is important is avoiding the throttling of the Metropolis filter.

\par A key challenge for warm starting---whether for MHMC, MALA, or any other algorithm---is that warm starts require a strong performance metric. In particular, the initialization measure $\mu_0$ must be close to $\pi$ in the chi-squared divergence (or more generally any R\'enyi divergence $\Ren_q$ of order $q > 1$) rather than other common metrics such as total variation, Wasserstein, or KL divergence.\footnote{Briefly, this is because R\'enyi divergence closeness
ensures that if the $\pi$-average acceptance probability is of constant size, then so is the $\mu_0$-average acceptance probability. Other metrics such as total variation, Wasserstein, and KL divergence do not provide multiplicative guarantees, and thus in particular do not rule out bottleneck initializations with non-negligible $\mu_0$-probability but exponentially small $\pi$-probability; hence large $\pi$-average acceptance probability does \emph{not} imply large $\mu_0$-average acceptance probability. See~\cite[\S1.2.1]{AltChe24Warm} for a detailed discussion.} For these strong metrics, the fastest known algorithm for warm starting requires $d^{1/2}$ complexity~\cite{AltChe24Warm}; see the related work in \S\ref{ssec:intro:rel}. This $d^{1/2}$ complexity is useless for the purpose of warm starting MHMC as it falls short of the desired $d^{1/4}$ complexity.

\subsection{Contributions}\label{ssec:intro:cont}

Our main result improves the complexity of high-accuracy sampling from $d^{1/2}$ to $d^{1/4}$ under standard assumptions of strong-log-concavity and smoothness. This scaling $d^{1/4}$ has long been believed optimal for MHMC, as detailed above.

\begin{theorem}[Informal statement of Theorem~\ref{thm:main-slc}]\label{thm-intro:main}
     Consider a target distribution $\pi \propto \exp(-V)$ on $\R^d$, where $V$ is strongly convex, smooth, and has Frobenius-Lipschitz Hessian. 
     There is an algorithm that uses $O(d^{1/4} \log^2 1/\eps)$ first-order queries to produce a sample from a distribution $\mu$ where $\chi^2(\mu \mmid \pi) \leq \eps$. 
\end{theorem}

Previously, the state-of-the-art result for high-accuracy log-concave sampling had complexity $\widetilde O(d^{1/2} )$~\cite{Fan23dimension, AltChe24Warm, Chen+26HighAccDiffusion}---in this paragraph, we only track the dimension dependence, as we focus on high-accuracy samplers whose dependence on $\varepsilon$ scales as $\polylog(1/\varepsilon)$. Despite extensive empirical evidence suggesting that MHMC is faster than MALA\@, MHMC has proven more challenging to analyze. MALA was recently shown to have tight scaling $O(d^{1/2})$ in the strongly log-concave and smooth setting~\cite{AltChe24Warm}. Under the stronger Frobenius-Lipschitz Hessian assumption,~\cite{chengatmiry23} showed that MHMC admits an improved $d^{1/4}$ complexity bound from a warm start, i.e., when initialized at a distribution $\mu_0$ satisfying $\chi^2(\mu_0 \mmid \pi) = O(1)$. 
However, this is a stringent requirement, e.g., the natural Gaussian initialization centered at the mode of $\pi$ is exponentially far in the relevant metric $\chi^2(\mu_0 \mmid \pi) = \exp(O(d))$~\cite{chewibook}.
Since the previous best algorithm for generating a warm start required $\widetilde O(d^{1/2})$ iterations~\cite{AltChe24Warm}, this also bottlenecked the complexity for MHMC\@.

The core result underlying Theorem~\ref{thm-intro:main} is a $d^{1/4}$ algorithm for computing warm starts. This enables us to algorithmically exploit, for the first time, the breakthrough of~\cite{chengatmiry23}. Our result applies under the same assumptions on $\pi$.

\begin{figure}[t]
    \begin{subfigure}[b]{0.32\textwidth}
        \centering
        \includegraphics[width=\textwidth]{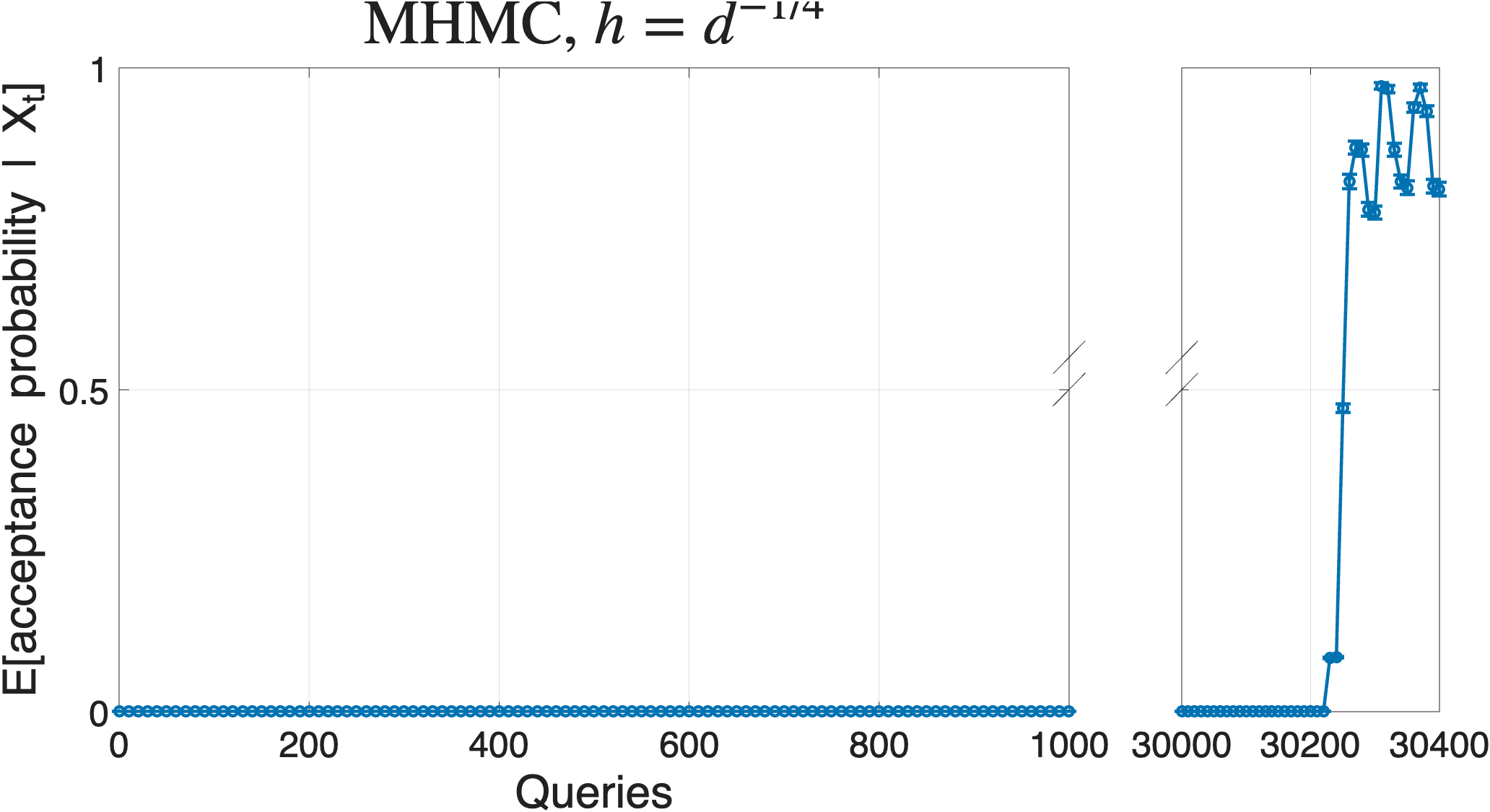}
        \caption{\footnotesize 
        MHMC with large steps.
        }
        \label{fig:prob-large}
    \end{subfigure}   
    \hfill
    \begin{subfigure}[b]{0.32\textwidth}
        \centering
        \includegraphics[width=\textwidth]{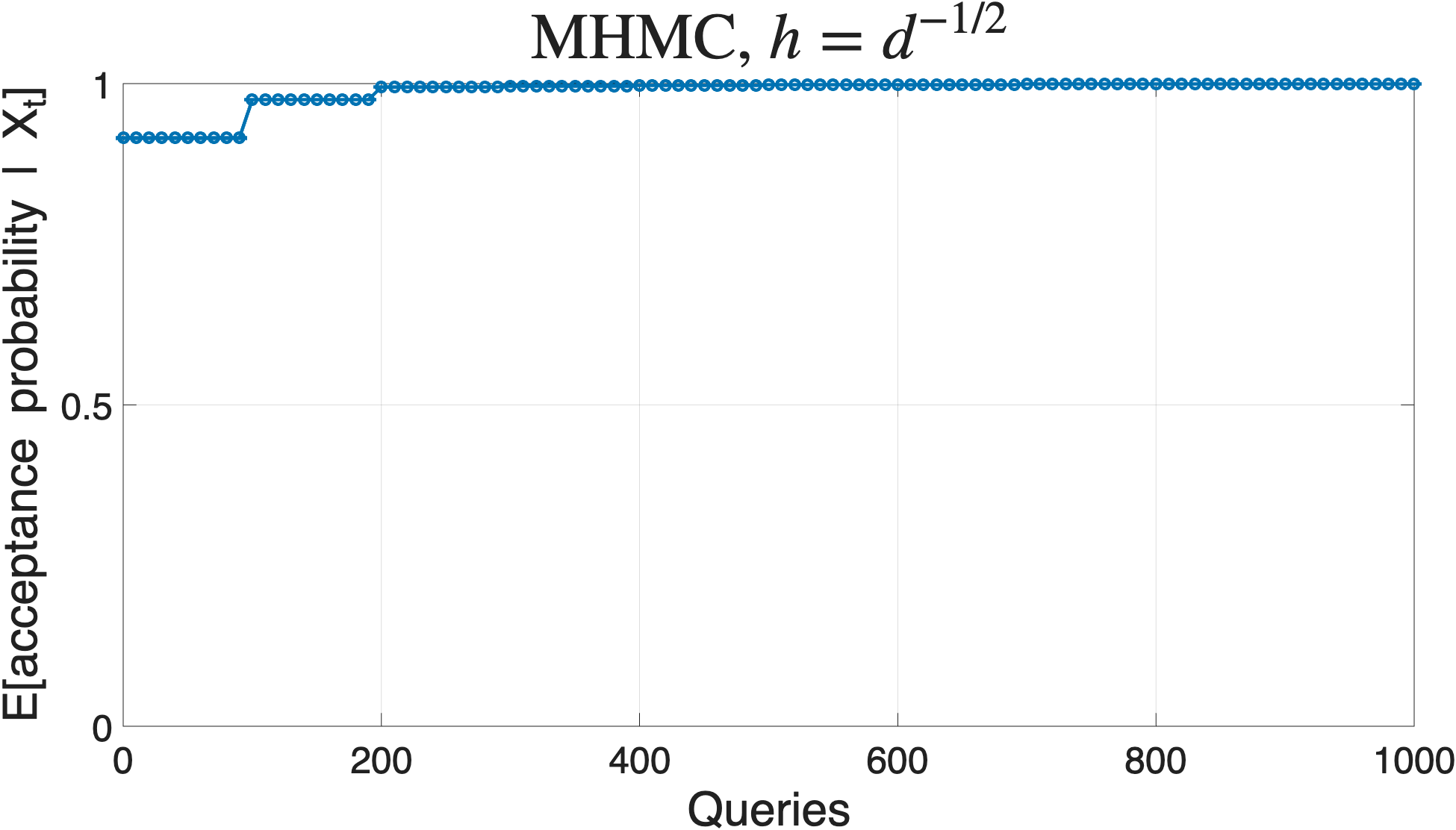}
        \caption{\footnotesize
        MHMC with small steps.
        }
        \label{fig:prob-small}
    \end{subfigure}
    \hfill
     \centering
    \begin{subfigure}[b]{0.32\textwidth}
        \centering
        \includegraphics[width=\textwidth]{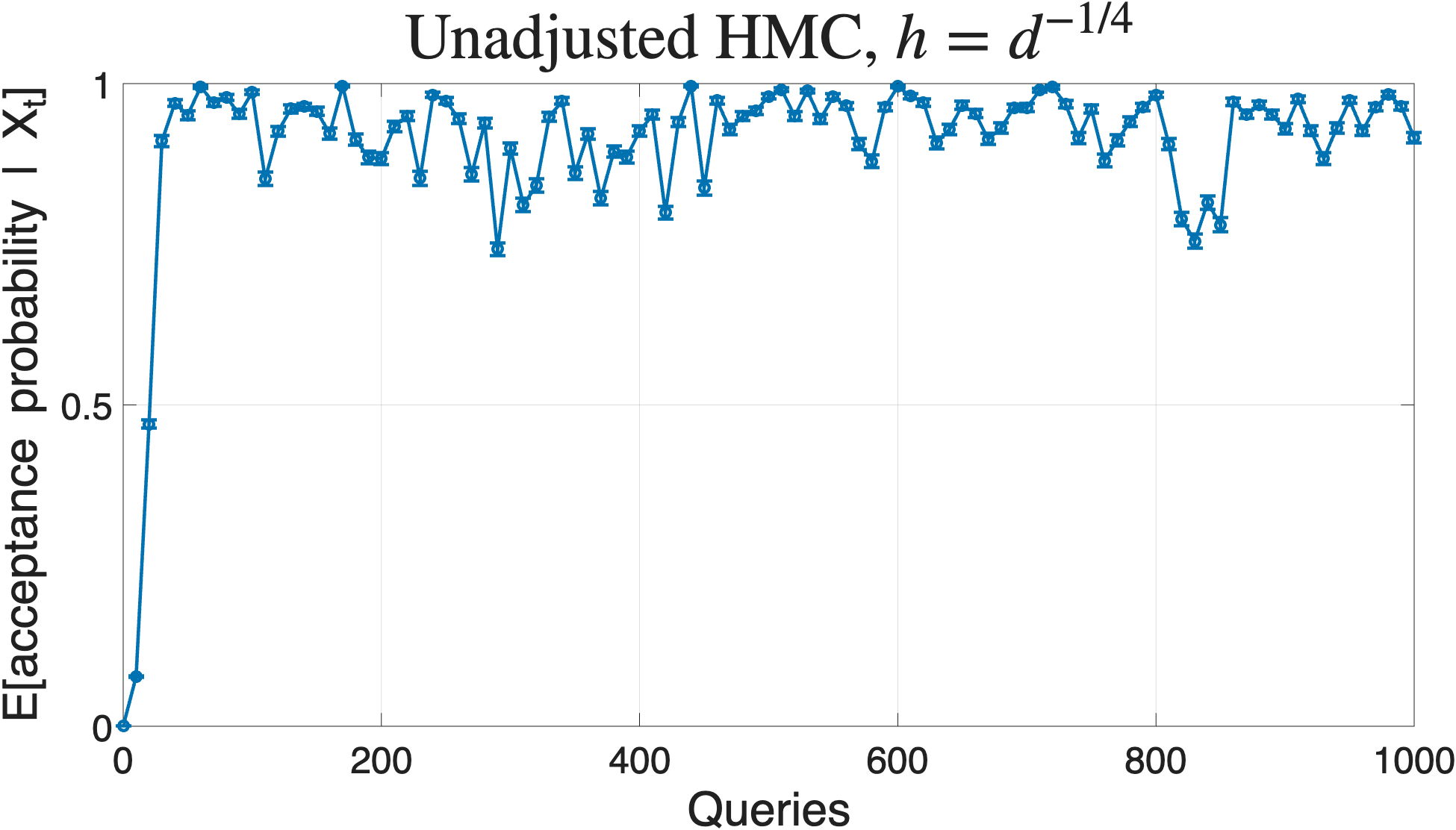}
        \caption{\footnotesize 
        \emph{Unadjusted} HMC with large steps.
        }
        \label{fig:prob-large-unadjusted}
    \end{subfigure}
    \caption{\footnotesize 
   A key algorithmic insight is that \emph{unadjusted} HMC rapidly escapes cold starts using large step sizes $h \asymp d^{-1/4}$. This algorithm quickly leads to iterates which, if used as an initialization for MHMC, would have high acceptance probability (\textbf{right}).
   This motivates our two-phase algorithmic proposal: escape the cold start via unadjusted HMC, then exploit the warm start using MHMC. Removing the Metropolis filter bypasses the issue of low acceptance probability at cold starts (\textbf{left}) without requiring the use of small step sizes which result in slow movement (\textbf{middle}).
    }
    \label{fig:accept}
\end{figure}
 
\begin{theorem}[Informal statement of Theorem~\ref{thm:main-warm}]\label{thm-intro:warm}
    Consider a target distribution $\pi \propto \exp(-V)$ on $\R^d$, where $V$ is strongly convex, smooth, and has Frobenius-Lipschitz Hessian. There is an algorithm that uses $\widetilde O(d^{1/4})$ first-order queries to produce a sample from a distribution $\mu$ where $\chi^2(\mu \mmid \pi) \leq 1$.  
\end{theorem}

Interestingly, we compute this warm start using a variant\footnote{We use a slight variant of the standard leapfrog discretization of Hamiltonian's equations which involves a partial rather than full momentum refresh. This is necessary for our analysis techniques, details in \S\ref{sec:app}.} of \emph{unadjusted} HMC. 
Removing the Metropolis--Hastings filter enables using large step sizes without being throttled by low acceptance probability. This enables rapidly escaping cold starts, see Figure~\ref{fig:accept}.
Our full algorithm in Theorem~\ref{thm-intro:main} therefore has two phases: generate the warm start via unadjusted HMC; then exploit the warm start via MHMC. This theoretically inspired framework offers guidance to practical implementations of MHMC.

There is a long line of work on non-asymptotic guarantees for unadjusted HMC. Previously, the state-of-the-art complexity was $\widetilde O(d^{1/4} \poly(1/\eps))$~\cite{MangoubiVishnoi18}, but crucially only in Wasserstein or TV-type error metrics and was therefore inapplicable to warm start results which require the much stronger metric of chi-squared (or any R\'enyi divergence with R\'enyi parameter strictly larger than $1$). See \S\ref{sec:results} for details. Theorem~\ref{thm:main-warm} achieves this warm start.

\par We conclude this discussion with two remarks.

\begin{remark}[Regularity]
    For simplicity, Theorems~\ref{thm-intro:main} and~\ref{thm-intro:warm} are stated under the standard assumption of strongly log-concave $\pi$ (equivalently, strongly convex $V = -\log \pi$). This can be relaxed to weak log-concavity or standard isoperimetric inequalities on $\pi$ such as log-Sobolev inequalities or Poincar\'e inequalities. Details in \S\ref{sec:results}.
\end{remark}

\begin{remark}[Condition number]
    Modulo logarithmic factors, our MHMC runtimes scale optimally in the dimension $d$ and accuracy $\eps$. The other parameter that arises in the analysis is the condition number $\kappa$, i.e., the ratio of the largest and smallest eigenvalues of $\nabla^2 V$. The direct use of our two-phase algorithm has scaling $\kappa^{3/2}$. But a standard reduction~\cite{chen2022improved,AltChe24Warm} via a proximal sampler outer loop reduces this to linear scaling in $\kappa$, which matches the state-of-the-art results for MHMC~\cite{chenetal2020hmc,chengatmiry23}. It is an interesting question for future work whether a tightening of our analysis can directly give linear scaling in $\kappa$ without any changes to the algorithm. 
\end{remark}

\subsection{Related work}\label{ssec:intro:rel}

In addition to the literature described above, here we provide further context about related work.

\paragraph*{HMC.} Hamiltonian Monte Carlo, originally called Hybrid Monte Carlo, was first proposed in the seminal paper~\cite{duane1987hybrid} for the study of lattice models in quantum field theory. In the ensuing decades, HMC became widely popular in the statistics and data science communities; see for example the surveys and textbooks~\cite{liu2001monte,neal2011mcmc,bou2018geometric} for a historical account of HMC and a discussion of its applications.
To this day, HMC remains a popular algorithm in practice, in both its original form as well as modern variants, including adaptive parameter selection such as the No-U-Turn Sampler (NUTS)~\cite{hoffman2014no}, Riemannian variants~\cite{girolami2011riemann}, 
parallelized variants~\cite{calderhead2014general,lee2018algorithmic}, and different discretizations of Hamilton's equations that use St\"ormer--Verlet leapfrog integrator~\cite{neal2011mcmc}, partial momentum refreshments~\cite{monmarche2021high}, Runge--Kutta--Nystrom methods~\cite{bou2025randomized}, randomized integration times~\cite{bou2017randomized,apers2024hamiltonian, BouMar25HMC}, and more~\cite{neal2011mcmc,bou2018geometric}. 

\par A key question throughout this HMC literature is to understand the scaling on the dimension $d$. This is perhaps the most dominant question on the theoretical side, and on the practical side it guides implementations. For example, the popular software package PyMC scales MHMC step sizes as $h \asymp d^{-1/4}$ based on theoretical predictions about effective step sizes from a warm start~\cite{pymc-nuts-docs}. The literature on HMC is extensive; below we highlight only the most relevant results, focusing in particular on dimensional scaling.

\paragraph*{Idealized HMC.} HMC requires simulating Hamilton's equations of motion, which in general cannot be integrated in closed form and thus require numerical discretization. Idealized HMC refers to HMC dynamics assuming exact integration. Although this is not algorithmically implementable, it provides a useful baseline. Towards this end, a line of work has developed Wasserstein convergence guarantees for idealized HMC based on increasingly tighter coupling arguments. This led to convergence rates of order $\kappa^2 \log 1/\eps$~\cite{MangoubiSmith21}, $\kappa^{1.5} \log 1/\eps$~\cite{lee2018algorithmic}, and finally the tight bound $\kappa \log 1/\eps$~\cite{chenvempala2019hmc} in the setting where $\pi \propto \exp(-V)$ is $\kappa$-conditioned, i.e., $V$ is $\alpha$-strongly convex and has $\beta$-Lipschitz gradients, where $\kappa \deq \beta / \alpha$. Note that these results are dimension-independent since they do not account for the discretization error of Hamiltonian's equations, discussed next. 

\paragraph*{Metropolis-adjusted HMC (MHMC).} Numerically integrating Hamilton's equations introduces discretization bias; following the original paper~\cite{duane1987hybrid}, the standard guidance for HMC is to eliminate this bias by adding a Metropolis--Hastings filter. This ensures that the stationary distribution is exactly $\pi$. 
The downside is that the filter can lead to slowdowns if the acceptance probability is low. Understanding this issue---and how it depends on the step size scaling and the initialization warmness---has been a central topic of study for MHMC since the 1980s. As detailed earlier, it is now well understood that the complexity scales as $d^{1/4}$ from a warm start, but degrades to $d^{1/2}$ from a cold start. Such results were originally shown for the simpler proxy of non-vanishing acceptance probability~\cite{creutz88global,gupta88tuning,kennedy91acceptances,Beskos13,lee2021lower,lee25thesis}, but have recently been established for end-to-end mixing time bounds~\cite{lee2021lower, chengatmiry23, lee25thesis}. It was therefore open whether $d^{1/4}$ iteration complexity is achievable for MHMC---or for any high-accuracy sampling algorithm---since it is unknown how to compute warm starts faster than $d^{1/2}$~\cite{AltChe24Warm}. 

\paragraph*{Unadjusted HMC.} Without the Metropolis--Hastings filter, unadjusted HMC has discretization bias. This leads to so-called low-accuracy sampling, in that the resulting iteration complexity scales in the accuracy $\eps$ as $\poly(1/\eps)$ rather than the exponentially faster $\polylog(1/\eps)$ rate of MHMC. However, the removal of the filter simplifies analyses and also leads to faster convergence to moderate accuracy $\eps$ since there is no slowdown from rejected steps. 

\par Over the past decade, a number of elegant results have shown $d^{1/4} \poly(1/\eps)$ iteration complexities for unadjusted HMC in various settings---although a key limitation of all existing results is that they are restricted to Wasserstein or TV metrics, which do not provide a warm start as described above. Closest to our setting is the $d^{1/4} \poly(1/\eps)$ result of~\cite{MangoubiVishnoi18} for Wasserstein convergence of leapfrog HMC under strong log-concavity and third-order derivative bounds.~\cite{apers2024hamiltonian} show $d^{1/4} \poly(1/\eps)$ total variation convergence in order to warm start MHMC, but their results apply only to Gaussian targets $\pi$. A topical recent direction is extending beyond log-concavity; for example, recent work has investigated target distributions $\pi$ that are strongly log-concave outside of a ball~\cite{bou2020coupling,bou2023mixing,chak2025reflection} or more generally satisfy isoperimetric inequalities such as the log-Sobolev inequality~\cite{camrud2023second}, and/or have weakly interacting mean-field potentials~\cite{bou2023mixing,bou2023convergence,camrud2023second}. In the latter setting, $d^{1/4}$ scalings were shown in~\cite{bou2023convergence,camrud2023second}. Other $d^{1/4}$ scalings have also been shown for variants of unadjusted HMC using NUTS~\cite{bou2024mixing} or partial momentum refreshments~\cite{monmarche2021high}; the former result is specifically for Gaussian $\pi$ and the latter for separable $\pi$ (i.e., product distributions in an orthonormal basis). Additionally, $d^{1/4}$ scaling for Wasserstein guarantees have also been shown for algorithms different from unadjusted HMC that nevertheless leverage the same overarching principle of lifting and high-order integration~\cite{mou2021high}. However, despite these many interesting results with $d^{1/4}$ scaling in the past decade, it remains open whether these can be extended to R\'enyi guarantees---as is required for warm starting MHMC. Our main result (Theorem~\ref{thm:main-warm}) resolves this question. Note also that since our bounds hold in R\'enyi divergence, they also imply mixing in weaker metrics like Wasserstein, TV, and KL by standard comparison inequalities~\cite{chewibook}.

We pause to focus on the works which have the clearest resemblance to our analysis. These include the works of~\cite{bou2023mixing, monmarche2024entropic, chak2025reflection, gouraud2025hmc}. In particular, we highlight the short-time regularity which appears in~\cite[Theorem 2]{monmarche2024entropic} and~\cite[Lemma 15]{bou2023mixing}, which is similar to the rate obtained in our regularity bounds for short times. Furthermore, many of the error lemmas associated with the discretization are reminiscent of the bounds found in e.g.,~\cite[Section 3]{monmarche2024entropic}. These works also contain extensions of these analyses to stochastic gradients and other settings. Our main ideas are to extend these regularization results to long times using the shifted Girsanov ideas in~\cite{GuiWan12DegenBismut, scr1}, and to introduce a corrective step for the position in order to improve the entropic error, inspired by the schemes proposed in~\cite{zhang2025analysis, scr4}. See \S\ref{sec:overview} for a technical overview.

\paragraph*{MALA.} MALA faces a similar warm start issue as MHMC, as briefly mentioned above; we detail this here. Just like MHMC, MALA can use much larger step sizes from a warm start while ensuring reasonable acceptance probability, namely $h \asymp d^{-1/3}$ from a warm start~\cite{roberts1998optimal} but only $h \asymp d^{-1/2}$ from a cold start~\cite{kuntz18}, under third-order derivative assumptions. Quantitatively, mixing time analysis was carried out under weaker (second-order derivative) assumptions, leading to a similar gap in the complexities from a warm start versus a cold start, namely $O(d^{1/2} \, \polylog(1/\eps))$~\cite{chewi2021optimal} versus $\Omega(d \, \polylog(1/\eps))$~\cite{lee2021lower}. This discrepancy suggested the analogous key question for MALA: can one efficiently compute a warm start in $d^{1/2}$ iterations? This natural question was asked by several groups~\cite{chewi2021optimal, lee2021lower, luwang2022zigzag, wuschche2022minimaxmala, Che+24LMC} and was recently resolved by~\cite{AltChe24Warm} by using a different algorithm (Underdamped Langevin Monte Carlo) to warm start MALA. This led to an end-to-end $O(d^{1/2} \polylog(1/\eps))$ complexity for a high-accuracy algorithm for sampling from well-conditioned distributions that remains the state-of-the-art. The main result of this paper uses MHMC to improve this complexity to $O(d^{1/4} \polylog(1/\eps))$ under higher-order smoothness.

\paragraph*{Underdamped Langevin.} The behavior of HMC closely resembles the behavior of the underdamped Langevin dynamics, a classical SDE system that dates back to Kolmogorov~\cite{Kol34}.
The underdamped Langevin system has a rich convergence theory, based on the theory of hypocoercivity as pioneered by Villani in his monograph~\cite{villani2009hypocoercivity}.
Our analysis of HMC with partial momentum refreshments in \S\ref{sec:oho} draws upon similar intuition and recent developments in this literature; see the technical overview section \S\ref{ssec:oho_overview} for a more detailed discussion.

\paragraph*{Shifted composition.}
As explained in \S\ref{ssec:oho_overview}, our analysis builds upon certain regularity estimates for diffusion processes, known in the literature as \emph{Harnack inequalities} or \emph{reverse transport inequalities}. By now, there are various ways to establish these estimates, but most relevant for our work is the method of coupling as pioneered by F.-Y.\ Wang and his coauthors~\cite{ArnThaWan06HarnackCurvUnbdd, Wang12Coupling}. The \emph{shifted composition} sequence of works~\cite{scr1, scr2, scr3, scr4} develops an information-theoretic reformulation of this method which, crucially, enables its application to discrete-time Markov kernels, including those arising from MCMC discretizations. In \S\ref{ssec:obabco_overview}, we discuss how our analysis of HMC discretization leverages and builds upon the shifted composition framework.

\paragraph*{Independent work.} Independently, the recent paper~\cite{bou2026tail} provided a suite of results with a similar goal of showing R\'enyi bounds for HMC. They consider full momentum refreshments which our techniques cannot handle. On the other hand, their dimension dependence is at least $d^{3/4}$, which does not provide a faster warm start and in particular does not reach the $d^{1/4}$ scaling for MHMC that we obtain in Theorem~\ref{thm:main-warm}.

        \section{Background}\label{sec:prelim}

\subsection{(Metropolized) Hamiltonian Monte Carlo}\label{sec:hmc_bg}

Hamiltonian Monte Carlo (HMC) refers to a family of methods based on Hamilton's equations of motion. Given a target distribution $\pi \propto \exp(-V)$, we interpret the function $V : \R^d\to\R$ as the ``potential energy'', and we form the Hamiltonian which is the sum of the potential energy and the kinetic energy (corresponding to a particle of unit mass):
\begin{align*}
    H : \R^d\times \R^d\to\R\,, \qquad
    H(x, p) \deq V(x) + \frac{1}{2}\,\norm p^2\,.
\end{align*}
The distribution over phase space $\R^d\times \R^d$ is then taken to have density $\bs\pi \propto \exp(-H)$.
Note that the $x$-marginal of $\bs\pi$ is $\pi$, so that drawing an approximate sample from $\bs\pi$ and keeping only the $x$-component leads to an approximate sample from $\pi$.

The Hamiltonian ODE associated with $H$ is the following first-order ODE system:
\begin{align}\label{eq:ham_ode}
    \dot x_t = p_t\,, \qquad \dot p_t = -\nabla V(x_t)\,.
\end{align}
The Hamiltonian dynamics are volume-preserving and preserve $H$, so they leave $\bs\pi$ invariant; however, they are not ergodic.
Therefore, the standard algorithmic approach is to periodically refresh the momentum.
Given a time horizon $T > 0$, let $\Phi_T : \R^d\times\R^d \to\R^d\times \R^d$ denote the flow map corresponding to the ODE system~\eqref{eq:ham_ode}. That is, let $\Phi_T(x,p)$ denote the solution to~\eqref{eq:ham_ode} at time $T$, initialized at $(x_0, p_0) = (x,p)$.
Then, \emph{ideal HMC} iterates the following steps. For $n=0,1,2,\dotsc$,
\begin{itemize}
    \item \emph{Refresh} the momentum: draw $P_n \sim \cN(0, I)$.
    \item \emph{Integrate} the Hamiltonian ODE: set $(X_{n+1}, P_{n+1}') \deq \Phi_T(X_n, P_n)$.
\end{itemize}
Ideal HMC is known to converge in law to $\bs\pi$ (in its first argument) at an exponential rate, provided that $V$ is strongly convex and smooth~\cite{chenvempala2019hmc}.

In practice, the ODE system~\eqref{eq:ham_ode} cannot be exactly integrated, so we must use a numerical approximation.
Here we consider the \emph{leapfrog integrator} (sometimes called the \emph{St\"ormer--Verlet} integrator).
Given a step size $h > 0$ with $K \deq T/h \in \N$, we define $\widehat\Phi_{T,h} : \R^d\times\R^d\to\R^d \times\R^d$ via
\begin{align}\label{eq:leapfrog}
    \widehat\Phi_{T,h}(x,p) \deq (\widehat x_{Kh}, \widehat p_{Kh})\,, \qquad \begin{cases}
        (\widehat x_0,\widehat p_0) \deq (x, p)\,, \\
        \widehat p_{(k + 1/2)h} \deq \widehat p_{kh} - (h/2)\,\nabla V(\widehat x_{kh})\,, &({\rm B}) \\
        \widehat x_{(k+1)h} \deq \widehat x_{kh} + h\widehat p_{(k+1/2)h}\,, &({\rm A}) \\
        \widehat p_{(k+1)h} \deq \widehat p_{(k+1/2)h} - (h/2)\,\nabla V(\widehat x_{(k+1)h})\,. &({\rm B})
    \end{cases}
\end{align}

The use of the leapfrog integrator introduces bias; that is, replacing $\Phi_T$ with $\widehat\Phi_{T,h}$ in ideal HMC distorts the invariant measure.
One can correct for the bias by adding a Metropolis--Hastings filter, leading to the \emph{Metropolized Hamiltonian Monte Carlo} (\emph{MHMC}) algorithm.
Since the leapfrog integrator is also volume-preserving (indeed it is symplectic), it leads to the following simple form for the correction step. For $n=0,1,2,\dotsc$:
\begin{itemize}
    \item \emph{Refresh} the momentum: draw $P_n \sim \cN(0, I)$.
    \item \emph{Integrate} the Hamiltonian ODE\@: set $(X_{n+1}', P_{n+1}') \deq \widehat\Phi_{T,h}(X_n, P_n)$.
    \item \emph{Accept} the step with probability $1 \wedge \exp\{H(X_n, P_n) - H(X_{n+1}', P_{n+1}')\}$.
    If the step is accepted, we set $X_{n+1} \deq X_{n+1}'$.
    Otherwise, we set $X_{n+1} \deq X_n$.
\end{itemize}
As is standard in theoretical analysis, we use the $\frac{1}{2}$-lazy version of this chain, in which the acceptance probability is replaced by $\frac{1}{2} \wedge \frac{1}{2}\exp\{H(X_n, P_n) - H(X_{n+1}', P_{n+1}')\}$.

As displayed in~\eqref{eq:leapfrog}, the position and momentum updates are often denoted ``(A)'' and ``(B)'' respectively in the literature.
When HMC is run without applying the Metropolis--Hastings filter, it is known as \emph{unadjusted HMC}, and in that context a common alternative is to consider partial, rather than full, momentum refreshment.
Define the ``(O)'' step to be the following partial momentum refreshment $p \mapsto p^{\msf O}$ which preserves the standard Gaussian:
\begin{align}\tag{O}\label{eq:O}
    p^{\msf O} \sim \cN(e^{-\gamma h/2}\, p\,, \; (1-e^{-\gamma h})\,I)\,.
\end{align}
Here, $\gamma > 0$ denotes the \emph{friction} parameter, and we suppress its dependence in the notation.
Different ways of composing together the (A), (B), and (O) steps lead to various \emph{splitting schemes} for unadjusted HMC\@.
For example, in this paper we consider (a variant of) the popular OBABO scheme, consisting of iterating the O, B, A, B, O steps in that order, although other palindromic combinations are also possible.

\subsection{Divergences between probability measures}\label{ssec:prelim-divergences}

In this section, we briefly review divergences between probability measures which are used in this work. For further background on R\'enyi divergences see, e.g.,~\cite{van2014renyi}.

\begin{defin}[R\'enyi divergence]\label{def:renyi}
    The R\'enyi divergence of order $q > 1$ between two probability measures $\mu, \nu$ is defined as
    \begin{align*}
        \Renyi_q(\mu \mmid \nu) \deq \frac{1}{q-1} \log \int \Bigl(\frac{\D \mu}{\D \nu}\Bigr)^q \, \D \nu\,,
    \end{align*}
    if $\mu \ll \nu$, and $+ \infty$ otherwise.
    When $q=1$, we take $\Ren_q = \KL$; and when $q=2$, then $\Ren_2 = \log(1+\chi^2)$, where $\chi^2$ is the chi-squared divergence.
\end{defin}

\begin{prop}[Properties of R\'enyi divergences]\label{prop:renyi}\mbox{}
\begin{enumerate}
    \item (Monotonicity) For $q \leq q'$ and any $\mu, \nu \in \mc P(\Omega)$,
    \begin{align*}
        \Renyi_q(\mu \mmid \nu) \leq \Renyi_{q'}(\mu \mmid \nu)\,.
    \end{align*}
    \item (Data-processing inequality) For any $q \ge 1$, any Markov kernel $P$, and any $\mu,\nu \in \mc P(\Omega)$,
    \begin{align*}
        \Renyi_q(\mu P\mmid \nu P) \le \Renyi_q(\mu \mmid \nu)\,.
    \end{align*}
    \item (Weak triangle inequality) For any $q > 1$, the following holds for any triple $\mu, \nu, \pi \in \mc P(\Omega)$,
    \begin{align*}
        \Renyi_q(\mu \mmid \pi) \leq \frac{q-\nfrac{1}{2}}{q-1}\, \Renyi_{2q}(\mu \mmid \nu) + \Renyi_{2q-1}(\nu \mmid \pi)\,.
    \end{align*}
\end{enumerate}
\end{prop}

Finally, we make use of the following principle, introduced in~\cite{scr1}.

\begin{theorem}[Shifted composition rule]\label{thm:shifted_chain_rule}
    Let $\msf X$, $\msf X'$, $\msf Y$ be three random variables jointly defined on a standard probability space $\Omega$. Let $\mb P$, $\mb Q$ be two probability measures over $\Omega$, with superscripts denoting the laws of random variables under these measures. Then, for any $q \ge 1$,
	\begin{align*}
		\Renyi_q(\mb P^{\msf Y} \mmid \mb Q^{\msf Y})
		&\le \Renyi_q(\mb P^{\msf X'} \mmid \mb Q^{\msf X}) + \inf_{\gamma \in \Coup(\mb P^{\msf X}, \mb P^{\msf X'})}\; \operatorname*{\gamma-ess\,sup}_{(x,x') \in\R^d\times \R^d} \Renyi_q(\mb P^{\msf Y\mid \msf X=x} \mmid \mb Q^{\msf Y\mid \msf X=x'})\,.
	\end{align*}
\end{theorem}

\subsection{Functional inequalities}\label{ssec:prelim-functional}

Throughout we operate under the following assumption.
\begin{assumption}\label{as:regularity}
    We assume the measure $\pi \propto \exp(-V)$ is $\beta$-log-smooth. That is, $V$ is twice continuously differentiable and $-\beta I \preceq \nabla^2 V(x) \preceq \beta I$ for all $x \in \R^d$.
\end{assumption}
 Before proceeding, we first define some useful isoperimetric inequalities.
\begin{defin}
    A measure $\pi$ satisfies a log-Sobolev inequality (LSI) with constant $C_{\msf{LSI}}$ if for all smooth $f: \R^d \to \R_{>0}$,
    \begin{align}\label{eq:lsi}\tag{$\msf{LSI}$}
        \ent_\pi f \deq \E_\pi \Bigl[f \log \frac{f}{\E_\pi f}\Bigr] \leq \frac{C_{\msf{LSI}}}{2}\, \E_\pi [\norm{\nabla f}^2]\,. 
    \end{align}
    A measure $\pi$ satisfies a weaker inequality known as a Poincar\'e inequality (PI) with constant $C_{\msf{PI}}$ if for all smooth $f: \R^d \to \R$,
    \begin{align}\label{eq:pi}\tag{$\msf{PI}$}
        \var_\pi f \deq \E_\pi [\abs{f-\E_\pi f}^2 ] \leq {C_{\msf{PI}}}\, \E_\pi [\norm{\nabla f}^2]\,. 
    \end{align}
\end{defin}
The following is a sequence of isoperimetric assumptions, from strongest to weakest, under which we will state our rates. A standard chain of implications shows that Assumption~\ref{ass:iso:slc} implies Assumption~\ref{ass:iso:lsi} implies Assumption~\ref{ass:iso:pi}. See e.g.,~\cite[Chapter 2]{chewibook} for further background. 
\begin{assumption}\label{ass:iso}
Assume, for $\alpha > 0$,
\begin{enumerate}[label=(\alph*), ref=\theassumption(\alph*), leftmargin=*]
  \item\label{ass:iso:slc} $\pi$ is $\alpha$-strongly log-concave (i.e., $\nabla^2 V \succeq \alpha I$).
  
  \item\label{ass:iso:lsi} $\pi$ satisfies~\eqref{eq:lsi} with constant $1/\alpha$.
  \item\label{ass:iso:pi} $\pi$ satisfies~\eqref{eq:pi} with constant $1/\alpha$.
\end{enumerate}
\end{assumption}
        \section{Main result: high-accuracy sampling in \texorpdfstring{$d^{1/4}$}{d\^{}(1/4)} steps}\label{sec:results}

\subsection{Formal statement of result}

We recall the following Frobenius-Lipschitz Hessian assumption, which was also used in~\cite{chengatmiry23}.

\begin{assumption}[Frobenius-Lipschitz Hessian]\label{as:higher-reg}
    Assume that $V$ is three times continuously differentiable and that $\nabla^2 V$ is Lipschitz continuous in the Frobenius norm.
    That is, there exists $\beta_\eff > 0$ such that for all $x,y\in\R^d$, $\norm{\nabla^2 V(x) - \nabla^2 V(y)}_{\rm F} \le \beta_\eff\,\norm{x-y}$.
\end{assumption}

This assumption is equivalent to a tensor norm bound on $\nabla^3 V$, namely $\norm{\nabla^3 V}_{\{1,2\},\{3\}} \le \beta_\eff$.
Recall that for a $3$-tensor $T$, we define its $\{1,2\}, \{3\}$ norm as
\begin{align*}
    \norm{T}_{\{1,2\}, \{3\}} \deq \sup\Bigl\{\sum_{i,j,k=1}^d T_{i,j,k}\, x_{i,j}\, y_k\; : \; \norm{x}_{\rm F} \leq 1,\, \norm{y} \leq 1 \Bigr\}\,.
\end{align*}
The specific consequence of Assumption~\ref{as:higher-reg} that we need is the following probabilistic tail bound for an object typically called a Gaussian chaos, whose proof we defer to Appendix~\ref{app:chaos}.

\begin{lemma}[Gaussian chaos bound]\label{lem:chaos}
    Let $T \in \R^{d\times d\times d}$ be a symmetric $3$-tensor, and let $\bar\beta \deq \norm T_{\{1,2\},\{3\}}$.
    Then, for $\xi\sim \cN(0, I)$ and for all $\delta \in (0,1/2)$, with probability at least $1-\delta$,
    \begin{align*}
        \norm{T[\xi,\xi]} \lesssim \bar\beta\,\Bigl(d^{1/2} + \log \frac{1}{\delta}\Bigr)\,.
    \end{align*}
\end{lemma}

\begin{remark}
    The conclusion of Lemma~\ref{lem:chaos}, with $T$ replaced by $\nabla^3 V(x)$, was used as an assumption in~\cite{zhang2025analysis}. By showing that the desired tail bound is a consequence of Assumption~\ref{as:higher-reg}, we resolve an open question in that paper. 
\end{remark}

With Assumption~\ref{as:higher-reg} in hand, we now state our main result. For simplicity, we make the standard assumption that $V$ is translated to have minimizer at $0$, since otherwise one can accurately compute the minimizer of $V$ in an efficient pre-processing step (e.g., using gradient descent).

\begin{theorem}[Faster high-accuracy sampling]\label{thm:main-slc}
    Under Assumptions~\ref{as:regularity},~\ref{ass:iso:slc}, and~\ref{as:higher-reg}, and assuming that $V$ is minimized at $0$, we can obtain a sample $X$ such that $\msf d(\operatorname{law}(X) \mmid \pi) \leq \varepsilon$, with
    \begin{align*}
        \msf d \in \{\msf{TV}, \sqrt{\msf{KL}}, \sqrt{\chi^2}, \sqrt{\alpha}\, W_2 \}\,,
    \end{align*}
    using at most $N$ oracle queries, where
    \begin{align*}
        N = \widetilde \Theta\Bigl((\kappa^{3/2} + \kappa^{1/2} \kappa_\eff^{1/2} + \kappa_\eff^{2/3})\, d^{1/4}  \log^{2} \frac{1}{\varepsilon}\Bigr)\,.
    \end{align*}
    Here, $\kappa \deq \beta/\alpha$ and $\kappa_\eff \deq \beta_\eff/\alpha^{3/2}$.
\end{theorem}

Previously,~\cite{chengatmiry23} established a high-accuracy sampling guarantee with dimension dependence $d^{1/4}$ for the Metropolized HMC algorithm from a warm start under the same assumptions. Our main contribution is to demonstrate that a warm start can be obtained in $d^{1/4}$ steps using a variant of HMC without the use of the Metropolis--Hastings filter. This result is stated below.

\begin{theorem}[Algorithmic warm start]\label{thm:main-warm}
    Under Assumptions~\ref{as:regularity},~\ref{ass:iso:slc}, and~\ref{as:higher-reg}, and assuming that $V$ is minimized at $0$,
    for all $q = O(1)$, the OBABCO algorithm (see \S\ref{ssec:obabco}) with appropriate parameters outputs a sample $X$ satisfying $\Renyi_q(\operatorname{law}(X) \mmid \bs\pi) \le 1$ using
    \begin{align*}
        N = \widetilde{O}\Bigl(\bigl(\kappa^{3/2} + \kappa^{1/2} \kapH^{1/2}\bigr)\, d^{1/4}\Bigr)\qquad\text{gradient evaluations}\,,
    \end{align*}
    where $\kappa \deq \beta/\alpha$ and $\kapH \deq \beta_\eff/\alpha^{3/2}$.
\end{theorem}

\begin{remark}
    \cite{zhang2025analysis} recently showed that a different algorithm (a certain discretization of the underdamped Langevin dynamics) achieves $\KL(\operatorname{law}(X) \mmid \bs\pi) \leq 1$ in $\widetilde O((\kappa^{5/4} + \kappa^{1/2} \kappa_\eff^{1/2})\,d^{1/4})$ iterations.
    However, such KL guarantees do not provide a warm start, as $\Renyi_q$ guarantees are needed for $q > 1$.
    Moreover, the techniques of~\cite{zhang2025analysis} are not easily adapted to R\'enyi guarantees because they rely on the shifted-composition local error framework in~\cite{scr4}, which is currently only known for KL analysis.
    This paper therefore takes a different approach, in terms of both the algorithm and analysis.
\end{remark}

Finally, we note that the condition number dependence in Theorem~\ref{thm:main-slc} can be improved via the proximal sampler, as described below (formal statement in Corollary~\ref{cor:prox-iso}).

\subsection{Extension to other settings via the proximal sampler}

We first pause to provide some background on the proximal sampler~\cite{leeshentian2021rgo, chen2022improved}. For step size $h > 0$, augment the target distribution $\pi(x) \propto \exp(-V(x))$ to a joint distribution $\hat{\bs \pi}$ on $\R^d \times \R^d$ with density
\begin{align*}
    \hat{\bs \pi}(x, y) \propto \exp\Bigl(- V(x) - \frac{1}{2h}\, \norm{x-y}^2 \Bigr)\,.
\end{align*}
The proximal sampler algorithm alternates sampling between two conditional distributions:
\begin{itemize}
    \item \textit{Forward step:} sample $Y_k \mid X_k \sim \pi^{Y \mid X}(\cdot \mid X_k) = \mc N(X_k, hI)$.

    \item \textit{Backward step:} sample $X_{k+1} \mid Y_k \sim \pi^{X \mid Y}(\cdot \mid Y_k)$.
\end{itemize}
While the forward step is trivial to implement, the backward step is not. It is the computation bottleneck of the algorithm and is known as the \emph{restricted Gaussian oracle} (RGO).  

A key property of the proximal sampler is that it enables improving high-accuracy sampling algorithms such as ours in Theorem~\ref{thm:main-slc}. Algorithmically, the idea is to use the high-accuracy sampling algorithm to implement the RGO step in the proximal sampler. This reduction leads to two key benefits. First, it improves the condition number dependence of the high-accuracy sampler from $\kappa^p$ to $\kappa$, for any $p \geq 1$. Second, it extends high-accuracy sampling guarantees under strong log-concavity to more general isoperimetric assumptions (log-Sobolev or Poincar\'e inequalities) which capture a wide class of non-log-concave distributions. The basic idea behind both these improvements is that---regardless of the value of $\kappa$ and whether $\pi$ satisfies isoperimetry rather than strong log-concavity---the RGO is strongly log-concave with condition number $\kappa \asymp 1$ due to the choice of step size $h = \nfrac{1}{2\beta}$. Indeed, this choice of $h$ ensures that the RGO potential has strong convexity and smoothness parameters both of order $\beta$ (meanwhile the parameter $\beta_\eff$ remains unchanged, so $\kappa_\eff$ becomes $\beta_\eff/\beta^{3/2}$). 
We refer to~\cite[Appendix D]{AltChe24Warm} for a detailed discussion. 

Applying this proximal sampler reduction to our result in Theorem~\ref{thm:main-slc} yields the following corollary. Below we assume isoperimetric inequalities~\eqref{eq:lsi} or~\eqref{eq:pi} instead of strong convexity (see \S\ref{ssec:prelim-functional} for definitions), and we use the notation $\kappa \deq \beta/\alpha$ and $\kappa_\eff \deq \beta/\alpha^{3/2}$ where $1/\alpha$ is the associated isoperimetric constant. For simplicity, we assume access to a \emph{proximal oracle} for $V$ with step size $h > 0$, which computes $\argmin_{y\in\R^d} \{V(y) + \frac{1}{2h}\, \norm{x-y}^2\}$ given a query $x \in \R^d$. This assumption is not strictly necessary and can be replaced with a more laborious argument based on the computation of an approximate minimizer, along the lines of~\cite[Appendix D]{AltChe24Warm}. We omit these details for simplicity, as it does not dominate the computational cost. 

\begin{cor}[Isoperimetric targets via the proximal sampler]\label{cor:prox-iso}
    Suppose Assumptions~\ref{as:regularity},~\ref{ass:iso:lsi}, and~\ref{as:higher-reg} hold. Then, assuming we have access to a proximal oracle for $V$ with step size $\nfrac 1{2\beta}$, we can obtain a sample $X$ such that $\msf d(\law(X) \mmid \pi) \leq \varepsilon$ for     \begin{align*}
        \msf d \in \{\msf{TV}, \sqrt{\msf{KL}}, \sqrt{\chi^2}, \sqrt{\alpha}\, W_2 \}\,,
    \end{align*}
    using at most $N$ oracle queries, where
    \begin{align}\label{eq:lsi-complexity}
        N = \widetilde \Theta\Bigl((\kappa + \kappa_\eff^{2/3})\, d^{1/4} \log^3 \frac{\Renyi_2(\mu_0 \mmid \pi)}{\varepsilon^2}\Bigr)\,.
    \end{align} 
    If instead we assume~\ref{ass:iso:pi} instead of~\ref{ass:iso:lsi}, then the bound becomes
    \begin{align*}
        N = \widetilde \Theta\Bigl((\kappa + \kappa_\eff^{2/3})\,  d^{1/4} \log^3 \frac{\chi^2(\mu_0 \mmid \pi)}{\varepsilon^2}\Bigr)\,.
    \end{align*} 
\end{cor}
\begin{remark}
    Via the Bakry--\'Emery criterion, Assumption~\ref{ass:iso:slc} implies Assumption~\ref{ass:iso:lsi} with the same constant. Thus~\eqref{eq:lsi-complexity} improves upon the condition number dependence in Theorem~\ref{thm:main-slc}. 
\end{remark}

Finally, if instead of a Poincar\'e inequality we assume weak convexity of the potential (Assumption~\ref{ass:iso:slc} with $\alpha = 0$), then we can obtain a bound reminiscent of~\cite[Theorem 5.6]{AltChe24Warm} but with an additional improvement in the leading dimensional factor.

\begin{cor}[Weakly convex targets via the proximal sampler]\label{cor:prox-wc}
    Suppose Assumptions~\ref{as:regularity},~\ref{ass:iso:slc}, and~\ref{as:higher-reg} hold with $\alpha = 0$. Then, assuming we have access to a proximal oracle for $V$ with step size $\nfrac 1 {2\beta}$, we can obtain a sample $X$ such that $\msf d(\law(X) \mmid \pi) \leq \varepsilon$ for     \begin{align*}
        \msf d \in \{\msf{TV}, \sqrt{\msf{KL}}\}\,,
    \end{align*}
    using at most $N$ oracle queries, where
    \begin{align}\label{eq:wlc-complexity}
        N = \widetilde \Theta\Bigl((\beta + \beta_\eff^{2/3})\, \frac{d^{1/4}\,W_2^2(\mu_0, \pi)}{\varepsilon^2}\Bigr)\,.
    \end{align}
\end{cor}

\subsection{Proof of high-accuracy sampling guarantees using the warm start}

In this section we show how the claimed high-accuracy sampling guarantees (Theorem~\ref{thm:main-slc}, Corollary~\ref{cor:prox-iso}, and Corollary~\ref{cor:prox-wc}) follow from our algorithmic warm start result (Theorem~\ref{thm:main-warm}). The rest of the paper is then devoted to proving the latter result.

\subsubsection{Proof of Theorem~\ref{thm:main-slc}}

It is helpful to recall two preliminary facts about MHMC. We begin with the main result of~\cite{chengatmiry23} which bounds the complexity of MHMC from a warm start.\footnote{The statement of Lemma~\ref{lem:chen-gatmiry} here differs from~\cite{chengatmiry23} in two minor ways. First, the result in~\cite{chengatmiry23} is written in terms of the squared Cheeger constant of $\pi$, but that is bounded by $\alpha$ under strong log-concavity of $\pi$~\cite{milman2009role}; we state this simplified version using $\alpha$ as it suffices our purposes. Second, the result in~\cite{chengatmiry23} is stated for a $\Renyi_\infty$ warm start, but that is needed only in~\cite[Lemma 1]{chengatmiry23} to inform the choice of $s$ in the definition of $s$-conductance, and the argument of~\cite[\S 7]{chewibook} or~\cite[Theorem D.3]{AltChe24Warm} allows this to be replaced by a $\chi^2$ warm start.}

\begin{lemma}[Complexity of MHMC under a warm start]\label{lem:chen-gatmiry}
    Let $\pi$ satisfy Assumptions~\ref{as:regularity},~\ref{ass:iso:slc}, and~\ref{as:higher-reg}. Given a warm start $\bar{\mu}_0$ with $\chi^2(\bar{\mu}_0 \mmid \pi) = O(1)$, MHMC with step size $h$ and $N_1$ leapfrog steps obtains $\msf{TV}(\bar{\mu}_0 (\bs P_{N_1}^{\msf{MHMC}})^{N_2} \mmid \pi) \leq \varepsilon$ if we choose parameters
    \begin{align*}
        h^{-1} = \widetilde \Theta\Bigl((\beta^{1/2} + \beta_\eff^{1/3})\, d^{1/4} \log \frac{1}{\varepsilon}\Bigr)\,, \quad N_1 = \widetilde \Theta\Bigl( \frac{1}{(\beta^{1/2} + \beta_\eff^{1/3})\, h}\Bigr)\,, \quad N_2 = \widetilde \Theta\Bigl( (\kappa + \kappa_\eff^{2/3}) \log \frac{1}{\varepsilon}\Bigr)\,.
    \end{align*}
    The total oracle complexity is
    \begin{align*}
        N_1 \cdot N_2 \lesssim (\kappa + \kappa_\eff^{2/3})\, d^{1/4} \log^{2} \frac{1}{\varepsilon}\,.
    \end{align*}
\end{lemma}

We also recall the following lemma from~\cite[Lemma 5.2]{AltChe24Warm} which enables boosting a mixing guarantee from TV to $\chi^2$, when given a R\'enyi warm start.
\begin{lemma}[{Boosting $\msf{TV}$-to-$\chi^2$ mixing} given R\'enyi warm start]\label{lem:boosting}
    Let $P$ be any Markov chain with stationary measure $\pi$. If $\mu_N =\mu_0 P^N$, then
    \begin{align*}
        \chi^2(\mu_N \mmid \pi) \leq \sqrt{\msf{TV}(\mu_N, \pi) \cdot \bigl(\exp(2\Renyi_3(\mu_0 \mmid \pi)) + 1 \bigr)}\,.
    \end{align*}
\end{lemma}

\begin{proof}[Proof of Theorem~\ref{thm:main-slc}]
    First, apply Theorem~\ref{thm:main-warm} to obtain a warm start $\bar{\mu}_0$ satisfying $\Ren_3(\bar{\mu}_0 \mmid \pi) \leq 1$. This step has complexity
    \begin{align*}
        N_{\msf{warm}} = \widetilde O \Bigl(\bigl(\kappa^{3/2} + \kappa^{1/2} \kappa_{\eff}^{1/2}\bigr)\, d^{1/4} \Bigr)\,.
    \end{align*}
    Second, run MHMC from this warm start to obtain a sample from a distribution $\hat{\mu}$ satisfying $\TV(\hat{\mu},\pi) \leq \eps^4$.
    Since the R\'enyi warm start $\Ren_3(\bar{\mu}_0 \mmid \pi) \leq 1$ implies a $\chi^2$ warm start $\chi^2(\bar{\mu}_0 \mmid \pi) \lesssim 1$ by the monotonicity in Proposition~\ref{prop:renyi}, we can apply Lemma~\ref{lem:chen-gatmiry}. Thus this step has complexity
    \begin{align*}
        N_{\msf{MHMC}} = \widetilde O \Bigl(\bigl(\kappa + \kappa_\eff^{2/3}\bigr)\, d^{1/4} \log^2 \frac{1}{\varepsilon} \Bigr)\,.
    \end{align*} 
    The sum $N_{\msf{warm}} + N_{\msf{MHMC}}$ gives the desired complexity.

    Finally, we upgrade the mixing guarantee from $\TV$ to stronger metrics. By Lemma~\ref{lem:boosting}, $\chi^2(\hat\mu \mmid \pi) \lesssim \sqrt{\TV(\hat{\mu},\pi)} \leq \eps^2$. This implies the other guarantees by standard comparison inequalities: it implies the KL guarantee by monotonicity of the R\'enyi divergence as $q \searrow 1$ (Proposition~\ref{prop:renyi}) and the Wasserstein guarantee by Talagrand's $T_2$ inequality (see, e.g.,~\cite[Chapter 1]{chewibook}).
\end{proof}

\subsubsection{Proof of Corollaries~\ref{cor:prox-iso} and~\ref{cor:prox-wc}}

We make use of the following standard lemma about the R\'enyi divergence at a Gaussian ``cold start'' initialization (for a proof see e.g.,~\cite[\S 1.5]{chewibook}). The assumption $\nabla V(0) = 0$ is by a translation, see also the discussion preceding Theorem~\ref{thm:main-slc}.

\begin{lemma}[Initialization bound]\label{lem:initialization}
    Suppose $\pi \propto \exp(-V)$ satisfies $0 \preceq \alpha I \preceq \nabla^2 V \preceq \beta I$ and $\nabla V(0) = 0$. Then the Gaussian initialization $\widehat \mu_0 = \mc N(0, \beta^{-1} I)$ satisfies
    \begin{align*}
        \Renyi_\infty(\widehat \mu_0 \mmid \pi) \leq \frac{d}{2} \log \kappa\,.
    \end{align*}
\end{lemma}

\begin{proof}[Proof of Corollary~\ref{cor:prox-iso}]
    We adapt the proof in~\cite[Appendix D]{AltChe24Warm} by replacing the high-accuracy sampler there (based on MALA) with the improved high-accuracy sampler developed here in Theorem~\ref{thm:main-slc} (based on HMC). Below we reserve $h \asymp \nfrac{1}{\beta}$ for the step size of the proximal sampler. 

    \paragraph{LSI case.} By the proof of~\cite[Appendix D.4]{AltChe24Warm}, it suffices to run the proximal sampler with step size $h = \nfrac{1}{2\beta}$ for 
    \begin{align*}
        N_{\operatorname{prox}} \asymp \kappa \log \frac{\Renyi_2(\mu_0^X \mmid \pi^X)}{\varepsilon^2}
    \end{align*}
    iterations and perform each RGO sampling step approximately to accuracy 
    \begin{align*}
        \varepsilon^2_{\operatorname{RGO}} \asymp \frac{\varepsilon^2}{\kappa}\,.
    \end{align*}
    The bottleneck of this proximal sampler algorithm is the implementation of the approximate RGO. We do this using the improved high-accuracy sampling algorithm in Theorem~\ref{thm:main-slc}. Note that because of our step size choice, this RGO distribution $\pi^{X \mid Y = Y_n}$ is strongly log-concave with log-concavity bounded below by $\Omega(\beta)$ and condition number $\Theta(1)$. Moreover by Lemma~\ref{lem:initialization}, given access to the proximal oracle for $V$ (which allows us to locate the mode of the RGO), the initial R\'enyi divergence is $O(d)$. Thus Theorem~\ref{thm:main-slc} shows that the complexity of each approximate RGO step is
    \begin{align}\label{eq:inner-complexity}
        N = \widetilde \Theta \Bigl(\Bigl(1 + \frac{\beta_\eff^{1/2}}{\beta^{3/4}} + \frac{\beta_{\eff}^{2/3}}{\beta} \Bigr)\,  d^{1/4} \log^2 \frac{1}{\varepsilon_{\operatorname{RGO}}}\Bigr)\,.
    \end{align}
    To conclude, multiply $N \cdot N_{\operatorname{prox}}$, and simplify by noting that $\kappa^{1/4} \kappa_\eff^{1/2} \lesssim \kappa + \kappa_\eff^{2/3}$ by Young's inequality.
    
    \paragraph{Poincar\'e case.} 
     By the proof of~\cite[Appendix D.4]{AltChe24Warm}, it suffices to run the proximal sampler with the same step size $h$ and adjusted parameters
    \begin{align*}
        N_{\operatorname{prox}} \asymp \kappa \log \frac{\chi^2(\mu_0^X \mmid \pi^X)}{\varepsilon^2}\,, \qquad \varepsilon_{\operatorname{RGO}}^2 \asymp \frac{\varepsilon^2}{\kappa}\,.
    \end{align*}
    The rest of the argument is identical.
\end{proof}

\begin{proof}[Proof of Corollary~\ref{cor:prox-wc}]
    In this weakly convex setting, the proof of~\cite[Theorem 5.6]{AltChe24Warm} shows that it suffices to run the proximal sampler with the same step size $h$ and adjusted parameters
    \begin{align*}
        N_{\operatorname{prox}} \asymp \frac{\beta W_2^2(\mu_0^X, \pi^X)}{\varepsilon^2}\,, \qquad \varepsilon^2_{\operatorname{RGO}} \lesssim \min \Bigl\{\varepsilon^4, \frac{\varepsilon^4}{\beta^2 W_2^2(\mu_0^X, \pi)}, \frac{\varepsilon^2}{\beta^{3/2} W_2^{3/2}(\mu_0^X, \pi)}, \frac{1}{\beta W_2(\mu_0^X, \pi)}\Bigr\}\,.
    \end{align*}
    The rest of the argument is identical.
\end{proof}
        \section{Technical overview and road map}\label{sec:overview}

In this section, we give a technical overview of the algorithmic warm start result (Theorem~\ref{thm:main-warm}) and a road map for the remainder of this paper. Our algorithmic warm start is produced via a variant of the OBABO scheme for unadjusted HMC (see \S\ref{sec:hmc_bg} for background), and our goal is to analyze its convergence in R\'enyi divergence.
Toward that end, our strategy is to first introduce an idealized version of these dynamics, called the OHO dynamics (described in \S\ref{ssec:oho_overview}).
Then we control the discretization error between OBABO and OHO (described in \S\ref{ssec:obabco_overview}).

\subsection{OHO dynamics (Section~\ref{sec:oho})}\label{ssec:oho_overview}

A key conceptual insight in our work is to introduce a process which we call the OHO dynamics.
Here, the ``O'' steps refer to the partial momentum refreshment steps in equation~\eqref{eq:O}, and the ``H'' step refers to integration of the ideal Hamiltonian flow for time $h$:
\begin{align}\tag{H}\label{eq:H}
    (x, p) \mapsto (x^{\msf H}, p^{\msf H}) \deq \Phi_h(x, p)\,.
\end{align}
Since the H and O steps individually leave invariant $\bs\pi$, it follows that $\bs \pi$ is also the invariant measure of the OHO dynamics.
As $h\searrow 0$ for fixed $T > 0$, the iterated process $(\text{OHO})^{\lfloor T/h\rfloor}$ is expected to converge, in a suitable sense, to the underdamped Langevin SDE run for time $T$ (see~\cite[\S 5.3]{chewibook}).

Our first goal is to establish the exponential decay of the R\'enyi divergence along these idealized dynamics. Since these dynamics are expected to behave similarly to the underdamped Langevin diffusion, we first briefly review the literature for the latter.
The importance of exponential convergence results in R\'enyi for algorithmic warm starts was identified in~\cite{AltChe24Warm}, in which this property was termed ``hyperequilibration'' (in analogy to hypercontractivity).
We remark that existing PDE techniques for establishing hypocoercivity are currently unable to show hyperequilibration; see the discussion therein for details.\footnote{In brief, hypocoercive techniques can show the exponential decay of the chi-squared divergence $\chi^2$, which is related to the R\'enyi divergence of order $2$ via $\chi^2 = \exp(\Ren_2) -1$, but this does not suffice to establish nearly dimension-free convergence, which is needed for our results. Indeed, since the initial R\'enyi divergence to $\bs\pi$ typically scales as $\widetilde\Theta(d)$, a result of the form $\chi^2(\bs\mu_t\mmid\bs \pi) \le c_0 \exp(-c_1 t)\,\chi^2(\bs\mu_0 \mmid \bs \pi)$ requires a mixing time of order $t \gg \log \chi^2(\bs\mu_0 \mmid \bs \pi) = \widetilde \Theta(d)$, whereas hyperequilibration requires the mixing time to only be of order $O(\log d)$.}
In that work, hyperequilibration for the underdamped Langevin diffusion was bypassed by directly considering the discretization. Later, it was realized that hyperequilibration follows from an inequality of the following type, which we call a \emph{Harnack inequality} (see~\cite{scr1} for further discussion):
\begin{align}\label{eq:uld_harnack}
    \Ren_q(\delta_{x,p} \bs P_T\mmid \delta_{\bar x,\bar p} \bs P_T) \le C_{\msx, T}\,\norm{x-\bar x}^2 +C_{\msp,T}\, \norm{p-\bar p}^2 \qquad\text{for all}~x,\bar x,p,\bar p \in\R^d\,,
\end{align}
for suitable constants $C_{\msx,T}, C_{\msp,T} > 0$, where $\bs P_T$ denotes the underdamped Langevin semigroup at time $T$.
Such an inequality was first established in~\cite{GuiWan12DegenBismut}, and it is readily combined with exponential decay of the $W_\infty$ distance to $\bs\pi$ along the underdamped Langevin dynamics under strong log-concavity to produce exponentially decaying R\'enyi bounds.
Recently, in our prior work~\cite{scr4}, we gave a direct proof of an inequality of the form~\eqref{eq:uld_harnack} with a constant $C_T$ decaying exponentially as $T\to\infty$, again under strong log-concavity.

To extend these results to OHO\@, our first order of business is to establish $W_\infty$ contraction under strong log-concavity.
This is done in Lemma~\ref{lem:oho_contraction} below and is inspired by the coordinate change used in the coupling analysis of the underdamped Langevin dynamics, but with suitable modifications.

With this result in hand, we turn to establishing the Harnack inequality~\eqref{eq:uld_harnack}, with $\bs P_T$ replaced by the Markov kernel corresponding to the OHO dynamics.
Since OHO no longer has the structure of an SDE\@, this precludes a straightforward extension of the proof for the underdamped Langevin dynamics.
Instead, we leverage the idea of a one-shot coupling from~\cite{roberts2002one, bou2023mixing, monmarche2024entropic, chak2025reflection, gouraud2025hmc} to establish a R\'enyi bound that holds for a single iteration.
We then use the framework of shifted composition, recently developed in the works~\cite{scr1, scr2, scr3, scr4}, to upgrade this ``short-time'' result to a ``long-time'' Harnack inequality of the form~\eqref{eq:uld_harnack}; see Theorem~\ref{thm:harnack}. Appealingly, our Harnack inequality for the OHO dynamics essentially matches the one for the underdamped Langevin dynamics (in fact, exactly so as the discretization parameter $h\searrow 0$),

\subsection{OBABCO discretization (Section~\ref{sec:app})}\label{ssec:obabco_overview}

Next, we study discretization of the OHO dynamics.
At a high level, we follow the approach of~\cite{scr3, scr4} in which local error estimates based on coupling arguments are combined with Harnack-type inequalities to obtain R\'enyi divergence bounds.

Suppose that $(X^\msOHO, P^\msOHO)$ and $(\widehat X,\widehat P)$ are the outputs of the OHO dynamics and a suitable discretization, started from the same point.
We define the local position and momentum errors to be
\begin{align*}
    \mc E^\msx(x,p) \deq \norm{\widehat X - X^\msOHO}_{L^2}\,, \qquad \mc E^\msp(x,p) \deq \norm{\widehat P - P^\msOHO}_{L^2}\,,
\end{align*}
for an appropriate choice of coupling.
Note that these errors can depend on the initial point $(x_0, p_0)$.
The standard technique of $W_2$ local error analysis (see~\cite{MilTre21StochNum}) converts these one-step estimates into multi-step bounds.
Informally,
\begin{align*}
    \mc E^\msx + \mc E^\msp \lesssim d^{1/2} h^p \qquad\implies\qquad W_2(\widehat{\bs \mu}_N, \bs\pi) \le \varepsilon~\text{in}~N \asymp (d/\varepsilon^2)^{1/(2p-2)}~\text{steps}\,.
\end{align*}
Here, $\widehat{\bs\mu}_N$ is the law of the output of the discretization after $N$ steps, and heuristically the bound above arises from summing the local error $d^{1/2} h^p$ over $N \asymp 1/h$ steps, and setting the final error to be at most $\varepsilon$.
Hence, we aim to establish such an estimate with $p = 3$ to achieve $d^{1/4}$ scaling.

In the context of discretization of the underdamped Langevin dynamics, via the technique of shifted composition,~\cite{scr4} showed (informally)
\begin{align}\label{eq:renyi_local_err}
    h^{-1} \mc E^\msx + \mc E^\msp \lesssim d^{1/2} h^p \qquad\implies\qquad \KL(\widehat{\bs \mu}_N \mmid \bs\pi) \le \varepsilon^2~\text{in}~N \asymp (d/\varepsilon^2)^{1/(2p-2)}~\text{steps}\,.
\end{align}
In particular, KL discretization seems to require a more stringent local error bound in which the position local error is multiplied by an extra factor of $h^{-1}$.
This arises due to the degenerate regularity of underdamped Langevin, but is thankfully compensated by the fact that well-designed numerical schemes often satisfy $\mc E^\msx \lesssim h\mc E^\msp$.
Unfortunately, the framework therein does not directly apply to R\'enyi divergence bounds.
In this work, we show that~\eqref{eq:renyi_local_err} essentially still holds when $\KL$ is replaced by $\Renyi_q$ by a ``discrete Girsanov'' martingale argument (Lemma~\ref{lem:renyi_mg}).\footnote{However, we note that the framework of~\cite{scr4} produces more refined bounds by taking into account the so-called weak error, whereas this is not handled by our present approach.}

Thus, it remains to find an appropriate numerical scheme. The standard leapfrog discretization OBABO (see \S\ref{sec:hmc_bg} for a definition) that is used in implementations does not satisfy the additional property that $\mc E^\msx \lesssim h\mc E^\msp$ which is needed for R\'enyi discretization. 
Indeed, one can show that $\mc E^\msp \lesssim d^{1/2} h^3 \asymp \mc E^{\msx}$ even for Gaussian targets and at stationarity (i.e., with high probability over the initial point $(x,p) \sim\bs\pi$).

Hence, we propose a novel correction step~\eqref{eq:C} which exactly corrects for this deficiency, leading to our proposed OBABCO scheme. The~\eqref{eq:C} step is constructed to match an extra order of the Taylor expansion of~\eqref{eq:H} and ensures the desired property $\mc E^\msx \lesssim h\mc E^\msp$. Moreover this extra step does not require any additional gradient evaluations and is easy to implement.

There is one further challenge: the local error bound of order $d^{1/2} h^3$ holds at stationarity, but not necessarily when far from stationarity.
This is because it crucially relies on delicate Gaussian chaos inequalities (Lemma~\ref{lem:chaos}).
To get around this, we conduct a two-phase analysis; this does not change the algorithm, only our approach to study it.
In the first phase, we show that with step size $h \asymp d^{-1/4}$ and by using a crude bound on the one-step local error (of order $d^{1/2} h^2$) but which does not require the Gaussian chaos inequality, we can reach a R\'enyi divergence of order $d^{1/2}$.
Then, in the second phase, we show that a R\'enyi divergence of order $d^{1/2}$ is enough to leverage change of measure inequalities that transfer the Gaussian chaos bounds from stationarity to the algorithm iterates.
Combining these two phases yields our algorithmic warm start (Theorem~\ref{thm:main-warm}).
        \section{Analysis of OHO}\label{sec:oho}

In this section, we carry out the analysis of the idealized OHO dynamics. See \S\ref{ssec:oho_overview} for an overview.
Recall that the~\eqref{eq:O} step refers to a partial refreshment of the momentum, and the~\eqref{eq:H} step refers to the ideal Hamiltonian flow for time $h$.
We adopt the following convenient notation. Starting from $(x_0, p_0)\in\R^d\times \R^d$, we denote the intermediate steps of the OHO dynamics via superscripts:
\begin{align}
    (X^\msO, P^\msO) &\deq (x_0,\; e^{-\gamma h/2}\, p_0 + \sqrt{1-e^{-\gamma h}}\, \bar B_1)\,,\tag{O}\\
    (X_h^{\msOH}, P_h^{\msOH}) &\deq \Phi_h(X^\msO, P^\msO)\,, \tag{H}\\
    (X^{\msOHO}, P^{\msOHO}) &\deq (X_h^{\msOH},\; e^{-\gamma h/2}\, P_h^{\msOH} + \sqrt{1-e^{-\gamma h}}\, \bar B_2) \tag{O}\,.
\end{align}
Here, $\bar B_1, \bar B_2\sim \cN(0,I_d)$ are independent. 
Note that we also add a subscript along the~\eqref{eq:H} step in order to index time, which is convenient for the proofs.

Below we first state our results about contraction (in \S\ref{ssec:oho:contraction}) and regularity (in \S\ref{ssec:harnack}) of the OHO dynamics, and then we provide all the proofs for this section (in \S\ref{ssec:oho_proofs}).

\subsection{Contraction}\label{ssec:oho:contraction}

Our first goal is to establish almost-sure contraction of the OHO dynamics, under strong convexity of $V$ and synchronous coupling.
Let $(\bar x_0,\bar p_0) \in \R^d\times\R^d$, and use bars to denote the OHO process with this initialization: $(\bar X^\msO, \bar P^\msO)$, $(\bar X_h^\msOH, \bar P_h^\msOH)$, $(\bar P^\msOHO, \bar X^\msOHO)$.
We couple this with the process $(X^\msO, P^\msO)$, $( X_h^\msOH,  P_h^\msOH)$, $( P^\msOHO,  X^\msOHO)$ from initialization $(x_0, p_0)$ by taking the same noise variables $\bar B_1$, $\bar B_2$.
We use the following shorthand notation to refer to the distances between these processes:
\begin{align*}
    \delta X^\msOHO \deq X^\msOHO - \bar X^\msOHO\,, \qquad \delta P^\msOHO \deq P^\msOHO - \bar P^\msOHO\,, \qquad \delta x_0 \deq x_0 - \bar x_0\,, \qquad \delta p_0 \deq p_0 - \bar p_0\,.
\end{align*}

\begin{lemma}[Contraction of OHO]\label{lem:oho_contraction}
    Let $0 \prec \alpha I \preceq \nabla^2 V \preceq \beta I$, $\gamma \gg \sqrt{\beta}$, and $h \ll \nfrac{\alpha^{1/2}}{\gamma^2}$, and adopt the notation above.
    Then, for some universal constant $c' > 0$, almost surely,
    \begin{align*}
        \norm{\delta X^\msOHO}^2 + \bigl\lVert \delta X^\msOHO  + \frac{2}{\gamma}\,\delta P^\msOHO\bigr\rVert^2
        &\le \exp\Bigl(-\frac{c' \alpha h}{\gamma}\Bigr)\,\Bigl( \norm{\delta x_0}^2 + \bigl\lVert \delta x_0 + \frac{2}{\gamma} \,\delta p_0\bigr\rVert^2\Bigr)\,.
    \end{align*}
\end{lemma}

\begin{remark}[Coordinate system]
    Lemma~\ref{lem:oho_contraction} shows that OHO contracts in the coordinate system $(x, x+\frac{2}{\gamma}\,p)$, which is commonly used in the analysis of the underdamped Langevin diffusion~\cite[\S 5.3]{chewibook}.
    Since it can be checked that $ \tfrac{1}{3}(\norm x^2 + \tfrac{4}{\gamma^2}\,\norm p^2) \le \norm x^2 + \| x + \tfrac{2}{\gamma}\,p\|^2 \le 3\,(\norm x^2 + \tfrac{4}{\gamma^2}\,\norm p^2)$, 
    Lemma~\ref{lem:oho_contraction} implies the same contraction in the coordinate system $(x, \frac{2}{\gamma}\,p)$ up to a constant pre-factor.
\end{remark}

\begin{remark}[Comparison with the literature]
    Lemma~\ref{lem:oho_contraction} is similar to contraction results for the underdamped Langevin diffusion~\cite[\S 5.3]{chewibook}, ideal HMC~\cite{chenvempala2019hmc}, and discretizations such as OBABO (e.g.,~\cite{LeiPauWha24Kinetic}).
    For example, Lemma~\ref{lem:oho_contraction} uses the same coordinate system and a similar proof as the analysis for ULD, and achieves exponential convergence with the same $\alpha/\gamma$ rate.
    In fact, Lemma~\ref{lem:oho_contraction} recovers the contraction for ULD by letting $h\searrow 0$.
    However, these prior results do not cover the OHO dynamics.
\end{remark}

\subsection{Harnack inequalities}\label{ssec:harnack}

Let $\bs P_h$ denote the Markov kernel corresponding to one iteration of OHO\@.
The main results of this section are the following Harnack inequalities\footnote{The terminology alludes to the fact that the following inequalities are dual to a certain family of functional inequalities; see~\cite{scr1} for further discussion.} for the OHO dynamics.
We prove one result under the assumption $0 \preceq \nabla^2 V \preceq \beta I$, and we show a second improved result under the following higher-order smoothness assumption.

\begin{assumption}[Operator-Lipschitz Hessian]\label{as:hess-lip}
    For all $x_0, \bar x_0 \in \R^d$ and some constant $\betH \geq 0$,
    \begin{align*}
        \norm{\nabla^2 V(x_0) - \nabla^2 V(\bar x_0)}_{\operatorname{op}} \leq \betH\, \norm{x_0 - \bar x_0}\,.
    \end{align*}
\end{assumption}

As for underdamped Langevin~\cite{scr4}, the strength of the Harnack inequalities of OHO are governed by the parameter
\begin{align}\label{eq:omeg}
    \omega \deq
        \alpha/\gamma\,.
\end{align}

\begin{theorem}[Harnack inequality for OHO]\label{thm:harnack}
     Suppose that $0 \prec \alpha I \preceq \nabla^2 V \preceq \beta I$ (strongly convex and smooth potential) and $\gamma \geq \sqrt{32 \beta}$ (high friction). Let $\kappa := \beta/ \alpha$ denote the condition number and let $T \deq Nh$ denote the total elapsed time. Then for any step size $h \ll \frac{1}{\beta^{1/2} q^{1/2}} \wedge  \frac{1}{\kappa^{1/2}\gamma}$,
    \begin{align*}
        \Renyi_q(\delta_{x_0, p_0} \bs P_{h}^N \mmid \delta_{\bar x_0, \bar p_0} \bs P_h^N) \leq q C(\alpha, \beta, \gamma, T)\, \Bigl\{ \norm{x_0 - \bar x_0}^2 + \frac{1}{\gamma_0^2}\, \norm{p_0 - \bar p_0}^2\Bigr\} + \beta^2 dh^3 q T\,,
    \end{align*}
    where
    \begin{align*}
        C(\alpha, \beta, \gamma, T) \lesssim \frac{1}{\gamma}\, \Bigl(\frac{\omega}{\exp(c\omega T) - 1} \Bigr)^3 + \frac{\gamma \omega}{\exp(c\omega T) - 1}\,, \qquad \gamma_0 \gtrsim \gamma + \frac{\one_{T \leq 1/\abs{\omega}}}{T}\,,
    \end{align*}
    and $c > 0$ is a universal constant.
    
    Moreover, if Assumption~\ref{as:hess-lip} also holds and $h \lesssim \nfrac{1}{\betH^{2/7} \gamma^{1/7} d^{1/7}}$, then the final term is not present:
    \begin{align*}
        \Renyi_q(\delta_{x_0, p_0} \bs P_{h}^N \mmid \delta_{\bar x_0, \bar p_0} \bs P_h^N) \leq q C(\alpha, \beta, \gamma, T)\, \Bigl\{ \norm{x_0 - \bar x_0}^2 + \frac{1}{\gamma_0^2}\, \norm{p_0 - \bar p_0}^2\Bigr\}\,.
    \end{align*}
\end{theorem}

Theorem~\ref{thm:harnack} provides detailed information about the long-term regularity and convergence of the OHO dynamics. In the sequel, we only need the following consequence: when $\alpha > 0$ and $\gamma \ge \sqrt{32\beta}$, the constant $C(\alpha,\beta,\gamma, T)$ decays at the rate $\exp(-\Omega(\omega T))$ as $T\to\infty$, yielding rapid R\'enyi convergence of OHO to the stationary measure $\bs\pi$. As discussed in \S\ref{ssec:oho_overview}, this avoids the use of hypocoercive PDE techniques, which are currently not known to imply exponential decay in R\'enyi divergence even for the simpler underdamped Langevin dynamics.

\par Theorem~\ref{thm:harnack} focuses on the setting of strongly convex potentials ($\alpha > 0$) and high friction ($\gamma \gtrsim \sqrt{\beta})$ since this is the setting of relevance for our application to sampling algorithms in \S\ref{sec:app}. However, the Harnack inequalities in Theorem~\ref{thm:harnack} extend more generally without these assumptions. We detail this extension in \S\ref{sec:app-oho}.

\begin{remark}\label{rem:harnack-recover}
    The R\'enyi bounds in Theorem~\ref{thm:harnack} match the analogous bounds for underdamped Langevin in~\cite[Theorem 3.2]{scr4} up to absolute constants, except for an additional additive term $\beta^2 d h^3 q T$ without Assumption~\ref{as:hess-lip}.
    We do not believe this term is fundamental.
    For example, applying Theorem~\ref{thm:harnack} for one iteration and then Lemma~\ref{lem:oho_contraction} for $N-1$ iterations yields the alternative bound
    \begin{align*}
        \Renyi_q(\delta_{x_0, p_0} \bs P_{h}^N \mmid \delta_{\bar x_0, \bar p_0} \bs P_h^N) \lesssim q \,\Bigl(\frac{\norm{x_0 - \bar x_0}^2}{\gamma h^3} + \frac{\norm{p_0 - \bar p_0}^2}{\gamma h} \Bigr)\exp\bigl(-\Omega(\omega T)\bigr) + \beta^2 dh^4\,,
    \end{align*}
    for which the additive term remains bounded as $T\to\infty$ (rather than growing at the rate $\Omega(T)$). This is perhaps an artefact of our one-shot coupling proof (Theorem~\ref{thm:regularity} below).
    Nevertheless, we note that (1) the bound in Theorem~\ref{thm:harnack} recovers the one for underdamped Langevin as $h\searrow 0$, and (2) the additive term is not the bottleneck for our eventual applications.
\end{remark}

To prove Theorem~\ref{thm:harnack}, we build upon a line of work developing the \emph{shifted composition} framework~\cite{scr1, scr2, scr3, scr4}, which in turn was inspired by the coupling-based analysis of~\cite{ArnThaWan06HarnackCurvUnbdd, Wang12Coupling}.
At a high level, this framework seeks to ``upgrade'' a contraction inequality for a Markov process---such as the one that we established for OHO in Lemma~\ref{lem:oho_contraction}---to a R\'enyi divergence bound.
The blueprint for the framework is as follows:
\begin{enumerate}
    \item Prove a \emph{one-step} regularity bound, i.e., establish Theorem~\ref{thm:harnack} for $N=1$.
    \item Construct a suitable \emph{auxiliary process} whose law starts at $\delta_{\bar x_0,\bar p_0}$ and ends at $\delta_{x_0,p_0} \bs P_h^N$.
    Then, to bound the R\'enyi divergence between $\delta_{x_0,p_0}\bs P_h^N$ and $\delta_{\bar x_0,\bar p_0} \bs P_h^N$, it suffices to bound the R\'enyi divergence along the paths corresponding to the auxiliary process and the second process.
    \item Show that the auxiliary process and the first process contract together; in particular, show that the law of the auxiliary process at the terminal time is indeed $\delta_{x_0,p_0} \bs P_h^N$.
    \item Using the shifted composition rule (Theorem~\ref{thm:shifted_chain_rule}) and the one-step regularity bound from Step 1, bound the R\'enyi divergence along the paths in terms of distances between the auxiliary process and the first process, which are controlled by Step 3.
\end{enumerate}
We now briefly comment on the application of this framework to the OHO dynamics, deferring the details to \S\ref{ssec:oho_proofs}.
Step 1 is accomplished by a one-shot coupling argument, similar to the works of~\cite{roberts2002one, bou2023mixing, monmarche2024entropic, chak2025reflection, gouraud2025hmc}. Once the right auxiliary process has been constructed, Step 3 is shown by a (significantly) more involved version of Lemma~\ref{lem:oho_contraction}.

Finally, Steps 2 (construction of the auxiliary process) and 4 (application of shifted composition) are non-trivial, since the OHO dynamics are degenerate in the same way as the underdamped Langevin diffusion.
However, here we can reuse the machinery of~\cite{scr4} which tackled the underdamped Langevin case, and so we focus on Steps 1 and 3.

\subsection{Proofs}\label{ssec:oho_proofs}

\subsubsection{Proof of Lemma~\ref{lem:oho_contraction}}

\begin{proof}[Proof of Lemma~\ref{lem:oho_contraction}]
Consider the twisted norm
\begin{align*}
    \mc L(\delta x_0, \delta p_0) \deq \norm{\delta x_0}^2 + \norm{\delta x_0 + \frac{2}{\gamma}\, \delta p_0}^2\,.
\end{align*}
Due to the synchronous coupling, the second~\eqref{eq:O} step produces
\begin{align*}
    \mc L(\delta X^\msOHO, \delta P^\msOHO) = \mc L(\delta X^{\msOH}_h, e^{-\frac{\gamma h}{2}}\, \delta P^{\msOH}_h)\,.
\end{align*}
We seek to show that for an appropriate $c < 1$,
\begin{align*}
    \norm{\delta X^{\msOH}_h}^2 + \bigl\lVert \delta X^{\msOH}_h + \frac{2}{\gamma}\, e^{-\frac{\gamma h}{2}}\, \delta P^{\msOH}_h\bigr\rVert^2
    &\overset{\text{?}}{\leq} c\,\Bigl( \norm{\delta x_0}^2 + \bigl\lVert\delta x_0 + \frac{2}{\gamma}\, \delta p_0\bigr\rVert^2\Bigr) \\
    &=c\,\Bigl( \norm{\delta X^\msO}^2 + \bigl\lVert\delta X^\msO + \frac{2}{\gamma}\,e^{\frac{\gamma h}{2}}\, \delta P^\msO\bigr\rVert^2\Bigr)\,.
\end{align*}
It suffices therefore to show that the time-varying Lyapunov function
\begin{align*}
    \mc L_t( \delta X_t^{\msOH}, \delta P_t^{\msOH}) \deq \norm{\delta X_t^{\msOH}}^2 + \bigl\lVert \delta X_t^{\msOH} + \frac{2}{\gamma}\, e^{\gamma (\frac{h}{2} - t)}\, \delta P_t^{\msOH}\bigr\rVert^2
\end{align*}
decays exponentially in time along the Hamiltonian flow~\eqref{eq:H}. 
Indeed, note that this interpolates between the desired quantities at times $t=0$ and $t=h$. Differentiating this function and defining $\delta Z_t^{\msOH} \deq \delta X_t^{\msOH} + \frac{2}{\gamma}\, e^{\gamma (\frac{h}{2} - t)}\, \delta P_t^{\msOH}$, we have
\begin{align*}
    \partial_t \delta X_t^{\msOH} &= \delta P_t^{\msOH} = \frac{\gamma}{2}\, e^{-\gamma (\frac{h}{2} - t)}\, (\delta Z_t^{\msOH} - \delta X_t^{\msOH})\,, \\[0.25em]
    \partial_t \delta Z_t^{\msOH} &= \delta P_t^{\msOH} - \frac{2}{\gamma}\, e^{\gamma (\frac{h}{2} - t)}\, H_t\, \delta X_t^{\msOH} - 2e^{\gamma (\frac{h}{2} - t)}\, \delta P_t^{\msOH} \\[0.25em]
    &= \bigl(\frac{\gamma}{2}\, e^{-\gamma (\frac{h}{2} - t)} - \gamma\bigr)\,(\delta Z_t^{\msOH} -\delta X_t^{\msOH}) - \frac{2}{\gamma}\, e^{\gamma (\frac{h}{2} - t)}\, H_t\, \delta X_t^{\msOH}\,,
\end{align*}
where here $H_t$ is the integrated Hessian $\int_0^1 \nabla^2 V((1-s)X_t^\msOH + s \bar X_t^{\msOH}) \, \D s$ which obeys the bounds $\alpha I \preceq H_t \preceq \beta I$ for all $t \in [0,h]$.
Thus, if we take the time-derivative of the Lyapunov function, we find, letting $c_t \deq  e^{\gamma (\frac{h}{2} - t)}$ for $t \in [0, h]$,
\begin{align*}
    \frac{1}{2}\,\partial_t 
     \mc L_t( \delta X_t^{\msOH}, \delta P_t^{\msOH})
    = -\begin{bmatrix}
        \delta X_t^\msOH \\[0.25em]
        \delta Z_t^\msOH
    \end{bmatrix}^\top \underbrace{\begin{bmatrix}
        \frac{\gamma}{2 c_t} & -\frac{\gamma}{2} + \frac{c_t H_t}{\gamma} \\[0.5em]
        -\frac{\gamma}{2} + \frac{c_t H_t}{\gamma} & \gamma\,(1-\frac{1}{2 c_t})
    \end{bmatrix}}_{\eqqcolon M_t}
    \begin{bmatrix}
        \delta X_t^\msOH \\[0.25em]
        \delta Z_t^\msOH
    \end{bmatrix}\,.
\end{align*}
Note that when $c_t \equiv 1$, this is precisely the matrix which appears in the standard proof of contraction for the underdamped Langevin dynamics, which corresponds to the infinitesimal version of this calculation (this only occurs if $h \searrow 0$).

Since $H_t$ is positive semidefinite, it can be diagonalized, hence $M_t$ can be block-diagonalized into $2 \times 2$ block matrices. By the formula for the minimum eigenvalue of a $2 \times 2$ matrix, we have
\begin{align*}
    \lambda_{\min}(M_t) = \frac{\gamma}{2} \cdot \min_{\lambda \in \text{spec}(H_t)}\biggl\{ 1 - \sqrt{\Bigl(1 - \frac{2c_t \lambda}{\gamma^2 }\Bigr)^2 + \Bigl(1-\frac{1}{c_t}\Bigr)^2}\biggr\}\,.
\end{align*}
We now bound the terms in the square root. Observe that 
\begin{align*}
    \Bigl(1 - \frac{2c_t \lambda}{\gamma^2 }\Bigr)^2
    = \Bigl( 1 - \Theta\Bigl( \frac{\lambda}{\gamma^2} \Bigr) \Bigr)^2
    = 1 - \Theta\Bigl( \frac{\lambda}{\gamma^2} \Bigr) \,,
\end{align*}
where the first step is because $h \leq 1/\gamma$ implies $c_t \asymp 1$, and the second step is because $\lambda / \gamma^2 \ll \lambda / \beta \leq 1$. Also observe that 
\begin{align*}
    \Bigl( 1 - \frac{1}{c_t} \Bigl)^2
   \leq \gamma^2 h^2 \ll \frac{\alpha}{\gamma^2} \,,
\end{align*}
where the first step is because $h \leq 1/\gamma$, and the second step is because $h \ll \sqrt{\alpha} / \gamma^2$. By combining the above three displays, applying the Taylor expansion of $1 - \sqrt{1 - \delta} \asymp \delta$ for small $\delta > 0$, and using the fact that $H_t$ has eigenvalues in $[\alpha,\beta]$, we conclude
\begin{align}
\lambda_{\min}(M_t)
&= \frac{\gamma}{2} \cdot \min_{\lambda \in \text{spec}(H_t)}\bigl\{ 1 - \sqrt{1 - \Theta(\lambda/\gamma^2)}\bigr\}
\asymp \gamma \cdot \min_{\lambda \in \text{spec}(H_t)} \frac{\lambda}{\gamma^2} = \frac{\alpha}{\gamma}\,. \qedhere
\end{align}
\end{proof}

\subsubsection{First-order approximation of the Hamiltonian flow}

We first record the following useful helper results.

\begin{lemma}[First-order approximation]\label{lem:hmc-expansion}
Let $(x_t, p_t)_{t\ge 0}$, $(\bar x_t, \bar p_t)_{t\ge 0}$ coevolve along the Hamiltonian dynamics~\eqref{eq:ham_ode} with step size $h \lesssim \nfrac{1}{\beta^{1/2}}$. Define $H_t \deq \int_0^1 \nabla^2 V((1-s) \,x_t + s \,\bar x_t) \, \D s$ and assume that $-\beta I \preceq \nabla^2 V \preceq \beta I$. Denoting $\delta x_t \deq x_t - \bar x_t$ and $\delta p_t \deq p_t - \bar p_t$, then
\begin{align*}
    \delta x_h = \delta x_0 + h \,\delta p_0 + R_x\,, \qquad \delta p_h = \delta p_0 - \Bigl(\int_0^h H_s \, \D s\Bigr)\, \delta x_0 + R_p\,,
\end{align*}
where
\begin{align*}
    \norm{R_x} &\lesssim \beta h^2\, \norm{\delta x_0} + \beta h^3\, \norm{\delta p_0}\,, \\
    \norm{R_p} &\lesssim \beta h^2\, \norm{\delta p_0} + \beta^2 h^3\, \norm{\delta x_0}\,.
\end{align*}
\end{lemma}
\begin{proof}
    Let us start with the error bound for $x$. Solving Hamilton's equations gives
    \begin{align}\label{eq:ham-ev-x}
        \delta x_h 
        &= \delta x_0 + h \, \delta p_0 - \int_0^h \int_0^t \nabla V(x_s)\, \D s \, \D t = \delta x_0 + h\, \delta p_0 - \int_0^h (h-s)\, H_s\, \delta x_s \, \D s\,.
    \end{align}
    Via a differential comparison argument (Lemma~\ref{lem:diff-compare}) combined with the spectral bound $-\beta I \preceq H_s \preceq \beta I$ and step size assumption $h \lesssim \nfrac{1}{\beta^{1/2}}$, we have for $s \leq h$ that
    \begin{align*}
        \norm{\delta x_s} \leq (\norm{\delta x_0} + s\,\norm{\delta p_0})\,e^{\beta^{1/2} h} \lesssim \norm{\delta x_0} + s\, \norm{\delta p_0}\,. 
    \end{align*}
    Substituting this into~\eqref{eq:ham-ev-x} gives the bound on $\norm{R_x}$.

    Next we establish the error bound for $p$. By the equation for $\delta x_s$ above,
    \begin{align*}
        \delta p_h = \delta p_0 - \int_0^h H_s\, \delta x_s \, \D s = \delta p_0 - \bigl(\int_0^h H_s \, \D s\bigr) \,\delta x_0 - h\, \bigl(\int_0^h H_s \, \D s \bigr)\, \delta p_0 - \int_0^h H_s\,R_x(s) \, \D s\,.
    \end{align*}
    Using the bound on $\norm{R_x}$ gives the proper bound for the remainder $\norm{R_p}$ as well.
\end{proof}

\begin{lemma}[Differential comparison lemma]\label{lem:diff-compare}
    Let $b,c \geq 0$. Suppose $y: [0, h] \to \R_+$ is continuous and satisfies for all $t \in [0, h]$,
    \begin{align}\label{eq:growth-bound}
        y(t)-y(0) \leq bt + c \int_0^t (t-s)\, y(s) \, \D s\,.
    \end{align}
    Then, for all $t \in [0, h]$,
    \begin{align*}
        y(t) \leq (y(0)+bt)\, e^{{\sqrt{c} t}}\,.
    \end{align*}
\end{lemma}
\begin{proof}
    Let $z(t) \deq y(0) + bt + c \int_0^t (t-s)\, z(s) \, \D s$ for $t \geq 0$. To compute $z(t)$, note that it solves the ODE $z'' = cz$ with boundary conditions $z(0) = y(0)$ and $z'(0) = b$. Thus
    \begin{align*}
        z(t) = y(0) \cosh(\sqrt{c} t) + \frac{b}{\sqrt{c}} \sinh(\sqrt{c} t) \leq (y(0)+bt)\, e^{\sqrt{c} t}\,.
    \end{align*}
    Above, the last step uses the elementary inequalities $\cosh x \leq e^x$ and $\sinh x \leq xe^x$, valid for $x \geq 0$.

    \par It therefore suffices to show that $w \deq y - z$ satisfies $w(t) \leq 0$ for all $t \in [0,h]$. This is straightforward if $c = 0$; henceforth assume $c > 0$. Suppose, for the sake of contradiction, that $M_0 \deq \max_{t \in [0, t_0]} w(t) > 0$ for $t_0 < c^{-1/2}$. Then for all $t \in [0, t_0]$, $w(t) \leq c \int_0^t (t-s) w(s)  \, \D s \leq c t_0^2 M_0$. Maximizing the left hand side over $t \in [0,t_0]$ and dividing through by $M_0$, we have $1 \leq c t_0^2$, which is a contradiction. Hence $w(t) \leq 0$ for $t \in [0,t_0]$. Now repeat the argument on the next interval $t \in [t_0,2t_0]$, the only extra step being that $c \int_0^t (t-s)w(s) \, \D s \leq c \int_{t_0}^t (t-s)w(s) \, \D s$ since we have just shown that $w \leq 0$ on $[0,t_0]$. Iterating this argument allows us to conclude that $w \leq 0$ on the whole interval $[0,h]$, as desired.
\end{proof}

\subsubsection{One-shot coupling}

The goal of this subsection is to prove the following regularity result for one iteration of OHO\@.

\begin{theorem}[Short-time regularity for OHO]\label{thm:regularity}
    Assume that $-\beta I \preceq \nabla^2 V \preceq \beta I$.
    Then, if $h \ll \nfrac{1}{\beta^{1/2} q^{1/2}} \wedge \nfrac{1}{\gamma}$,
    \begin{align*}
        \Renyi_q(\delta_{(x_0, p_0)} \bs P_h \mmid \delta_{(\bar x_0, \bar p_0)} \bs P_h) \lesssim \frac{q\,\norm{\bar x_0 - x_0}^2}{\gamma h^3} + \frac{q\,\norm{\bar p_0 - p_0}^2}{\gamma h} + \beta^2 d h^4 q\,.
    \end{align*}
    Moreover, if Assumption~\ref{as:hess-lip} holds, then
    \begin{align*}
        \Renyi_q(\delta_{(x_0, p_0)} \bs P_h \mmid \delta_{(\bar x_0, \bar p_0)} \bs P_h) \lesssim q\,\norm{\bar x_0 - x_0}^2\,\Bigl(\frac{1}{\gamma h^3} + \betH^2 h^4 d \Bigr)+ \frac{q\,\norm{\bar p_0 - p_0}^2}{\gamma h}\,.
    \end{align*}
    In particular, if $h \lesssim \nfrac{1}{\betH^{2/7} \gamma^{1/7} d^{1/7}}$, then the second factor in the $x$ term absorbs into the first, yielding
        \begin{align*}
        \Renyi_q(\delta_{(x_0, p_0)} \bs P_h \mmid \delta_{(\bar x_0, \bar p_0)} \bs P_h) \lesssim  \frac{q\, \norm{x_0 - \bar x_0}^2}{\gamma h^3} + \frac{q\, \norm{\bar p_0 - p_0}^2}{\gamma h}\,.
    \end{align*}
\end{theorem}

To prove these results, we make use of the one-shot coupling method, similar to the works of~\cite{bou2023mixing, monmarche2024entropic, chak2025reflection, gouraud2025hmc}. This technique originates in earlier works~\cite{roberts2002one, madras2010quantitative}. In particular, we highlight the short-time regularity which appears in~\cite[Theorem 2]{monmarche2024entropic} and~\cite[Lemma 15]{bou2023mixing}.
We proceed to describe the strategy.

Recall that $\Phi_t(\cdot, \cdot)$ denotes the Hamiltonian flow map $(x_0, p_0) \mapsto (x_t,p_t)$ for the Hamiltonian dynamics in~\eqref{eq:ham_ode}.
Let $\Phi_t^{\msx}, \Phi_t^{\msp}$ denote the respective components of this map.
Then, for two independent standard Gaussian random variables $\bar B_1$, $\bar B_2$, we can write
\begin{align*}
    \bar X^\msOHO_h &= \Phi_h^{\msx}\bigl(\bar x_0,\, e^{-\gamma h/2} \bar p_0 + \sqrt{1-e^{-\gamma h}} \bar B_1\bigr)\,, \\
    \bar P^\msOHO_h &= e^{-\gamma h/2} \Phi_h^{\msp}\bigl(\bar x_0,\, e^{-\gamma h/2} \bar p_0 + \sqrt{1-e^{-\gamma h}} \bar B_1\bigr) + \sqrt{1-e^{-\gamma h}} \bar B_2\,.
\end{align*}
See the top process in Figure~\ref{fig:one-shot}.
Similarly, $(X^\msOHO, P^\msOHO)$ admits an analogous representation in terms of the coupled noise $\bar{B_1},\bar{B_2}$, but with $(x_0, p_0)$ in place of $(\bar x_0, \bar p_0)$:
\begin{align}\label{eq:shift-constraint-orig}
\begin{aligned}
    X^\msOHO_h &= \Phi_h^{\msx}\bigl(x_0,\, e^{-\gamma h/2} p_0 + \sqrt{1-e^{-\gamma h}} \bar B_1\bigr)\,, \\
    P^\msOHO_h &= e^{-\gamma h/2} \Phi_h^{\msp}\bigl(x_0,\, e^{-\gamma h/2} p_0 + \sqrt{1-e^{-\gamma h}} \bar B_1\bigr) + \sqrt{1-e^{-\gamma h}} \bar B_2\,.
\end{aligned}
\end{align}
See the bottom process in Figure~\ref{fig:one-shot}. In addition to this direct representation of $(X^\msOHO_h, P^\msOHO_h)$, we also consider an auxiliary process which starts from $(\bar x_0, \bar p_0)$ rather than $(x_0, p_0)$ and hits $(X^\msOHO_h, P^\msOHO_h)$ by using \emph{auxiliary} noise increments $B_1,B_2$:
\begin{align}\label{eq:shift-constraint}
\begin{aligned}
    X^\msOHO_h &= \Phi_h^{\msx}\bigl(\bar x_0,\, e^{-\gamma h/2} \bar p_0 + \sqrt{1-e^{-\gamma h}} B_1\bigr)\,, \\
    P^\msOHO_h &= e^{-\gamma h/2} \Phi_h^{\msp}\bigl(\bar x_0,\, e^{-\gamma h/2} \bar p_0 + \sqrt{1-e^{-\gamma h}} B_1\bigr) + \sqrt{1-e^{-\gamma h}} B_2\,.
\end{aligned}
\end{align}
See the diagonal blue process in Figure~\ref{fig:one-shot}. If there is a unique solution $(B_1,B_2)$, then it is a deterministic function of $(x_0, p_0, \bar x_0, \bar p_0, \bar B_1, \bar B_2)$.

\begin{figure}
    \centering
    \includegraphics[width=0.9\linewidth]{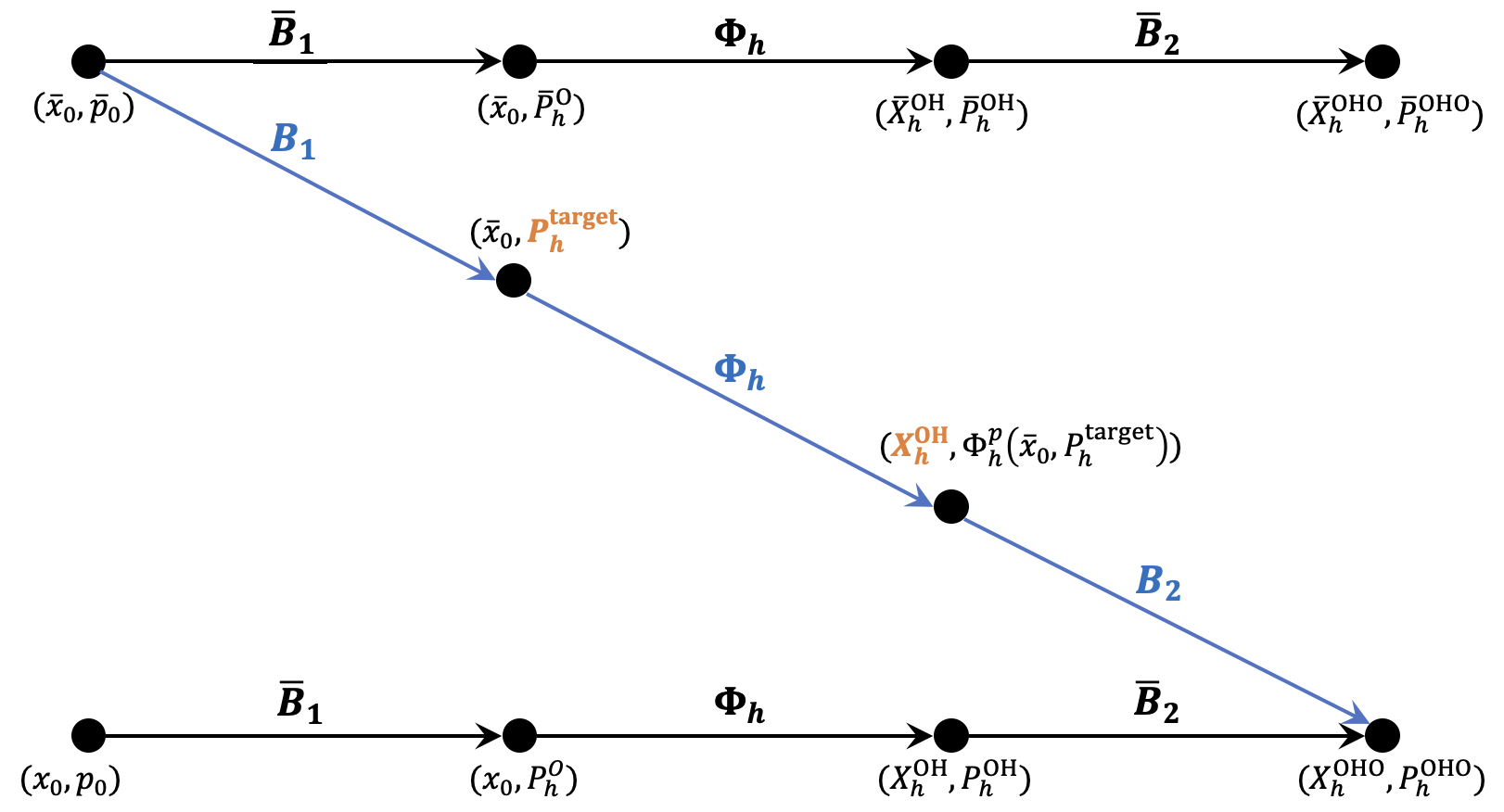}
    \caption{\footnotesize Schematic diagram of the one-shot coupling in the proof of Theorem~\ref{thm:regularity}. Top and bottom (in black): the two OHO processes $(\bar{x}_0,\bar{p}_0) \mapsto (\bar{X}_h^\msOHO,\bar{P}_h^{\msOHO})$ and $(x_0,p_0) \mapsto (X_h^\msOHO,P_h^{\msOHO})$ are coupled to use the same Gaussian noise increments $\bar{B}_1, \bar{B}_2$. Diagonal (in blue): different noise increments $B_1,B_2$ enable an auxiliary OHO process $(\bar{x}_0,\bar{p}_0) \mapsto (X_h^\msOHO,P_h^{\msOHO})$ that starts from one process and ends at the other. To achieve this interpolation, the auxiliary noise increments $B_1,B_2$ are uniquely determined as a function of $x_0,p_0,\bar{x}_0,\bar{p}_0,\bar{B_1},\bar{B_2}$. Orange: $B_1$ in the first ``O'' step is uniquely determined so that the resulting momentum $P_h^{\target}$ will enable the correct $x$-coordinate $X_h^\msOH$ after the ``H'' step. $B_2$ is then uniquely determined so that the $p$-coordinate $P_h^\msOHO$ will match after the final ``O'' step. 
    }
    \label{fig:one-shot}
\end{figure}

This expresses $(\bar X_h^\msOHO, \bar P_h^\msOHO) = \mc T(\bar B_1,\bar B_2)$ for some transformation $\mc T$, and $(X_h^\msOHO, P_h^\msOHO) = \mc T(B_1,B_2)$ for the \emph{same} transformation $\mc T$.
Thus, by the data-processing inequality,
\begin{align}\label{eq-pf:thm-regularity:postprocessing}
    \Renyi_q(\delta_{(x_0, p_0)} \bs P_h \mmid \delta_{(\bar x_0, \bar p_0)} \bs P_h)
    \le \Renyi_q(\law(B_1,B_2) \mmid \law(\bar B_1,\bar B_2))\,.
\end{align}
To bound the right-hand side, we use the change-of-variables formula, which requires understanding the mapping $(\bar B_1,\bar B_2) \mapsto (B_1,B_2)$ (with $x_0$, $p_0$, $\bar x_0$, $\bar p_0$ fixed). We do this below.

\begin{proof}[Proof of Theorem~\ref{thm:regularity}]
\textbf{Step 1: computing the mapping $(\bar B_1,\bar B_2) \mapsto (B_1,B_2)$.} First we solve for $B_1$ by ensuring that $X^\msOHO_h$ agrees in the two representations~\eqref{eq:shift-constraint-orig} and~\eqref{eq:shift-constraint}. This amounts to
\begin{align*}
    \Phi_h^{\msx}\bigl(\bar x_0,\, e^{-\gamma h/2} \bar p_0 + \sqrt{1-e^{-\gamma h}} B_1\bigr)
    =
    \Phi_h^{\msx}\bigl(x_0,\, e^{-\gamma h/2} p_0 + \sqrt{1-e^{-\gamma h}} \bar B_1\bigr)
    \,.
\end{align*}
To solve for $B_1$, let $\Psi_h(\bar x_0, x)$ denote the unique solution (guaranteed for $h \ll \nfrac{1}{\beta^{1/2}}$, see the comment after Lemma~\ref{lem:inv_fn}) to the equation $\Phi_h^{\msx}(\bar x_0, \Psi_h(\bar x_0, x)) = x$. In words, $\Psi_h(\bar x_0, x)$ is the momentum needed to hit $x$ after time $h$ starting from $x_0$. Furthermore, $\Psi_h, \Phi_h^\msx$ will both be diffeomorphisms whose derivatives we will calculate explicitly. Then the above equation is equivalent to
\begin{align}\nonumber
    e^{-\gamma h/2} \bar p_0 + \sqrt{1-e^{-\gamma h}} B_1
    =
    P_h^\target 
    \deq \Psi_h(\bar x_0, \Phi_h^{\msx}(x_0, P_h^\msO))
    \,,
\end{align}
where we use the shorthand $P_h^\msO \deq e^{-\gamma h/2} p_0 + \sqrt{1-e^{-\gamma h}} \bar B_1$.
Solving for $B_1$ gives
\begin{align}\label{eq-pf:thm-regularity:B1}
    B_1 = \frac{P_h^\target - e^{-\gamma h/2} \bar p_0}{\sqrt{1-e^{-\gamma h}} }\,.
\end{align}
\par Next we solve for $B_2$ by ensuring that $P^\msOHO_h$ agrees in the two representations~\eqref{eq:shift-constraint-orig} and~\eqref{eq:shift-constraint}.
This amounts to
\begin{align*}
    e^{-\gamma h/2} \Phi_h^{\msp}\bigl(\bar x_0,\, P_h^\target \bigr) + \sqrt{1-e^{-\gamma h}} B_2
    =
    e^{-\gamma h/2} \Phi_h^{\msp}\bigl(x_0,\, P_h^\msO \bigr) + \sqrt{1-e^{-\gamma h}} \bar B_2
    \,.
\end{align*}
Solving for $B_2$ gives
\begin{align}\label{eq-pf:thm-regularity:B2}
    B_2
    =
   \bar B_2 +  
   \frac{e^{-\gamma h/2}}{\sqrt{1-e^{-\gamma h}} }\, \bigl( 
   \Phi_h^{\msp}\bigl(x_0,\, P_h^\msO \bigr) - \Phi_h^{\msp}\bigl(\bar x_0,\, P_h^\target \bigr)\bigr)
    \,.
\end{align}

\par \textbf{Step 2: bounding the R\'enyi divergence in terms of Radon--Nikodym derivatives.}
By~\eqref{eq-pf:thm-regularity:B1} and~\eqref{eq-pf:thm-regularity:B2}, the mapping $(\bar B_1, \bar B_2) \mapsto (B_1, B_2)$ is a diffeomorphism over $\R^d \times \R^d$. Let $F := F_{x_0, p_0, \bar x_0, \bar p_0}$ denote the function sending $\bar{B_1}$ to $B_1$. Then since $\bar B_1, \bar B_2$ are independent standard Gaussians, the Radon--Nikodym derivative decomposes as 
\begin{align}
    \frac{\D \law(B_1, B_2)}{\D \law(\bar B_1, \bar B_2)}
    =
    \frac{\D \law(B_1)}{\D \Gamma} 
    \cdot 
    \frac{\D \law(B_2 \mid B_1)}{\D \Gamma} 
   =
    \frac{\D F_{\#} \Gamma}{\D \Gamma} 
    \cdot 
    \frac{\D \law(B_2 \mid B_1)}{\D \Gamma} \,,
\end{align}
where $\Gamma$ denotes the standard Gaussian distribution. Hence by~\eqref{eq-pf:thm-regularity:postprocessing}, we can bound the desired R\'enyi divergence by
\begin{align}\label{eq-pf:thm-regularity:renyi}
    \Renyi_q(\delta_{(\bar x_0, \bar p_0)} \bs P_h  \mmid \delta_{(x_0, p_0)} \bs P_h) 
    &\leq \Renyi_q(\law(B_1,B_2) \mmid \law(\bar B_1,\bar B_2)) \nonumber
    \\ &= \frac{1}{q-1} \log  \E\Bigl[\Bigl(\frac{\D  F_\# \Gamma}{\D \Gamma}(\bar B_1) \Bigr)^q \,\E \Bigl[\Bigl(\frac{\D \operatorname{law}(B_2 \mid \bar B_1)}{\D \Gamma}\Bigr)^q \Bigm\vert \bar B_1 \Bigr]\Bigr]
\end{align}

\par \textbf{Step 3: Radon--Nikodym derivative between $B_1$ and $\bar B_1$.} Let us focus on the first term in~\eqref{eq-pf:thm-regularity:renyi}. Express the Radon--Nikodym between $B_1$ and $\bar B_1$ as
\begin{align*}
    \frac{\D F_{\#} \Gamma}{\D \Gamma}(b) &= \exp\Bigl(-\frac{1}{2}\, \norm{F^{-1}(b)}^2 + \frac{1}{2}\, \norm{b}^2 \Bigr)\, \abs{\det \nabla F^{-1}(b)} \\
    &= \exp\Bigl(-\frac{1}{2}\, \norm{F^{-1}(b) - b}^2 + \langle b, b - F^{-1}(b)\rangle \Bigr)\, \bigl\lvert\frac{1}{\det \nabla F \circ F^{-1}(b)}\bigr\rvert\,.
\end{align*}
As we later see, the transform $F$ is invertible, and we will give precise bounds on the magnitude of its determinant. We evaluate, for $q \geq 2$,
\begin{align*}
    &\E_\Gamma \Bigl(\frac{\D F_{\#} \Gamma}{\D  \Gamma} \Bigr)^q = \E_{\Gamma}\Bigl[\Bigl(\frac{\D F_{\#} \Gamma}{\D  \Gamma} \Bigr)^{q-\frac{1}{2}} \cdot \Bigl(\frac{\D  F_{\#} \Gamma}{\D  \Gamma} \Bigr)^{\frac{1}{2}}\Bigr] \\
    &\qquad = \E_{\mathbf Q} \Bigl[\exp\Bigl(-\frac{q-1/2}{2}\, \norm{\bar B_1 - B_1}^2 + (q-1/2)\, \langle B_1, B_1 - \bar B_1\rangle \Bigr) \, {\operatorname{det}^{q-1/2} \nabla F^{-1}(B_1)} \cdot \Bigl(\frac{\D  F_{\#} \Gamma}{\D \Gamma} \Bigr)^{\frac{1}{2}}\Bigr]\,,
\end{align*}
where, under $\mathbf Q$, $B_1$ is Gaussian and $\bar B_1 = F^{-1}(B_1)$. 
Now, let $F_q^{-1} \deq (2q-1)\,F^{-1} - (2q-2)\id$. If we can bound $(2q-1)\, \norm{\nabla F^{-1} - I}_{\op} \leq \frac{1}{2}$ uniformly, then $F_q$ corresponds to a well-defined invertible function, so this definition is valid. 
By Cauchy--Schwarz,
\begin{align*}
    \E_\Gamma \Bigl(\frac{\D F_{\#} \Gamma}{\D \Gamma} \Bigr)^q
    &\le \E_{\mathbf Q}\Bigl[\exp\Bigl(-\frac{(2q-1)^2}{2}\,\norm{\bar B_1-B_1}^2 + (2q-1)\,\langle B_1, B_1-\bar B_1\rangle\Bigr) \det \nabla F_q^{-1}(B_1)\Bigr]^{1/2} \\
    &\qquad{}\times \E_{\mathbf P}\Bigl[\exp\Bigl(\bigl(-\frac{2q-1}{2} + \frac{(2q-1)^2}{2}\bigr)\, \norm{B_1 - \bar B_1}^2\Bigr)\,\frac{\det^{2q-1} \nabla F^{-1}(B_1)}{\det\nabla F_q^{-1}(B_1)}\Bigr]^{1/2}\,.
\end{align*}
Under $\mathbf P$, $\bar B_1$ is Gaussian and $B_1 = F(\bar B_1)$. Moreover, the first term is $\sqrt{\E_\Gamma \frac{\D  (F_q)_{\#} \Gamma}{\D \Gamma}} = 1$.

Now, suppose temporarily that we could ensure $\lambda_{\max} \deq \norm{\nabla F^{-1}(\cdot) - I}_{\operatorname{op}} \ll \nfrac{1}{q}$.
In fact, by~\eqref{eq-pf:thm-regularity:B1},
\begin{align*}
    \lambda_{\max} \le \sup_{p\in\R^d}{\bigl\lVert\bigl(\nabla_p[\Psi_h(\bar x_0, \Phi_h^{\msx}(x_0, p))]\bigr)^{-1} - I\bigr\rVert_{\rm op}}\,.
\end{align*}
To bound this, first introduce the related quantity
\begin{align*}
    \bar \lambda_{\max} \le \sup_{p\in\R^d}{\norm{\nabla_p[\Psi_h(\bar x_0, \Phi_h^{\msx}(x_0, p))] - I}_{\rm op}}\,,
\end{align*}
and we bound this in Lemma~\ref{lem:hmc-distortion} and Lemma~\ref{lem:hmc-distortion-II} below, with and without Assumption~\ref{as:hess-lip}. The maximum eigenvalue is bounded according to the two Jacobian bounds as
\begin{align}\label{eq:eig-bounds}
    \bar \lambda_{\max} &\lesssim  \begin{cases}
        \beta h^2\,, & \text{without Hessian Lipschitzness}\,, \\
        \betH h^2\, \norm{x_0 - \bar x_0}\,, & \text{with Hessian Lipschitzness (Assumption~\ref{as:hess-lip})}\,.
    \end{cases}
\end{align}
Assume that $\bar \lambda_{\max} \ll \nfrac{1}{q}$ and $q \geq 2$. Then, we obtain $\lambda_{\max} \ll \nfrac{1}{q}$ as well, and furthermore that $\lambda_{\max} \lesssim \bar \lambda_{\max}$, since the relation
\begin{align*}
    \norm{A^{-1} - I}_{\op} \leq \frac{\norm{A-I}_{\op}}{1-\norm{A-I}_{\op}}
\end{align*}
always holds so long as $\norm{A-I}_{\op} < 1$. We note the subtle point that the condition $\bar \lambda_{\max} \ll \nfrac{1}{q}$ holds so long as $h \ll \nfrac{1}{\beta^{1/2} q^{1/2}}$, as we can always use the first bound in~\eqref{eq:eig-bounds}. Under these conditions, invoking Lemma~\ref{lem:log-det},
\begin{align}
    \E_\Gamma \Bigl(\frac{\D F_{\#} \Gamma}{\D \Gamma} \Bigr)^q
    &\le \E_{\mathbf P}\Bigl[\exp\bigl(-\log \det(I + (2q-1)\,(\nabla F^{-1}-I)) + (2q-1) \log\det(I + \nabla F^{-1} - I) \bigr) \nonumber\\
    &\hspace{50pt} \times\exp\Bigl(\frac{(2q-1)^2}{2}\norm{B_1 - \bar B_1}^2\Bigr) \Bigr]^{1/2} \nonumber \\
    &\le \exp(O(q^2 \lambda_{\max}^2 d))\, \E_{\mb P}\Bigl[\exp\Bigl(\frac{(2q-1)^2}{2}\norm{B_1 - \bar B_1}^2\Bigr)\Bigr]^{1/2} \,.\label{eq:radon-nikodym-final}
\end{align}

Now, note that for $h \lesssim \nfrac{1}{\gamma}$,
\begin{align*}
    \norm{\bar B_1 - B_1}^2 = \frac{1}{1-e^{-\gamma h}}\, \norm{P_h^\target - \bar P^\msO_h}^2 \asymp \frac{1}{\gamma h}\, \norm{P_h^\target - \bar P_h^\msO}^2\,.
\end{align*}
We split the bound for this into two parts. First, we consider
\begin{align*}
    \norm{P_h^\target - \bar P^\msO} \leq \norm{P_h^\target - P_h^\msO} + \norm{P_h^\msO - \bar P_h^\msO} = \norm{P_h^\target - P_h^\msO} + e^{-\frac{\gamma h}{2}}\,\norm{p_0 - \bar p_0}\,,
\end{align*}
where the first term measures the deviation of $\Psi_h(\bar x_0, \Phi_h^{\msx}(x_0, \cdot))$ from identity, and the latter follows from the expression for the Ornstein--Uhlenbeck semigroup under synchronous coupling. For the first term, we note that since $\Phi_h^{\msx}(x_0, P^\msO) = \Phi_h^{\msx}(\bar x_0, P^\target)$,
 Lemma~\ref{lem:hmc-expansion} gives
\begin{align*}
    P_h^\msO - P_h^\target = -\frac{1}{h}\,(x_0 - \bar x_0) + \frac{R_x}{h}\,,
\end{align*}
where $\norm{R_x} \lesssim \beta h^3\, \norm{P_h^\msO - P_h^\target } + \beta h^2\, \norm{x_0-\bar x_0}$. This implies that, if $h \ll \nfrac{1}{\beta^{1/2}}$,
\begin{align}\label{eq:shooting-diff}
    \norm{P_h^\msO - P_h^\target} \lesssim \frac{1}{h}\,\norm{x_0 - \bar x_0}\,.
\end{align}
Putting this all together gives the resultant bound
\begin{align*}
     \norm{P_h^\target - \bar P_h^\msO}^2 \lesssim {\norm{p_0 - \bar p_0}^2} + \frac{\norm{x_0- \bar x_0}^2}{h^2}\,.
\end{align*}
Therefore,
\begin{align*}
    \log \E_\Gamma \Bigl(\frac{\D  F_{\#} \Gamma}{\D \Gamma} \Bigr)^q
    &\lesssim q^2\,\Bigl[ \lambda_{\max}^2 d + \frac{\norm{x_0-\bar x_0}^2}{\gamma h^3} + \frac{\norm{p_0-\bar p_0}^2}{\gamma h}\Bigr]\,.
\end{align*}

\par \textbf{Step 4: conditional Radon--Nikodym derivative between $B_2$ and $\bar B_2$.} We now focus on the second term in~\eqref{eq-pf:thm-regularity:renyi}. By~\eqref{eq-pf:thm-regularity:B2}, $B_2 - \bar B_2$ is constant conditional on $\bar B_1$ (i.e., it is not a function of $\bar B_2$).
Thus by a standard computation of the divergence between two Gaussians with the same covariance,
\begin{align*}
    \E_{\mb P} \Bigl[\Bigl(\frac{\D \operatorname{law}(B_2 \mid \bar B_1)}{\D \Gamma}\Bigr)^q \Bigm\vert \bar B_1 \Bigr]
    &= \E_{\mb P}\Bigl[\exp\Bigl(\frac{q\,(q-1)}{2}\,\norm{B_2-\bar B_2}^2\Bigr) \Bigm\vert \bar B_1 \Bigr]\,.
\end{align*}
Now by~\eqref{eq-pf:thm-regularity:B2}, we have for $h \lesssim \nfrac{1}{\gamma}$ that
\begin{align*}
     \norm{B_2 - \bar B_2}^2 \lesssim \frac{1}{\gamma h}\, \norm{\Phi_h^{\msp}(x_0, P^\msO) - \Phi_h^{\msp}(\bar x_0, P^\target)}^2\,.
\end{align*}
Via Lemma~\ref{lem:hmc-expansion}, we can bound this by
\begin{align*}
    \norm{\Phi_h^{\msp}(x_0, P^\msO) - \Phi_h^{\msp}(\bar x_0, P^\target)}^2 \lesssim \norm{P_h^\msO - P^\target}^2 + \beta^2 h^2\, \norm{x_0 - \bar x_0}^2 \lesssim \frac{1}{h^2}\, \norm{x_0 - \bar x_0}^2\,,
\end{align*}
for $h \lesssim \nfrac{1}{\beta^{1/2}}$, where the last inequality is via~\eqref{eq:shooting-diff}. This follows regardless of the value of $\bar B_1$. It is clear that the exponential moments of this can be absorbed into the other terms.

\par \textbf{Step 5: putting the bounds together.} Combining~\eqref{eq-pf:thm-regularity:renyi} with the bounds in steps 3 and 4 (to respectively bound the two integrated terms), we conclude the desired bound
\begin{align*}
    \Renyi_q(\delta_{(\bar x_0, \bar p_0)} \bs P_h  \mmid \delta_{(x_0, p_0)} \bs P_h) 
    &\leq \frac{1}{q-1} \log  \E_{\mb P}\Bigl[\Bigl(\frac{\D  F_\# \Gamma}{\D \Gamma}(\bar B_1) \Bigr)^q \,\E_{\mb P} \Bigl[\Bigl(\frac{\D \operatorname{law}(B_2 \mid \bar B_1)}{\D \Gamma}\Bigr)^q \Bigm\vert \bar B_1 \Bigr]\Bigr] \\
    &\lesssim q\,\Bigl(\lambda_{\max}^2 d  + \frac{\norm{x_0 - \bar x_0}^2}{\gamma h^3} + \frac{\norm{p_0 - \bar p_0}^2}{\gamma h}\Bigr)\,. \qedhere
\end{align*}
\end{proof}

\subsubsection{Jacobian bounds}

Here we establish the Jacobian bounds used in the preceding subsection. In words, these results quantify the stability of the map $\bar p_0 \mapsto p_0$ that ensures the same $x$-coordinate after running the Hamiltonian flow $\Phi_h^{\msx}(\bar x_0, \bar p_0) = \Phi_h^{\msx} (x_0, p_0)$ for time $h$ from different initializations $\bar{x_0}, x_0$. This type of bound is essential for one-shot coupling arguments~\cite[Lemma 26]{bou2023mixing}, although the particular statement will vary based on the assumptions.

\paragraph{Bound without Hessian Lipschitzness.}
\begin{lemma}[Jacobian bound]\label{lem:hmc-distortion}
    If $h \lesssim \nfrac{1}{\beta^{1/2}}$, then for any $p_0 \in \R^d$ and pair $(x_0, \bar x_0) \in \R^d \times \R^d$,
    \begin{align*}
        \norm{\nabla_{p_0}[\Psi_h(\bar x_0, \Phi_h^{\msx}(x_0, p_0))] - I} \lesssim \beta h^2\,.
    \end{align*}
\end{lemma}
\begin{proof}
    Fix $\bar x_0, x_0 \in\R^d$. Define
    \begin{align*}
        J(\cdot) \deq \nabla_2 \Phi_h^{\msx}(x_0, \cdot)\,, \qquad\text{and}\qquad K(\cdot) \deq \nabla_2 \Psi_h(\bar x_0, \cdot) = [(\nabla_2 \Phi_h^{\msx})(\bar x_0, \Psi_h(\bar x_0, \cdot))]^{-1}\,,
    \end{align*}
    where $\nabla_2$ denotes the gradient w.r.t.\ the second argument, and the second identity is via the implicit function theorem.

    If we define $x_t \deq \Phi_t^\msx(x_0, p_0)$ and $S_t \deq \nabla_{p_0} \Phi_t^{\msx}(x_0, p_0)$, then $S_t$ solves
    \begin{align}\label{eq:sde-stability-autonomous}
        \partial_t^2 S_t + \nabla^2 V(x_t)\, S_t = 0\,, \qquad S_0 = 0\,, \qquad \partial_t S_t\big|_{t = 0} = I\,, 
    \end{align}
    and so
    \begin{align*}
        S_t = tI - \int_0^t (t-s)\, \nabla^2 V(x_s)\, S_s \, \D s\,.
    \end{align*}
    If $\bar S_t \deq S_t - t I$, then
    \begin{align*}
        \bar S_t = -\int_0^t (t-s)\, \nabla^2 V(x_s)\, [s I + \bar S_s] \, \D s\,.
    \end{align*}
    This implies
    \begin{align*}
        \norm{\bar S_t}_{\rm op} \leq \beta \int_0^t (t-s)\, s \, \D s + \beta \int_0^t (t-s)\, \norm{\bar S_s}_{\rm op} \, \D s\,.
    \end{align*}
    Gr\"onwall's inequality then gives
    \begin{align*}
        \norm{\bar S_t}_{\rm op} \leq \frac{\beta t^3}{6}\, e^{\frac{\beta t^2}{2}}\,.
    \end{align*}
    This implies at $t = h$ that $\norm{J(\cdot) - hI}_{\rm op} \lesssim \beta h^3$ for $h \lesssim \nfrac{1}{\beta^{1/2}}$.
    This also implies that for some matrix $S$ with $\norm{S}_{\rm op} \lesssim \beta h^3$,
    \begin{align*}
        K(\cdot) = h^{-1}\, (I + h^{-1} S)^{-1}\,,
    \end{align*}
    so that
    \begin{align*}
        \norm{K(\cdot) - h^{-1}I}_{\rm op} \leq \frac{\norm{h^{-1} S}_{\rm op}}{h\,(1-\norm{h^{-1} S}_{\rm op})} \lesssim \beta h\,,
    \end{align*}
    so long as $h \lesssim \nfrac{1}{\beta^{1/2}}$. Now, we have
    \begin{align*}
        K(\Phi_h^{\msx}(x_0, \cdot))\, J(\cdot) = (h^{-1} I + O(\beta h)) \cdot (hI + O(\beta h^3)) = I + O(\beta h^2)\,,
    \end{align*}
    where the $O(\cdot)$-notation means there is a remainder term with operator norm bounded (up to absolute constants) by the indicated quantity.
    This bounds the desired term, by the chain rule.
\end{proof}

We also pause to justify the existence of $\Psi_h$.

\begin{lemma}[Inverse function theorem]\label{lem:inv_fn}
    Suppose that a map $F : \R^d\to\R^d$ satisfies the condition $\norm{\nabla F - cI}_{\rm op} < 1$, for some $c > 0$.
    Then, $F$ is invertible.
\end{lemma}
\begin{proof}
    This follows from the inverse function theorem, see~\cite[Proposition 3.2]{Tay23PDEI}.
\end{proof}

In the proof above, we have shown that for every $x_0\in\R^d$, $\Phi_h^\msx(x_0,\cdot)$ satisfies the condition of this lemma as long as $h \ll \nfrac{1}{\beta^{1/2}}$, hence $\Psi_h$ is well-defined.

\paragraph{Bound with Hessian Lipschitzness.}
\begin{lemma}[Jacobian bound under higher-order smoothness]\label{lem:hmc-distortion-II}
    Suppose Assumption~\ref{as:hess-lip} holds.     If $h \lesssim \nfrac{1}{\beta^{1/2}}$, then for any $p_0 \in \R^d$ and pair $(x_0, \bar x_0) \in \R^d \times \R^d$,
    \begin{align*}
        \norm{\nabla_{p_0} [\Psi_h(\bar x_0, \Phi_h^{\msx}(x_0, p_0))] - I} \lesssim \betH h^2\, \norm{\bar x_0 - x_0}\,.
    \end{align*}
\end{lemma}
\begin{proof}
    Recall from Lemma~\ref{lem:hmc-distortion} the definitions of $J(\cdot)$ and $K(\cdot)$, and furthermore that
    \begin{align*}
        K(\Phi_h^{\msx}(x_0, \cdot))\, J(\cdot) - I = K(\Phi_h^{\msx}(x_0, \cdot))\, \bigl(J(\cdot) - (\nabla_2 \Phi_h^{\msx})(\bar x_0, \Psi_h(\bar x_0, \Phi_h^{\msx}(x_0, \cdot)))\bigr)\,,
    \end{align*}
    by the second identity for $K(\cdot)$. Then, note that
    \begin{align*}
        J(\cdot) - (\nabla_2 \Phi_h^{\msx})(\bar x_0, \Psi_h(\bar x_0, \Phi_h^{\msx}(x_0, \cdot))) = \nabla_2 \Phi_h^{\msx}(x_0, \cdot) - (\nabla_2 \Phi_h^{\msx})(\bar x_0, \Psi_h(\bar x_0, \Phi_h^{\msx}(x_0, \cdot)))\,.
    \end{align*}
    Recalling the representation for these Jacobians derived in Lemma~\ref{lem:hmc-distortion}, we obtain for $x_s(p) \deq \Phi_s^\msx(x_0, p)$, $\bar x_s(p) \deq \Phi_s^\msx(\bar x_0, p)$ the two $x$-marginals of trajectories following the Hamiltonian flow~\eqref{eq:ham_ode} started at $(x_0, p)$ and $(\bar x_0, p)$ respectively,
    \begin{align*}
        \norm{\nabla_2 \Phi_h^{\msx}(x_0, p_0) - \nabla_2 \Phi_h^{\msx}(\bar x_0, \bar p_0)}_{\rm op}
        &= \norm{\nabla x_h(p_0) - \nabla \bar x_h(\bar p_0)}_{\rm op} \\
        &\le \int_0^h (h-s)\, \norm{\nabla^2 V(\bar x_s(\bar p_0)) - \nabla^2 V(x_s(p_0))}_{\rm op}\, \norm{\nabla x_s(p_0)}_{\rm op}\, \D s \\
        &\qquad + \int_0^h (h-s)\, \norm{\nabla^2 V(\bar x_s(\bar p_0))}_{\rm op}\, \norm{\nabla x_s(p_0) - \nabla \bar x_s(\bar p_0)}_{\rm op}\, \D s\,.
    \end{align*}
    The operator norm bounds derived in Lemma~\ref{lem:hmc-distortion} give $\norm{\nabla x_s(p_0) - sI}_{\rm op} \lesssim \beta h^3$. 
    Then, using Gr\"onwall's inequality, together with Lemma~\ref{lem:hmc-expansion} for $h \lesssim \nfrac{1}{\beta^{1/2}}$,
    \begin{align*}
        \norm{\nabla x_t(p_0) - \nabla \bar x_t(\bar p_0)}_{\rm op} &\lesssim \betH \int_0^t (t-s)\, s\,(\norm{x_0 - \bar x_0} + s\,\norm{p_0 - \bar p_0}) \, \D s \\
        &\qquad + \beta \int_0^t (t-s)\, \norm{\nabla x_s(p_0) - \nabla \bar x_s(\bar p_0)}_{\rm op} \, \D s \\
        &\lesssim \betH\,(t^3\,\norm{x_0 - \bar x_0} + t^4\, \norm{p_0 - \bar p_0})\,.
    \end{align*}
    Since we want $\nabla_2 \Phi_h^{\msx}(x_0, \cdot) - (\nabla_2 \Phi_h^{\msx})(\bar x_0, \Psi_h(\bar x_0, \Phi_h^{\msx}(x_0, \cdot)))$, we note that $\norm{\Psi_h(\bar x_0, \Phi_h^{\msx}(x_0, \cdot)) - \operatorname{id}}_{\sup} \lesssim \frac{1}{h}\, \norm{x_0 - \bar x_0}$ from Lemma~\ref{lem:hmc-expansion}, so that the $\norm{p_0 - \bar p_0}$ term above is of the same order as the $\norm{x_0 - \bar x_0}$ term. Lastly, we can bound the norm of $K(\cdot)$ by $h^{-1}$ as in Lemma~\ref{lem:hmc-distortion}. This all implies the result.
\end{proof}

\subsubsection{Proof of Theorem~\ref{thm:harnack}}

Next, we use the shifted composition framework~\cite{scr1, scr2, scr3, scr4} to boost the one-step coupling bound (Theorem~\ref{thm:regularity}) into a long-time regularity result (Theorem~\ref{thm:harnack}).
In particular, since the OHO dynamics exhibit hypocoercive convergence similarly to the underdamped Langevin diffusion (ULD), we leverage the framework developed for ULD in~\cite{scr4}.

Recall that our goal is to bound $\Ren_q(\delta_{x_0,p_0} \bs P_h^N \mmid \delta_{\bar x_0,\bar p_0}\bs P_h^N)$, where $\bs P_h$ is the OHO kernel. Let $(X_{nh}, P_{nh})_{n\ge 0}$ denote the iterates of OHO started at $(x_0, p_0)$, so that $(X_{nh}, P_{nh}) \sim \delta_{x_0,p_0} \bs P_h^n$.
Similarly, let $(\bar X_{nh}, \bar P_{nh})_{n\ge 0}$ denote the corresponding OHO iterates, but started at $(\bar x_0,\bar p_0)$.
We couple these two processes synchronously by using the same Gaussians during the~\eqref{eq:O} steps.

\begin{figure}
    \centering
    \includegraphics[width=0.97\linewidth]{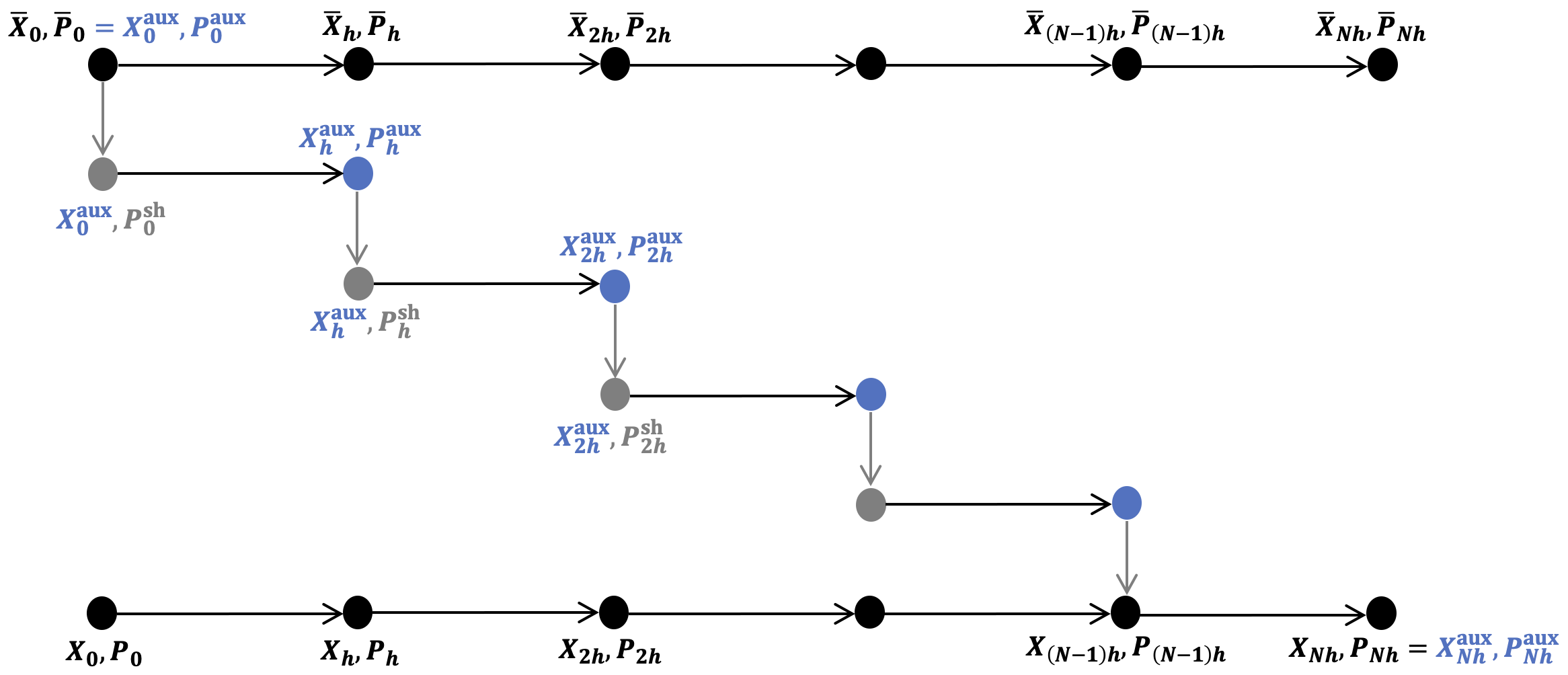}
    \caption{\footnotesize Schematic diagram for the proof of Theorem~\ref{thm:harnack}. Top and bottom (in black): the two processes $(\bar X_{kh}, \bar P_{kh})$ and $(X_{kh},P_{kh})$ that we seek to compare. Each horizontal arrow is an OHO step. Diagonal (in blue): the auxiliary process $(X_{kh}^{\aux},P_{kh}^{\aux})$ we construct to apply the shifted composition rule. The auxiliary process interpolates between the original two processes at initialization (top left) and termination (bottom right). The auxiliary process updates by first shifting the momentum (grey, downwards arrow) and then performing an OHO step (black, right arrow).}
    \label{fig:aux}
\end{figure}

To apply the shifted composition rule, we define an auxiliary third process $(X_{nh}^\aux, P_{nh}^\aux)_{n \ge 0}$ that interpolates between the two processes at initialization $(X_{0}^\aux, P_{0}^\aux) = (\bar X_{0}, \bar P_{0})$ and termination $(X_{Nh}^\aux, P_{Nh}^\aux) = (X_{Nh}, P_{Nh})$. See Figure~\ref{fig:aux}. The termination condition ensures that
\begin{align*}
    \Ren_q(\delta_{x_0,p_0} \bs P_h^N \mmid \delta_{\bar x_0,\bar p_0}\bs P_h^N)
    &= \Ren_q(\law(X_{Nh}^\aux, P_{Nh}^\aux) \mmid \law(\bar X_{Nh}, \bar P_{Nh}))\,.
\end{align*}
To bound the right-hand side, we use the shifted composition rule (Theorem~\ref{thm:shifted_chain_rule}), for which we need the one-shot bound in Theorem~\ref{thm:regularity}.
Then, the crux of the argument amounts to showing contraction between $(X_{nh}^\aux, P_{nh}^\aux)_{n\ge 0}$ and $(X_{nh}, P_{nh})_{n\ge 0}$, which ends up being a more difficult version of OHO contraction (Lemma~\ref{lem:oho_contraction}) due to the additional dynamics introduced in the auxiliary process to ensure the termination criterion. 

We describe the iterative construction of the auxiliary process.
First, update $P_{kh}^\aux$ to $P_{kh}^\sh$ as
\begin{align*}
    P_{kh}^\sh \deq P_{kh}^\aux - (1-e^{-\upeta_k^{\msp}})\, (P_{kh}^\aux - P_{kh}) -\upeta_k^{\msx}\, (X_{kh}^\aux - X_{kh})\,,
\end{align*}
for two sequences $(\upeta_k^{\msx})_{k \in [N]}$, $(\upeta_k^{\msp})_{k \in [N]}$. We define these as
\begin{align*}
    \upeta_k^{\msx} \deq \int_{kh}^{(k+1)h} \eta_t^{\msx} \, \D t\,, \quad \upeta_k^{\msp} \deq \int_{kh}^{(k+1)h} \eta_t^{\msp} \, \D t \,, \quad \eta_t^{\msx} \deq \frac{\gamma_t \eta_t^{\msp}}{2}\,, \quad  \eta_t^{\msp} \deq \frac{c_0 \omega}{\exp(\omega(Nh-t + Ah)) - 1}\,,
\end{align*}
for some large absolute constant $A$. Here, $\omega$ is defined in~\eqref{eq:omeg}, and $\gamma_t \deq \gamma + \eta_t^\msp$ is the ``effective friction'' as introduced in~\cite{scr4}.

Then, produce $(X_{(k+1)h}^\aux, P_{(k+1)h}^\aux)$ by running OHO with $(X_{kh}^\aux, P_{kh}^\sh)$ as a starting point.
For the final iteration, we simply set $(X_{Nh}^\aux, P_{Nh}^\aux) \deq (X_{Nh}, P_{Nh})$ by fiat.

We first prove the following lemma, which establishes contraction of the auxiliary process and the original process, where the distance is measured in a time-varying coordinate system making use of the effective friction:
\begin{align*}
    (\mathtt d_n^\aux)^2
    &\deq\norm{X_{nh}- X_{nh}^{\aux}}^2 + \bigl\lVert X_{nh} - X_{nh}^{\aux} + \frac{2}{\gamma_{nh}}\, (P_{nh} - P_{nh}^\aux)\bigr\rVert^2\,.
\end{align*}

\begin{lemma}[Auxiliary distance recursion]\label{lem:shift-contraction}
    Suppose that $0 \prec \alpha I \preceq \nabla^2 V \preceq \beta I$ (strongly convex and smooth potential) and $\gamma \geq \sqrt{32 \beta}$ (high friction). Then for any step size $h \ll \nfrac{1}{\kappa^{1/2} \gamma}$, 
    we have the following contraction for some absolute constant $c' > 0$:
    \begin{align*}
       (\mathtt d_{(k+1)}^\aux)^2
        &\le \exp\Bigl(- c' \int_{kh}^{(k+1)h}(\omega_+ + \eta_t^{\msp}) \, \D t\Bigr)\, (\mathtt d_k^\aux)^2\,.
    \end{align*}
\end{lemma}

\begin{proof}
Adapting the proof of Lemma~\ref{lem:oho_contraction}, we use the following time-varying Lyapunov functional, defined for $t \in [kh, (k+1)h)$ as
\begin{align*}
    \mc L_t( \delta X_t^{\aux, \msOH}, \delta P_t^{\aux, \msOH}) \deq \norm{\delta X_t^{\aux,\msOH}}^2 + \norm{a_t \delta X_t^{\aux,\msOH} + \frac{2c_t}{\gamma_t} \delta P_t^{\aux,\msOH} }^2\,.
\end{align*}
Here we use the shorthand 
\begin{align*}
    a_t \deq 1 + \frac{2}{\gamma_t}\, e^{\upeta_k^\msp}\int_{t}^{(k+1)h}\eta_s^{\msx} \, \D s\,, \qquad c_t \deq e^{-\frac{\gamma h}{2} + \int_t^{(k+1)h} \gamma_s \, \D s}\,.
\end{align*}
We also use the shorthand $\delta X_t^{\aux, \msOH} \deq X_t^{\aux, \msOH} - X_t^\msOH$ and similarly for $\delta P_t^{\aux, \msOH}$, where the two processes $(X_t^\msOH, P_t^\msOH)_{kh \le t < (k+1)h}$ and $(X_t^{\aux,\msOH}, P_t^{\aux,\msOH})_{kh\le t < (k+1)h}$ both evolve according to the Hamiltonian flow~\eqref{eq:ham_ode} but from respective initializations $(X_{kh}^\msO, P_{kh}^\msO)$ and $(X_{kh}^{\aux,\msO}, P_{kh}^{\sh,\msO})$.

It is convenient to express this Lyapunov functional via a change of coordinates. Define  
\begin{align*}
    Z_t^{\msOH} \deq a_t \, X_t^{\msOH} + \frac{2c_t}{\gamma_t} P_t^{\msOH}\,,
\end{align*}
and similarly for $Z_t^{\aux,\msOH}$ and $\delta Z_t^{\aux,\msOH}$.

    We now aim to show that the Lyapunov functional decays in time.
    Differentiating,
    \begin{align*}
        \partial_t \delta X_t^{\aux, \msOH} &= \delta P_t^{\aux, \msOH} = \frac{\gamma_t}{2c_t}\, \bigl(\delta Z_t^{\aux,\msOH} - a_t \delta X_t^{\aux,\msOH} \bigr)\,, \\
        \partial_t \delta P_t^{\aux, \msOH} &= -H_t\, \delta X_t^{\aux, \msOH}\,, \\
        \partial_t \delta Z_t^{\aux, \msOH} &= \frac{\gamma_t a_t}{2c_t}\, \bigl(\delta Z_t^{\aux,\msOH} - a_t \delta X_t^{\aux,\msOH} \bigr)  + \dot{a}_t \delta X_t^{\aux, \msOH}- \frac{2c_t}{\gamma_t} H_t\, \delta X_t^{\aux, \msOH} \\
        &\qquad\qquad+ \Bigl(\frac{2\dot{c}_t}{\gamma_t} - \frac{2  \dot{\gamma}_t c_t}{\gamma_t^2} \Bigr)\, \frac{\gamma_t}{2c_t}\, \bigl(\delta Z_t^{\aux, \msOH} - a_t \delta X_t^{\aux, \msOH} \bigr)\,.
    \end{align*}
    for a suitable averaged Hessian $H_t$ as in the proof of Lemma~\ref{lem:oho_contraction}. A calculation gives
    \begin{align*}
        \dot{c}_t = -c_t\gamma_t\,, \qquad \dot{a}_t = -\frac{2}{\gamma_t}\, e^{\upeta_k^\msp}\,\Bigl(\eta_t^{\msx} + \frac{\dot{\gamma}_t}{\gamma_t} \int_t^{(k+1)h} \eta_s^{\msx} \, \D s \Bigr)\,, \qquad \dot{\gamma}_t = \omega \eta_t^{\msp} + \frac{(\eta_t^{\msp})^2}{c_0}\,.
    \end{align*}
    Thus
    \begin{align*}
        &\partial_t \mc L_t( \delta X_t^{\aux, \msOH}, \delta P_t^{\aux, \msOH}) =  -2\begin{bmatrix}
        \delta X_t^{\aux, \msOH} \\
        \delta Z_t^{\aux, \msOH}
    \end{bmatrix}^\top \begin{bmatrix}
        \frac{\gamma_t a_t}{2c_t} & \frac{\gamma_t a_t^2}{4c_t} - \frac{\gamma_t}{4c_t} - \frac{\dot{a}_t}{2} - \frac{\gamma_t a_t}{2} - \frac{\dot{\gamma}_t a_t}{2\gamma_t} + \frac{c_t H_t}{\gamma_t}  \\[0.25em]
        * & \gamma_t -\frac{\gamma_t a_t}{2c_t} + \frac{\dot{\gamma}_t}{\gamma_t}
    \end{bmatrix}
    \begin{bmatrix}
        \delta X_t^{\aux, \msOH} \\
        \delta Z_t^{\aux, \msOH}
    \end{bmatrix}\,.
    \end{align*} 
    By the same diagonalization argument as in the proof of Lemma~\ref{lem:oho_contraction}, and by discarding the quantity $\dot \gamma_t / \gamma_t = \partial_t \log \gamma_t \geq 0$ which is non-negative since $\gamma_t$ is increasing, it therefore suffices to lower bound the minimal eigenvalue of the $2 \times 2$ matrix
    \begin{align*}
        \mc M^\lambda_t \deq \begin{bmatrix}
        \frac{\gamma_t a_t}{2c_t} & b_t^\lambda - \frac{\gamma_t}{2} \\[0.25em]
        * & \gamma_t -\frac{\gamma_t a_t}{2c_t}
    \end{bmatrix}\,,
    \end{align*}
    where we define
    \begin{align*}
        b_t^\lambda = \frac{\gamma_t}{4c_t}\,(a_t^2 - 1) - \frac{\dot{a}_t}{2} - \frac{\gamma_t\,(a_t-1)}{2} - \frac{\dot{\gamma}_t a_t}{2\gamma_t} + \frac{c_t \lambda}{\gamma_t}
    \end{align*}
   for $\lambda \in [\alpha,\beta]$. By an explicit calculation, the minimal eigenvalue of this $2\times2$ matrix is
    \begin{align*}
        \frac{\gamma_t}{2} - \sqrt{\frac{\gamma_t^2}{4}\,\bigl(\frac{a_t}{c_t} - 1 \bigr)^2 + \bigl(b_t^\lambda - \frac{\gamma_t}{2}\bigr)^2}\,.
    \end{align*}
    Now, it remains to show the following properties:
    \begin{align*}
        \omega_+ + \eta_t^\msp \lesssim b_t^\lambda \le  \frac{3\gamma_t}{4}\,, \qquad \frac{\gamma_t^2}{4} \,\Bigl(\frac{a_t}{c_t} - 1\Bigr)^2 \ll (\omega_+ + \eta_t^{\msp})\, \gamma_t\,,
    \end{align*}
    where $\ll$ means that the LHS is bounded by a sufficiently small implied constant.
    Indeed, if this holds, then this allows us to bound the minimum eigenvalue as follows:
    \begin{align*}
        \lambda_{\min}(\mc M_t^\lambda)
        &= \frac{\gamma_t}{2}\,\Bigl[1 - \sqrt{1 - \frac{4b_t^\lambda}{\gamma_t} + \frac{4(b_t^\lambda)^2}{\gamma_t^2} +\,\bigl(\frac{a_t}{c_t} - 1\bigr)^2}\Bigr]
        \ge \frac{\gamma_t}{2}\,\Bigl[1 - \sqrt{1 - \frac{b_t^\lambda}{\gamma_t} + o\bigl(\frac{b_t^\lambda}{\gamma_t}\bigr)}\Bigr]
        \gtrsim b_t^\lambda
        \gtrsim \omega_+ + \eta_t^\msp\,.
    \end{align*}
    This then produces the desired contraction.
    We now verify the two conditions above.
    In the following, we note that if $c_0/A \ll 1$ and $h \ll \nfrac{1}{\gamma} \wedge \nfrac{1}{A\abs{\omega}}$, then we can ensure $\eta_t^\msp h \le c$ for an absolute constant $c > 0$, which we can take to be as small as we like.
    In particular, we can take $a_t = 1 + O(\eta_t^\msp h) \le 1 + \frac{1}{40}$ and $c_t \in [1-\frac{1}{40}, 1+\frac{1}{40}]$, say. We organize the argument into three subclaims. 
    \begin{enumerate}[label=(\roman*)]
        \item Claim: $b_t^\lambda \leq \frac{3\gamma_t}{4}$. We bound each term in $b_t^{\lambda}$. The terms $-\frac{\gamma_t\,(a_t - 1)}{2} - \frac{\dot\gamma_t a_t}{2\gamma_t}$ are clearly strictly negative. Next, by the above bound on $c_t$ and the high-friction assumption $\gamma \ge \sqrt{32\beta}$, we have $c_t \lambda \leq \frac{\gamma^2}{12}$, hence $\frac{c_t \lambda}{\gamma_t} \leq \frac{\gamma_t}{12}$.
         Next, $\frac{\gamma_t}{4c_t}\, (a_t^2 - 1) \leq \frac{\gamma_t}{12}$. Lastly, as $\frac{\dot{\gamma}_t}{\gamma_t} \lesssim \eta_t^{\msp}$ and $h\eta_t^{\msp} \leq c \ll 1$, for small enough $c$ we have $-\dot{a}_t \leq \frac{8}{7}\, (\eta_t^{\msp} + O(h\, (\eta_t^{\msp})^2)) \leq \frac{7}{6}\, \eta_t^{\msp}$.
        Combining these bounds yields the claim.
  
       \item Claim: $\frac{\gamma_t}{4}\, \bigl(\frac{a_t}{c_t} - 1\bigr)^2 \ll \omega_+ + \eta_t^{\msp}$. Linearizing for $h \lesssim \nfrac{1}{\gamma}$, $a_t = 1 + O(\eta_t^{\msp} h)$, while $c_t^{-1} = 1 + O(\gamma_t h)$. Thus, we require
        \begin{align}\label{eq:gam_bd_blah}
            \gamma_t^3 h^2 \ll \omega_+ + \eta_t^{\msp}\,.
        \end{align}
        This is guaranteed by the step size assumption $h \ll \nfrac{1}{\kappa^{1/2} \gamma}$.
        
        \item Claim: $b_t^\lambda \gtrsim \omega_+ + \eta_t^{\msp}$. Discard the $\dot{\gamma}_t/{\gamma_t}$ term in $\dot{a}_t$. Now, we have
        \begin{align*}
            \frac{\gamma_t}{4c_t}\, (a_t^2 - 1) - \frac{\gamma_t\, (a_t  -1)}{2} &= \frac{\gamma_t}{4} \cdot \frac{4}{\gamma_t^2} \cdot e^{2\upeta_k^\msp}\, \Bigl(\int_t^{(k+1)h} \eta_s^{\msx} \, \D s\Bigr)^2 + \frac{\gamma_t\, (a_t^2 - 1)}{4}\, \Bigl(\frac{1}{c_t} - 1 \Bigr)
            \\
            &\gtrsim -\gamma_t \eta_t^{\msp} h \cdot \gamma_t h \gtrsim -\gamma_t^3 h^2\,.
        \end{align*}
        This is handled by~\eqref{eq:gam_bd_blah} above.
        As for the remaining terms, we want to show
        \begin{align*}
            \frac{e^{\upeta_k^\msp} \gamma_t \eta_t^{\msp}}{2} - \frac{\omega \eta_t^{\msp} a_t}{2} - \frac{(\eta_t^{\msp})^2 a_t}{2c_0} + c_t \alpha \gtrsim (\omega_+ + \eta_t^{\msp})\,\gamma_t\,.
        \end{align*}
        The left-hand side gives, bounding $a_t \leq \frac{6}{5}$,
        \begin{align*}
            \frac{e^{{\upeta_k^\msp}} \gamma_t \eta_t^{\msp}}{2} - \frac{\omega \eta_t^{\msp} a_t}{2} - \frac{(\eta_t^{\msp})^2 a_t}{2c_0} + c_t \alpha &\geq c_t \alpha + \Bigl(\frac{\gamma}{2} + \bigl(\frac{1}{2}-\frac{3}{5c_0}\bigr)\, \eta_t^{\msp} - \frac{3 \omega}{5}  \Bigr)\,\eta_t^{\msp} \\
            &\geq c_t \alpha + \frac{1}{4}\,(\gamma + \eta_t^{\msp} )\,\eta_t^{\msp} = c_t \alpha + \frac{\gamma_t \eta_t^{\msp}}{4} \,.
        \end{align*}
        Take $c_0$ large enough that $\frac{1}{2}-\frac{3}{5c_0} \ge \frac{1}{4}$ and absorb the $\omega$ term into the $\gamma$ term as $\gamma^2 \geq 6\alpha$. This shows the claim.

    \end{enumerate}
    
Having shown the decay of the Lyapunov functional $\mc L_t$, we now show that this implies the desired contraction of $\mathtt d^\aux$ after one iteration. 
Indeed, we have for some constant $c' > 0$
\begin{align}
    \mc L_{(k+1)h}\left(
    \delta X_{(k+1)h}^{\aux, \msOH}, \delta P_{(k+1)h}^{\aux, \msOH}
    \right) \leq \underset{\eqqcolon L < 1}{\underbrace{\exp\Bigl(-c' \int_{kh}^{(k+1)h} (\omega_+ + \eta_t^\msp) \, \D t\Bigr)} }\,\, \mc L_{kh}(\delta X_{kh}^{\aux, \msOH}, \delta P_{kh}^{\aux, \msOH})\,. 
    \label{eq-pf:contract-lyapunov}
\end{align}
This lets us establish the desired contraction:
\begin{align*}
    &\norm{\delta X_{(k+1)h}^{\aux,\msOHO}}^2 + \bigl\lVert \delta X_{(k+1)h}^{\aux,\msOHO} + \frac{2}{\gamma_{(k+1)h}}\,\delta P_{(k+1)h}^{\aux,\msOHO}\bigr\rVert^2 \\
    &\qquad = \norm{\delta X_{(k+1)h}^{\aux,\msOH}}^2 + \bigl\lVert \delta X_{(k+1)h}^{\aux,\msOH} + \frac{2}{\gamma_{(k+1)h}}\,e^{-\gamma h/2}\,\delta P_{(k+1)h}^{\aux,\msOH}\bigr\rVert^2 \\[0.25em]
    &\qquad \le L\,\Bigl[\norm{\delta X_{kh}^{\sh,\msO}}^2 + \Bigr\rVert \Bigl(1 + \frac{2}{\gamma_{kh}}\,e^{\upeta_k^\msp}\,\upeta_k^\msx\Bigr)\,\delta X_{kh}^{\sh,\msO} + \frac{2}{\gamma_{kh}} \,e^{-\gamma h/2 + \int_{kh}^{(k+1)h}\gamma_s\,\D s}\,\delta P_{kh}^{\sh,\msO} \Bigr\rVert^2\Bigr] \\[0.25em]
    &\qquad = L\,\Bigl[\norm{\delta X_{kh}^{\sh}}^2 + \Bigr\rVert \Bigl(1 + \frac{2}{\gamma_{kh}}\,e^{\upeta_k^\msp}\,\upeta_k^\msx\Bigr)\,\delta X_{kh}^{\sh} + \frac{2}{\gamma_{kh}} \,e^{-\gamma h + \int_{kh}^{(k+1)h}\gamma_s\,\D s}\,\delta P_{kh}^{\sh} \Bigr\rVert^2\Bigr] \\[0.25em]
    &\qquad = L\,\Bigl[\norm{\delta X_{kh}^{\aux}}^2 + \Bigr\rVert \Bigl(1 + \frac{2}{\gamma_{kh}}\,e^{\upeta_k^\msp}\,\upeta_k^\msx\Bigr)\,\delta X_{kh}^{\aux} + \frac{2}{\gamma_{kh}} \,e^{\upeta_k^\msp}\,\bigl(e^{-\upeta_k^\msp}\,\delta P_{kh}^{\aux} - \upeta_k^\msx\,\delta X_{kh}^\aux\bigr) \Bigr\rVert^2\Bigr] \\[0.25em]
    &\qquad = L\,\Bigl[\norm{\delta X_{kh}^{\aux}}^2 + \bigr\rVert \delta X_{kh}^{\aux} + \frac{2}{\gamma_{kh}} \,\delta P_{kh}^{\aux} \bigr\rVert^2\Bigr]\,.
\end{align*}
Above, the first and third step are by definition of the~\eqref{eq:O} process, the fourth step is by definition of the shift process, and the final step is by simplifying. The key second step (the only inequality) is by the decay shown in~\eqref{eq-pf:contract-lyapunov}. 
\end{proof}

Lemma~\ref{lem:shift-contraction} gives the contraction across iterations. Introduce the distances
    \begin{align*}
        d_n^\msx = \norm{X_{nh} - X_{nh}^\aux}_{L^\infty}\,, \qquad d_{nh}^\msp =\norm{P_{nh} - P_{nh}^\aux}_{L^\infty}\,.
    \end{align*}
    Recalling the identities from~\cite[Lemma 4.9]{scr4},
    \begin{align*}
        \mathtt d_n^\aux \asymp d_n^\msx + \frac{1}{\gamma_{nh}}\, d_n^\msp\,.
    \end{align*}
    Now, we combine this with Theorem~\ref{thm:regularity} to prove Theorem~\ref{thm:harnack}.

\begin{proof}[{Proof of Theorem~\ref{thm:harnack}}]
    We use the following notation:
    \begin{align*}
        \bs\mu_{nh} \deq \law(X_{nh}, P_{nh})\,, \qquad \bar{\bs \mu}_{nh} \deq \law(\bar X_{nh}, \bar P_{nh})\,, \qquad \bs\mu_{nh}^\aux \deq \law(X_{nh}^\aux, P_{nh}^\aux)\,.
    \end{align*}
    We split the proof into two cases.
    
    \textbf{Without Hessian Lipschitzness.}
    We handle the final step via a separate application of the shifted composition rule (Theorem~\ref{thm:shifted_chain_rule}). By Theorem~\ref{thm:regularity}, this yields
    \begin{align*}
        \Renyi_q(\bs \mu_{Nh} \mmid \bar{\bs \mu}_{Nh}) - \Renyi_q(\bs \mu_{(N-1)h}^\aux \mmid \bar{\bs \mu}_{(N-1)h}) &\lesssim  \frac{q\,(d_{N-1}^\msx)^2}{\gamma h^3} + \frac{q\,(d_{N-1}^\msp)^2}{\gamma h} + \beta^2 d h^4 q \\
        &\lesssim \frac{q\,(\mathtt d_{N-1}^\aux)^2}{\gamma h^3} + \beta^2 d h^4 q\,.
    \end{align*}
    Now, for the prior iterations, noting that we add a shift $-(1-e^{-\upeta_n^\msp})\, d_n^\msp - \upeta_n^\msx\, d_n^\msx$ to the momentum coordinate in each iteration and $\upeta_n^\msp \lesssim 1$, then Theorem~\ref{thm:regularity} and the shifted composition rule (Theorem~\ref{thm:shifted_chain_rule}) together give
    \begin{align*}
        \Renyi_q(\bs \mu_{(N-1)h}^\aux \mmid \bar{\bs \mu}_{(N-1)h}) &\lesssim \frac{q}{h} \sum_{n=0}^{N-2} (\upeta_n^\msx d_n^\msx + \upeta_n^\msp d_n^\msp)^2 + \beta^2 d h^4 qN \\
        &\lesssim q\sum_{n=0}^{N-2} h\gamma_{nh}^2\, (\eta_{nh}^\msp)^2\, (\mathtt d_n^\aux)^2 + \beta^2 dh^4 qN\,.
    \end{align*}
    The additional regularity terms must simply be summed along the flow. Now, note that the distance recursion, which relates ${\mathtt d}_n^\aux$ to its initial value, is precisely that of~\cite[Section 4]{scr4}, up to absolute constants. This implies that
    \begin{align}\label{eq:aux-dist-sum}
        \sum_{n=0}^{N-2} h\gamma_{nh}^2\, (\eta_{nh}^\msp)^2\, (\mathtt d_n^\aux)^2 \lesssim C(\alpha, \beta, \gamma, Nh)\, \Bigl\{\norm{x_0 - \bar x_0}^2 + \frac{1}{\gamma_0^2}\, \norm{p_0 - \bar p_0}^2 \Bigr\}\,.
    \end{align}
    Secondly, as in the proof of~\cite[Theorem 4.1]{scr4}, it was shown that an analogous term to $\frac{(\mathtt d_{N-1}^\aux)^2}{\gamma h^3}$ is bounded by the same right hand side in~\eqref{eq:aux-dist-sum} if $c'c_0 \geq 3$. Our shift sequence and recursion also validate the hypotheses of that theorem, and so a similar bound is satisfied (up to absolute constants). Putting this all together gives the theorem statement.

    \textbf{With Hessian Lipschitzness.} This is a straightforward adaptation of the previous argument; with $h \lesssim \nfrac{1}{\betH^{2/7} \gamma^{1/7} d^{1/7}}$, we have an identical regularity condition with the exception of the additive term $\beta^2 d h^4$ not being present. This yields the same bound but without the additive term.
\end{proof}
		\section{Discretization analysis}\label{sec:app}

In this section we establish our algorithmic warm start result (Theorem~\ref{thm:main-warm}). Below we first state our results about the proposed discretization scheme (in \S\ref{ssec:obabco}) and the two-stage analysis (in \S\ref{ssec:two-stage}), and then we provide all proofs of this section (in \S\ref{ssec:6:proofs}). See \S\ref{ssec:obabco_overview} for a high-level overview.

\subsection{OBABCO}\label{ssec:obabco}

Our method of producing a warm start is based on a variant of the standard OBABO splitting scheme (see the preliminaries section \S\ref{sec:hmc_bg} for background on OBABO).
As discussed below, OBABO does not admit the requisite error bounds, so we introduce a final corrective step~\eqref{eq:C}, leading to the proposed OBABCO scheme.

Formally, OBABCO replaces the~\eqref{eq:H} step in OHO by
\begin{align}
    P^{\msOB} &\deq P^\msO - \frac{h}{2}\, \nabla V(X^\msO) \tag{B}\label{eq:b-splitting}\,, \\
    X^{\msOBA} &\deq X^\msO + h P^{\msOB}\,, \tag{A}\label{eq:a-splitting}\\
    P^{\msOBAB} &\deq P^{\msOB} - \frac{h}{2}\, \nabla V(X^{\msOBA})\,, \tag{B}\, \\
    X^{\msOBABC} &\deq X^\msO + \frac{h}{3}\,(P^\msOBAB + 2 P^{\msOB}) + \frac{h^2}{6}\, \nabla V(X^\msO) \tag{C}\label{eq:C}\,.
\end{align}
Our final iterate will be $(X^\msOBABC, P^\msOBAB)$, to which we then apply~\eqref{eq:O} in the momentum coordinate.

The justification for this scheme is based on the computation of the strong errors below.

\begin{lemma}[Strong errors]\label{lem:errors}
    By coupling the Gaussian noise in the~\eqref{eq:O} steps, we have almost surely
    \begin{align*}
        h^{-1}\,\norm{X^{\msOBABCO} - X^{\msOHO}} + \norm{P^{\msOBABCO} - P^{\msOHO}}\lesssim \beta h^2\,(\norm{p} + h\, \norm{\nabla V(x)} + \sqrt{\gamma dh})\,.
    \end{align*}
    If additionally $V$ is three times continuous differentiable, then
    \begin{align*}
        h^{-1}\,&\norm{X^{\msOBABCO} - X^{\msOHO}} + \norm{P^{\msOBABCO} - P^{\msOHO}} \\ &\lesssim \beta h^3\,(\norm{\nabla V(x)} + \beta h\, \norm{p} + \beta h \sqrt{\gamma d h}) \\
        &\qquad +  h^2 \int_0^h \bigl(\norm{\nabla^3 V(x+tP^{\msOB})[P^{\msOB}, P^{\msOB}]} + \norm{\nabla^3 V(X_t^\msOH)[P_t^\msOH, P_t^\msOH]} \bigr) \, \D t\,.
    \end{align*}
\end{lemma}

To see why this computation leads us to expect that OBABCO can produce a warm start in $d^{1/4}$ steps, note that under Assumption~\ref{as:higher-reg}, for $(X,P) \sim \bs\pi$, $\norm{\nabla^3 V(X)[P, P]} \lesssim d^{1/2}$ with high probability (see Lemma~\ref{lem:chaos}).
Thus, at stationarity, the strong error scales as $O(d^{1/2} h^3)$, and if we sum up this error over $O(1/h)$ iterations, it leads to a cumulative error scaling as $O(d^{1/2} h^2)$.
Choosing step size $h \asymp \nfrac{1}{d^{1/4}}$ leads to a final error of at most $O(1)$, which should be sufficient for a warm start, and leads to a complexity of roughly $\nfrac{1}{h} \asymp d^{1/4}$ iterations. 

There are two main obstacles with this informal reasoning:
\begin{enumerate}
    \item On its own, the strong error would only control the $L^2$ distance between the trajectories of OBABCO and OHO\@, whereas for a warm start we require \emph{R\'enyi} bounds.
    \item The calculation above only holds at stationarity, but we will need strong error bounds even when the algorithm starts quite far away from stationarity.
\end{enumerate}
The first obstacle will be handled by extending the ideas of \S\ref{ssec:harnack}, that is, via one-shot coupling and a ``discrete Girsanov'' argument.
As for the second obstacle, we apply change of measure of arguments, as described in the next subsection.
This is the reason why we need the subtle high-probability bound appearing in Lemma~\ref{lem:chaos}.

\begin{remark}[Necessity of the~\eqref{eq:C} step]
    We introduce the correction step~\eqref{eq:C} so that in the strong error, the position error is a factor $h$ smaller than the momentum error. This is required later for our R\'enyi analysis to obtain a final iteration complexity of $d^{1/4}$. It is unclear if this algorithmic modification is strictly necessary.~\cite{monmarche2021high} provides a $d^{1/4}$ result without a correction step, but only in the weaker metric of $W_2$ and under the additional assumption of a separable target. 
\end{remark}

\subsection{Two-stage analysis}\label{ssec:two-stage}

Our proof relies on a two-stage analysis, which we now motivate. The refined bounds in our proof will rely on the higher-order Taylor expansions in Lemma~\ref{lem:errors}, which will require us to evaluate the expectations of quantities of the form $\exp(C\, \norm{\nabla^3 V(X)[P, P]}^2)$ (conditioned on a good event) for some constant $C$, with $(X, P) \sim \bs \mu_0 \bs P_h^n$, where $\bs P_h$ is the OHO kernel. Under the stationary measure $\bs \pi$, $X$ and $P$ are independent, and we can directly appeal to Lemma~\ref{lem:chaos} in order to deduce high probability bounds of the correct order.

However, since we need the expectations under $\bs \mu_0 \bs P_h^n$ instead, it is imperative that these high probability bounds also hold for that measure. To do this, we want to apply Lemma~\ref{lem:high-prob-change} which bounds the expectations under a change of measure. However, this only produces bounds of the correct order when $\Renyi_2(\bs \mu_0 \mmid \bs \pi) \lesssim d^{1/2}$, which does not occur at initialization as we expect this quantity to be of order $d$ (Lemma~\ref{lem:initialization}). Thus we conduct a two-phase analysis. 

In the first phase, we show that
by using the lower-order Taylor expansion in Lemma~\ref{lem:errors} (which does not use Lipschitzness of the Hessian), we can reach R\'enyi divergence of order $d^{1/2}$. Crucially, this holds even if we take a step size of order $d^{-1/4}$.
Then, in the second phase, we use the change-of-measure principle (Lemma~\ref{lem:high-prob-change}) to drive the R\'enyi divergence to constant order, again with step size of order $d^{-1/4}$.

For the first phase, we establish the following bound.
Below, we let $\widehat{\bs P}_h$ denote the OBABCO kernel (and recall that $\bs P_h$ denotes the OHO kernel).

\begin{theorem}[Coarse R\'enyi bound]\label{thm:coarse}
    Assume only that $-\beta I \preceq \nabla^2 V \preceq \beta I$, and that $h \ll \nfrac{1}{\beta^{1/2} q^{1/2}} \wedge \nfrac{\gamma^{1/2}}{\beta q T^{1/2}} \wedge \nfrac{1}{\gamma}$.
    Then, for $T \deq Nh$,
\begin{align*}
    \Renyi_q(\bs\mu_0 \widehat{\bs P}_h^N \mmid \bs\mu_0 \bs P_h^N) \lesssim \frac{\beta^2 h^2 q T}{\gamma}\, \bigl(d + \Renyi_2(\bs \mu_0 \mmid \bs \pi)\bigr)\,.
\end{align*}
\end{theorem}

\begin{remark}
    Since this result does not use the higher-order Taylor expansion, it also holds for OBABO without the~\eqref{eq:C} step.
    Indeed, it is easy to check that OBABO satisfies the first inequality in Lemma~\ref{lem:errors}.
    When combined with Theorem~\ref{thm:harnack} and the weak triangle inequality for R\'enyi divergences (Proposition~\ref{prop:renyi}), this implies that OBABO with $\gamma \asymp \sqrt \beta$ and initialized at a measure $\bs\mu_0$ with $\Renyi_2(\bs\mu_0\mmid\bs\pi) = \widetilde O(d)$ can reach $\Ren_2(\bs\mu_0\widehat{\bs P}_h^N \mmid \bs\pi) \le \varepsilon^2$ in $\widetilde O(\kappa^{3/2} d^{1/2}/\varepsilon)$ iterations.

    In comparison, under similar assumptions, the result of~\cite[Proposition 11]{monmarche2021high} implies an iteration complexity of roughly $\widetilde O(\kappa^2 d^{1/2}/\varepsilon)$ to reach $\varepsilon$ error in Wasserstein distance, or $\widetilde O(\kappa^2 d/\varepsilon)$ to reach $\varepsilon$ error in total variation distance.
    Hence, we see that our result yields guarantees in a stronger metric and with smaller iteration complexity.
\end{remark}

Next, we state our refined R\'enyi bound.

\begin{theorem}[Refined R\'enyi bound]\label{thm:obabco-final-renyi}
    Assume that $-\beta I \preceq \nabla^2 V \preceq \beta I$ and that Assumption~\ref{as:higher-reg} holds.
    Also, suppose that the following step size condition holds:
    \begin{align}\label{eq:big_stepsize_cond}
        h \ll \frac{1}{\gamma} \wedge \frac{\gamma}{\beta q} \wedge \frac{1}{\beta^{1/2} q^{1/2}} \wedge \frac{\gamma^{1/4}}{(\beta^{3/4} +\beta_\eff^{1/2})\, q^{1/2} T^{1/4}} \wedge \frac{1}{\beta_\eff^{2/7} \gamma^{1/7} d^{1/7}} \wedge  \frac{1}{\beta^{1/2} d^{1/4}} \wedge \frac{\gamma^{1/2}}{\beta d^{1/4} q T^{1/2}}\,.
    \end{align}
    Then, it holds that
    \begin{align*}
        \Renyi_q(\bs\mu_0 \widehat{\bs P}_h^N \mmid \bs\mu_0 \bs P_h^N)
        &\lesssim \frac{\beta^3 h^4 q T}{\gamma}\,(d+R_0) + \frac{\beta_\eff^2 h^4 qT}{\gamma}\,\Bigl(d+R_0^2 + \log^2 \frac{\gamma}{\beta_\eff^2 dh^4 q^2 T}\Bigr)\,,
    \end{align*}
    where $R_0 \deq \Renyi_2(\bs\mu_0\mmid\bs\pi)$.
\end{theorem}

To interpret these results, we focus on the dependence on the dimension.
Note that due to the presence of the $R_0^2$ term, Theorem~\ref{thm:obabco-final-renyi} only yields an iteration complexity bound scaling as $d^{1/4}$ if $R_0 \lesssim d^{1/2}$.
However, Theorem~\ref{thm:coarse} shows that with $h \asymp d^{-1/4}$ we can indeed reach a R\'enyi divergence of order $d^{1/2}$.
Therefore, by combining together Theorems~\ref{thm:coarse} and~\ref{thm:obabco-final-renyi} in a two-phase analysis, we can establish our main result (Theorem~\ref{thm:main-warm}).
We give a formal proof in \S\ref{sssec:main_warm_pf}.
	    \subsection{Proofs}\label{ssec:6:proofs}

\subsubsection{Proof of Lemma~\ref{lem:errors}}

\begin{proof}[Proof of Lemma~\ref{lem:errors}]
    We note that since $\D X_t^\msOH = P_t^\msOH$ and $\D P_t^\msOH = -\nabla V(X_t^\msOH)$,
    \begin{align}\label{eq:OH-coarse-taylor}
    \begin{aligned}
        X_h^\msOH = X^\msO + hP^\msO - \frac{h^2}{2}\, \nabla V(X^\msO) - \int_0^h \int_0^{t_1} \int_0^{t_2} \nabla^2 V(X_{t_3}^\msOH)\, P_{t_3}^\msOH \, \D t_3 \, \D t_2 \, \D t_1\,.
    \end{aligned}
    \end{align}
    On the other hand, via a Taylor expansion between $X^\msO$ and $X^\msOBA$,
    \begin{align}\label{eq:obabc-expansion}
    \begin{aligned}
        X^\msOBABC &= X^\msO + h P^\msO - \Bigl(\frac{h^2}{3}\, \nabla V(X^\msO) + \frac{h^2}{6}\, \nabla V(X^{\msOBA}) \Bigr)\\
        &= X^\msO + h P^\msO - \frac{h^2}{2}\, \nabla V(X^\msO) - \frac{h^2}{6} \int_0^h \nabla^2 V(X^{\msO} + t P^{\msOB})\, P^{\msOB} \, \D t\,.
    \end{aligned}
    \end{align}
    Thus, bounding $\norm{\nabla^2 V}_{\rm op} \leq \beta$,
    \begin{align}\label{eq:obabo-x-loose}
        \norm{X^\msOBABC - X_h^\msOH} \lesssim \beta h^3\, \Bigl(\norm{P^{\msOB}} + \sup_{t \in [0, h]}\norm{P_t^\msOH} \Bigr) \lesssim \beta h^3\, \bigl(\norm{P^\msO} + h\, \norm{\nabla V(X^\msO)}\bigr)\,,
    \end{align}
    for $h \ll \nfrac{1}{\beta^{1/2}}$, where we use Lemma~\ref{lem:ham-flow-grow} to control $\norm{P^\msOH_t}$.
    Also, note that $X^\msOBABC = X^\msOBABCO$ and $X_h^\msOH = X_h^\msOHO$.

    Under Hessian smoothness, we expand these each one step further, giving
    \begin{align}\label{eq:oh-expansion-2}
    \begin{aligned}
        X_h^\msOH &= X^\msO + hP^\msO - \frac{h^2}{2}\, \nabla V(X^\msO) - \frac{h^3}{6}\, \nabla^2 V(X^\msO)\, P^\msO \\
        &\qquad-  \int_0^h \int_0^{t_1} \int_0^{t_2} \int_0^{t_3} \Bigl[\nabla^3 V(X_{t_4}^\msOH) [P_{t_4}^\msOH, P_{t_4}^\msOH] - \nabla^2 V(X_{t_4}^\msOH)\, \nabla V(X_{t_4}^\msOH) \Bigr]\, \D t_4 \, \D t_3 \, \D t_2 \, \D t_1\,.
    \end{aligned}
    \end{align}
    On the other hand, we have
    \begin{align*}
        \nabla^2 V(X^\msO + tP^\msOB)\, P^\msOB = \Bigl(\nabla^2 V(X^\msO) + \int_0^t \nabla^3 V(X^\msO + s P^\msOB)\,P^\msOB \, \D s \Bigr)\,P^\msOB\,,
    \end{align*}
    and when multiplying against $\nabla^2 V(X^\msO)$, we also use $P^\msOB = P^\msO - \frac{h}{2} \,\nabla V(X^\msO)$, which yields
    \begin{align}\label{eq:obabc-expansion-2}
    \begin{aligned}
        X^\msOBABC &=  X^\msO + h P^\msO - \frac{h^2}{2}\, \nabla V(X^\msO) - \frac{h^2}{6} \int_0^h \nabla^2 V(X^{\msO} + t P^{\msOB})\, P^{\msOB} \, \D t \\[0.25em]
        &= X^\msO + hP^\msO - \frac{h^2}{2}\, \nabla V(X^\msO) - \frac{h^3}{6}\, \nabla^2 V(X^\msO)\, P^\msO  \\
        &\qquad - \frac{h^2}{6} \int_0^h \Bigl[\int_0^t \nabla^3 V(X^\msO + s P^{\msOB})   [P^{\msOB}, P^{\msOB}] \, \D s - \frac{h}{2}\, \nabla^2 V(X^\msO)\, \nabla V(X^\msO) \Bigr]\, \D t \,. 
    \end{aligned}
    \end{align}
    Thus, evaluating their difference, we have
    \begin{align}\label{eq:obabco-x-tight}
    \begin{aligned}
        \norm{X^\msOBABC - X_h^\msOH} &\lesssim h^3 \int_0^h \Bigl(\norm{\nabla^3 V(X^\msO + t P^{\msOB})   [P^{\msOB}, P^{\msOB}]} + \norm{\nabla^3 V(X_{t}^\msOH) [P_{t}^\msOH, P_{t}^\msOH]}\Bigr) \, \D t\\
        &\qquad+ \beta h^4\, \Bigl(\norm{\nabla V(X^\msO)} + \sup_{t \in [0, h]}\norm{\nabla V(X_t^\msOH)}\Bigr)\,,
    \end{aligned}
    \end{align}
    and we use Lemma~\ref{lem:ham-flow-grow} again to bound the second line by $\beta h^4\, (\norm{\nabla V(X^\msO)} + \beta h\, \norm{P^\msO})$.
    
    The momenta follow a similar expansion. Namely, under standard smoothness, we expand again
    \begin{align*}
        P_h^\msOH = P^\msO - h\, \nabla V(X^\msO) - \int_0^h \int_0^{t_1} \nabla^2 V(X_{t_2}^\msOH)\, P_{t_2}^\msOH \, \D t_2 \, \D t_1\,,
    \end{align*}
    and
    \begin{align}\label{eq:obab-momentum-expansion}
        P^\msOBAB &= P^\msO - \frac{h}{2}\, \bigl(\nabla V(X^\msO) + \nabla V(X^{\msOBA}) \bigr) \\
        &= P^\msO - h\, \nabla V(X^\msO) - \frac{h}{2} \int_0^h \nabla^2 V(X^\msO + t P^{\msOB})\, P^{\msOB} \, \D t\,.
    \end{align}
    Thus, via similar analysis as the position, we have
    \begin{align}\label{eq:momentum-obab-nosmooth}
        \norm{P^\msOBAB - P_h^\msOH} \lesssim \beta h^2\, \bigl(\norm{P^\msO} + h\, \norm{\nabla V(X^\msO)}\bigr)\,.
    \end{align}
    Under additional smoothness, we expand
    \begin{align*}
        P_h^\msOH &= P^\msO - h\,\nabla V(X^\msO) - \frac{h^2}{2}\, \nabla V(X^\msO)\, P^\msO \\
        &\qquad - \int_0^h \int_0^{t_1} \int_0^{t_2} \Bigl[\nabla^3 V(X_{t_3}^\msOH)[P_{t_3}^\msOH, P_{t_3}^\msOH] - \nabla^2 V(X_{t_3}^\msOH)\, \nabla V(X_{t_3}^\msOH) \Bigr] \, \D t_3 \, \D t_2 \, \D t_1\,.
    \end{align*}
    On the other hand, expanding the momentum in a similar way as we did the position,
    \begin{align*}
        P^\msOBAB &= P^\msO - h\,\nabla V(X^\msO) - \frac{h^2}{2}\, \nabla^2 V(X^\msO)\, P^\msO  \\
        &\qquad -\frac{h}{2} \int_0^h \Bigl[ \int_0^t \nabla^3 V(X^\msO + sP^{\msOB})[P^{\msOB}, P^{\msOB}] \, \D s - \frac{h}{2}\, \nabla^2 V(X^\msO)\, \nabla V(X^\msO)\Bigr] \, \D t\,.
    \end{align*}
    Thus, the error is now
    \begin{align}\label{eq:momentum-obab-smooth}
    \begin{aligned}
        \norm{P^\msOBAB - P_h^\msOH} &\lesssim h^2\int_0^h\Bigl(\norm{\nabla^3 V(X^\msO + t P^{\msOB})   [P^{\msOB}, P^{\msOB}]} + \norm{\nabla^3 V(X_{t}^\msOH) [P_{t}^\msOH, P_{t}^\msOH]}\Bigr) \, \D t\\
        &\qquad+ \beta h^3\, \bigl(\norm{\nabla V(X^\msO)} + \beta h\, \norm{P^\msO}\bigr)\,.
    \end{aligned}
    \end{align}
    Also, note that
    \begin{align*}
        \norm{P^\msOBABCO - P_h^\msOHO}
        &\le e^{-\gamma h/2}\,\norm{P^\msOBAB - P_h^\msOH}
        \le \norm{P^\msOBAB - P_h^\msOH}\,.
    \end{align*}
    
    In all cases, we further bound $\norm{P^\msO} \lesssim \norm{p_0}+ \sqrt{\gamma dh}$.
\end{proof}

\subsubsection{Cross-regularity for OBABCO}\label{ssscn:cross-reg-overview}

We follow a similar framework as that in Theorem~\ref{thm:regularity}, only that instead of running OHO from two different points $(x_0, p_0)$ and $(\bar x_0, \bar p_0)$, we run OHO and its approximation OBABCO from the same point $(x_0, p_0)$. 

Begin with the bound for a single iteration. Again, introduce $\widehat B_1, B_1, \widehat B_2, B_2$, and now define $\widehat \Phi_h^\msx(x_0, p_0)$ to be the position-coordinate output of the BABC procedure with initial point $(x_0, p_0)$.
Then, for two independent standard Gaussian random variables $B_1$, $B_2$, we can write
\begin{align*}
    X^\msOBABCO &= \widehat \Phi_h^{\msx}\bigl(x_0,\, e^{-\gamma h/2} p_0 + \sqrt{1-e^{-\gamma h}} B_1\bigr)\,, \\
    P^\msOBABCO &= e^{-\gamma h/2}\, \widehat \Phi_h^{\msp}\bigl(x_0,\, e^{-\gamma h/2} p_0 + \sqrt{1-e^{-\gamma h}} B_1\bigr) + \sqrt{1-e^{-\gamma h}} B_2\,.
\end{align*}
Similarly, $(X^\msOHO, P^\msOHO)$ admits an analogous representation, but with $\Phi_h^\msx, \Phi_h^\msp$ in place of $\widehat \Phi_h^\msx, \widehat \Phi_h^\msp$.
However, instead of directly using this representation, we solve for $\widehat B_1$, $\widehat B_2$ such that
\begin{align*}
\begin{aligned}
    X^\msOBABCO &= \Phi_h^{\msx}\bigl(x_0,\, e^{-\gamma h/2} p_0 + \sqrt{1-e^{-\gamma h}} \widehat B_1\bigr)\,, \\
    P^\msOBABCO &= e^{-\gamma h/2}\, \Phi_h^{\msp}\bigl(x_0,\, e^{-\gamma h/2} p_0 + \sqrt{1-e^{-\gamma h}} \widehat B_1\bigr) + \sqrt{1-e^{-\gamma h}} \widehat B_2\,.
\end{aligned}
\end{align*}
If there is a unique solution $(\widehat B_1,\widehat B_2)$, then it is a deterministic function of $(x_0, p_0, B_1, B_2)$.

This expresses $(X^\msOHO, P^\msOHO) = \mc T(B_1,B_2)$ for a transformation $\mc T$, and $(X^\msOBABCO, P^\msOBABCO) = \mc T(\widehat B_1,\widehat B_2)$ for the \emph{same} transformation $\mc T$.
Then, by the data-processing inequality,
\begin{align*}
    \Renyi_q(\delta_{(x_0, p_0)} \widehat{\bs P}_h \mmid \delta_{(x_0, p_0)} \bs P_h)
    \le \Renyi_q(\law(\widehat B_1,\widehat B_2) \mmid \law(B_1,B_2))\,.
\end{align*}
Here, $\widehat{\bs P}_h$ denotes the Markov kernel for an iteration of OBABCO\@.
To bound the right-hand side, we use the Jacobian change of variables formula, which requires an understanding of the mapping $(B_1,B_2) \mapsto (\widehat B_1,\widehat B_2)$ (with $x_0$, $p_0$ held fixed).
We prove the following lemma.

\begin{lemma}[OBABCO cross-regularity]\label{lem:obabco_cross}
    Assume that $-\beta I \preceq \nabla^2 V \preceq \beta I$ and that $h \ll \nfrac{1}{\beta^{1/2} q^{1/2}} \wedge \nfrac{1}{\gamma}$.
    Then, it holds that
    \begin{align*}
        &\exp\bigl((q-1)\,\Renyi_q(\delta_{(x_0, p_0)} \widehat{\bs P}_h \mmid \delta_{(x_0, p_0)} \bs P_h)\bigr) \\
        &\qquad \le \E\exp\Bigl(O\Bigl(\frac{q^2}{\gamma h}\,\bigl(\frac{1}{h^2}\,\norm{\widehat \Phi_h^\msx(x_0, P^\msO) - \Phi_h^\msx(x_0, P^\msO)}^2 + \norm{\widehat \Phi_h^\msp(x_0, P^\msO) - \Phi_h^\msp(x_0, P^\msO)}^2\bigr) + \Delta\Bigr)\Bigr)\,,
    \end{align*}
    where $\Delta \lesssim \beta^2 dh^4 q^2$, and under Assumption~\ref{as:hess-lip},
    \begin{align*}
        \Delta
        &\lesssim \beta^3 dh^6 q^3 + h^5 q^2 \int_0^h \bigl(\norm{\nabla^3 V(x_0 + t P^\msOB)\, P^\msOB}^2_{\rm F} + \norm{\nabla^3 V(X_t^\msOH)\, P_t^\msOH}^2_{\rm F}\bigr) \, \D t \\[0.25em]
        &\qquad + \betH^2 d h^6 q^2 \,\frac{\norm{\widehat \Phi_h^\msx(x_0, P^\msO) - \Phi_h^\msx(x_0, P^\msO)}^2}{h^2}\,.
    \end{align*}
\end{lemma}
\begin{proof}
Consider the control
\begin{align*}
    \widehat P^\target = \Psi_h(x_0, \widehat \Phi_h^\msx(x_0, P^\msO))\,,
\end{align*}
where $P^\msO = e^{-\gamma h/2} p_0 + \sqrt{1-e^{-\gamma h}} B_1$ as before.
As before, we have
\begin{align}\label{eq:obabo-brownian-transform}
    \sqrt{1-e^{-\gamma h}} \widehat B_1 = \widehat P^\target - e^{-\gamma h/2} p_0\,.
\end{align}
We also again define
\begin{align}\label{eq:obabo-brownian-transform-ii}
    \sqrt{1-e^{-\gamma h}} \widehat B_2 = e^{-\gamma h/2}\, \widehat \Phi_h^\msp(x_0, P^\msO) - e^{-\gamma h/2}\, \Phi_h^\msp(x_0, \widehat P^\target) + \sqrt{1-e^{-\gamma h}} B_2\,.
\end{align}
Following the strategy of Theorem~\ref{thm:regularity}, we consider the R\'enyi divergence between $\widehat B_1$ and $B_1$, and then the conditional R\'enyi divergence between $\widehat B_2$ and $B_2$ given $B_1$.

\textbf{First Brownian motion.}
Let $F$ be the map taking $B_1 \mapsto \widehat B_1$ given $(x_0, p_0)$. The same calculation as~\eqref{eq:radon-nikodym-final}, but expanding the $\log \det$ terms to one higher order, means it will be sufficient to bound
\begin{align}\label{eq:refined-radon-nikodym}
\begin{aligned}
    \E_\upgamma \Bigl(\frac{F_{\#} \upgamma}{\upgamma} \Bigr)^q  
    &\leq \E\bigl[\exp\bigl((2q-1)^2\, \norm{\widehat B_1 - B_1}^2 \bigr)\bigr]^{1/4}  \exp(O(q^3 \lambda_{\max}^3 d))\\
    &\qquad \times \E\bigl[\exp\bigl((2q-1)^2\,\norm{(\nabla F^{-1})(B_1) - I}_{\rm F}^2\bigr)\bigr]^{1/4}\,,
\end{aligned}
\end{align}
where we apply Cauchy--Schwarz and the determinant-trace expansion in Lemma~\ref{lem:log-det}. Lemma~\ref{lem:obabo-distortion} below tells us that $\lambda_{\max} \lesssim \beta h^2$, and so this expansion is valid so long as $h \ll \nfrac{1}{\beta^{1/2} q^{1/2}}$. The expectations are taken with $B_1$ standard Gaussian. This representation, where we take the Frobenius norm in the Jacobian term, will be needed in the Hessian smooth case. On the other hand, when we do not assume Hessian smoothness, we can additionally bound this squared Frobenius norm by $\lambda_{\max}^2 d$, noting that the eigenvalues of $\nabla F^{-1} - I$ and $\nabla F - I$ are interchangeable in magnitude when $\lambda_{\max} \ll 1$. The bound for $q^3 \lambda_{\max}^3 d$ combined with the Frobenius norm term follow from Lemma~\ref{lem:obabo-distortion} in either the standard or Hessian Lipschitz cases. We note in particular that $(\nabla F)(B_1)$ is precisely $\nabla_p \Psi_h(x_0, \widehat \Phi_h^\msx(x_0, P^\msO))$ found in Lemma~\ref{lem:obabo-distortion} after substituting the definition of $F$.

Now, we note from~\eqref{eq:obabo-brownian-transform} that
\begin{align}\label{eq:b-bhat-diff}
    \norm{\widehat B_1 - B_1}^2 = \frac{1}{1-e^{-\gamma h}}\, \norm{\widehat P^\target - P^\msO}^2 \asymp \frac{1}{\gamma h}\, \norm{\widehat P^\target - P^\msO}^2\,.
\end{align}
We also have from our previous calculations in Lemma~\ref{lem:hmc-expansion}, when $h \lesssim \nfrac{1}{\beta^{1/2}}$,
\begin{align*}
    \Phi_h^\msx(x_0, p_0) - \Phi_h^\msx(x_0, \bar p_0) = h\, \delta p_0 + O(\beta h^3\, \norm{\delta p_0})\,.
\end{align*}
Now,
\begin{align*}
    0 = \widehat \Phi_h^\msx(x_0, P^\msO) - \Phi_h^\msx(x_0, \widehat P^\target) = \widehat \Phi_h^\msx(x_0, P^\msO) - \Phi_h^\msx(x_0, P^\msO) + \Phi_h^\msx(x_0, P^\msO) - \Phi_h^\msx(x_0, \widehat P^\target)\,.
\end{align*}
Rearranging and using the prior inequality, we have
\begin{align}\label{eq:p-target-soln}
    \norm{\widehat P^\target - P^\msO} \lesssim \frac{1}{h}\, \bigl\lVert\widehat \Phi_h^\msx(x_0, P^\msO) - \Phi_h^\msx(x_0, P^\msO)\bigr\rVert\,.
\end{align}

\textbf{Second Brownian motion.}
Note that $\widehat B_2, B_2$ will be chosen to ensure matching of the momenta. Just as we reasoned in the proof of Theorem~\ref{thm:regularity}, we have
\begin{align*}
    \E \Bigl[\Bigl(\frac{\D \operatorname{law}(\widehat B_2 \mid B_1)}{\D \upgamma}\Bigr)^q \Bigm\vert B_1 \Bigr] &= \E \Bigl[\exp\Bigl(\frac{q\,(q-1)}{2}\, \norm{\widehat B_2 - B_2}^2\Bigr) \Bigm\vert B_1\Bigr]\,.
\end{align*}
This transformation is linear and so no Jacobians arise. Now, we can bound for $h \lesssim \nfrac{1}{\gamma}$
\begin{align}\label{eq:b2-expansion}
\begin{aligned}
     \norm{\widehat B_2 - B_2}^2 &\lesssim \frac{1}{\gamma h}\, \norm{\widehat\Phi_h^{\msp}(x_0, P^\msO) - \Phi_h^{\msp}(x_0, \widehat P^\target)}^2 \\
     &\lesssim \frac{1}{\gamma h}\, \bigl(\norm{\Phi_h^{\msp}(x_0, P^\msO) - \Phi_h^{\msp}(x_0, \widehat P^\target)}^2 + \norm{\widehat \Phi_h^{\msp}(x_0, P^\msO) - \Phi_h^{\msp}(x_0, P^\msO)}^2 \bigr) \,.
\end{aligned}
\end{align}
We bound these terms separately. The first term is bounded via Lemma~\ref{lem:hmc-expansion} as
\begin{align*}
    \frac{1}{\gamma h} \,\norm{\Phi_h^{\msp}(x_0, P^\msO) - \Phi_h^{\msp}(x_0, \widehat P^\target)}^2 \lesssim \frac{1}{\gamma h}\, \norm{P^\msO - \widehat P^\target}^2\,.
\end{align*}
This term will always be bounded as in~\eqref{eq:b-bhat-diff}, and so will only contribute to the same order as terms in $\widehat B_1 - B_1$. Now, we will always use the bound from~\eqref{eq:p-target-soln} that
\begin{align*}
    \norm{\widehat P^\target - P^\msO}^2 \lesssim \frac{1}{h^2}\, \norm{\widehat \Phi_h^\msx(x_0, P^\msO) - \Phi_h^\msx(x_0, P^\msO)}^2\,.
\end{align*}
On the other hand, for the second term in~\eqref{eq:b2-expansion}, note that in both~\eqref{eq:momentum-obab-nosmooth} or~\eqref{eq:momentum-obab-smooth}, there is a loss of a factor of $h$ compared to the position expansions in~\eqref{eq:obabo-x-loose} and~\eqref{eq:obabco-x-tight}. However, when bounding $B_1 - \widehat B_1$, we lost a factor of $h$ on the position expansions in~\eqref{eq:p-target-soln}, as seen above. Thus, the bounds on these terms will contribute the exact same terms as found when bounding $\widehat B_1 - B_1$. So we can effectively absorb the contribution of the $\widehat B_2 - B_2$ terms in our analysis.
\end{proof}

\subsubsection{Discrete Girsanov argument}\label{ssscn:discrete-girsanov}

Next, we pass from the one-step bound in Lemma~\ref{lem:obabco_cross} to a multi-step bound.
Unlike the KL divergence, the R\'enyi divergence of order larger than $1$ does not satisfy a chain rule, so we use the following martingale decomposition, which can be viewed as a discrete-time analogue of Girsanov's theorem in stochastic calculus.

\begin{lemma}[R\'enyi martingale decomposition]\label{lem:renyi_mg}
    Let $\bs\mu$, $\bs\nu$ be two possible joint distributions for a Markov chain $(X_0,X_1,\dotsc,X_N)$.
    Then,
    \begin{align*}
        \Renyi_q(\bs \mu \mmid \bs \nu)
        \le \frac{1}{2\,(q-1)} \log \frac{1}{N+1} \,\Bigl[&\sum_{n=1}^N \E_{\bs\nu} \exp\bigl((2q-1)\,(N+1)\,\Renyi_{2q}(\bs\mu_{n\mid n-1}(\cdot\mid X_{n-1}) \mmid \bs \nu_{n\mid n-1}(\cdot\mid X_{n-1})\bigr) \\
        &\qquad + \exp\bigl((2q-1)\,(N+1)\,\Renyi_{2q}(\bs\mu_{0}(\cdot) \mmid \bs \nu_{0}(\cdot)\bigr)\Bigr]\,.
    \end{align*}
    When $\bs\mu_0 = \bs\nu_0$, this simplifies to
    \begin{align*}
        \Renyi_q(\bs \mu \mmid \bs \nu)
        &\le \frac{1}{2\,(q-1)} \log \frac{1}{N} \sum_{n=1}^N \E_{\bs\nu} \exp\bigl((2q-1)\,N\,\Renyi_{2q}(\bs\mu_{n\mid n-1}(\cdot\mid X_{n-1}) \mmid \bs \nu_{n\mid n-1}(\cdot\mid X_{n-1})\bigr)\,.
    \end{align*}
\end{lemma}
\begin{proof}
    By the Markov property, the Radon--Nikodym derivative factorizes:
    \begin{align*}
        \frac{\D\bs\mu}{\D\bs\nu}(x_0,\dotsc,x_N)
        &= \underbrace{\frac{\D\bs\mu_0}{\D\bs\nu_0}(x_0)}_{\eqqcolon L_0}\, \prod_{n=1}^{N} \underbrace{\frac{\D\bs\mu_{n\mid n-1}}{\D\bs\nu_{n\mid n-1}}(x_{n}\mid x_{n-1})}_{\eqqcolon L_n}\,.
    \end{align*}
    Let $\ms F_n$ denote the $\sigma$-algebra generated by $X_0,X_1,\dotsc,X_n$.
    Then, $\{\prod_{k=0}^n L_k, \ms F_n\}_{n=0}^N$ is a $\bs\nu$-martingale.
    Also, for $q > 1$, $n \geq 1$, denote $R_{q,n} \deq \E_{\bs\nu}[L_n^q \mid \ms F_{n-1}]$, with $R_{q,0} \deq 1$.
    Then, it is straightforward to see that
    \begin{align*}
        \Bigl\{\prod_{k=0}^n \frac{L_k^q}{R_{q,k}}\,,\;\ms F_n\Bigr\}_{n=0}^N \qquad\text{is a non-negative local}~\bs\nu\text{-martingale}\,,
    \end{align*}
    hence it is a $\bs\nu$-supermartingale.
    By Cauchy--Schwarz,
    \begin{align*}
        \exp\bigl((q-1)\,\Renyi_q(\bs\mu\mmid \bs\nu)\bigr)
        &= \E_{\bs\nu} \prod_{n=0}^N L_n^q
        \le \Bigl\{ \Bigl(\E_{\bs \nu} \prod_{n=0}^N \frac{L_n^{2q}}{R_{2q,n}}\Bigr)\,\Bigl(\E_{\bs\nu} \prod_{n=0}^N R_{2q,n}\Bigr)\Bigr\}^{1/2}
        \le \Bigl\{ \E_{\bs\nu} \prod_{n=0}^N R_{2q,n}\Bigr\}^{1/2}\,.
    \end{align*}
    The proof is concluded via the AM--GM inequality.
\end{proof}

Let $(X_{nh}, P_{nh})_{n\in\N}$ denote the iterates of OHO started at $\bs\mu_0$, so that $P_{nh}^\msO$ denotes the application of~\eqref{eq:O} to $P_{nh}$.
From Lemma~\ref{lem:obabco_cross}, we now obtain the following bound.

\begin{cor}[OBABCO multi-step R\'enyi bound]\label{cor:obabco_multi_step}
    Assume that $-\beta I \preceq \nabla^2 V \preceq \beta I$ and that $h \ll \nfrac{1}{\beta^{1/2} q^{1/2}} \wedge \nfrac{1}{\gamma}$.
    Then, it holds that
    \begin{align*}
        &\exp\bigl(2\,(q-1)\,\Renyi_q(\bs\mu_0 \widehat{\bs P}_h^N \mmid \bs\mu_0 \bs P_h^N)\bigr) \\
        &\quad \le \frac{1}{N} \sum_{n=0}^{N-1} \E\exp\Bigl(O\Bigl(\frac{q^2 N}{\gamma h}\,\bigl(\frac{\norm{\widehat \Phi_h^\msx(X_{nh}, P_{nh}^\msO) - \Phi_h^\msx(X_{nh}, P_{nh}^\msO)}^2}{h^2} + \norm{\widehat \Phi_h^\msp(X_{nh}, P_{nh}^\msO) - \Phi_h^\msp(X_{nh}, P_{nh}^\msO)}^2\bigr)\Bigr) \\
        &\qquad\qquad\qquad{} \times \exp(O(N\Delta_n))\,,
    \end{align*}
    where $\Delta_n \lesssim \beta^2 dh^4 q^2$ and, under Assumption~\ref{as:hess-lip},
    \begin{align*}
        \Delta_n
        &\lesssim \beta^3 dh^6 q^3 + h^5 q^2 \int_0^h \bigl(\norm{\nabla^3 V(X_{nh} + t P_{nh}^\msOB)\, P_{nh}^\msOB}^2_{\rm F} + \norm{\nabla^3 V(X_{nh,t}^\msOH)\, P_{nh,t}^\msOH}^2_{\rm F}\bigr) \, \D t \\[0.25em]
        &\qquad + \betH^2 d h^6 q^2 \,\frac{\norm{\widehat \Phi_h^\msx(X_{nh}, P_{nh}^\msO) - \Phi_h^\msx(X_{nh}, P_{nh}^\msO)}^2}{h^2}\,.
    \end{align*}
\end{cor}

\subsubsection{Phase I: Coarse R\'enyi bound}

In this section and the next, $C, \bar C > 0$ will denote absolute constants whose value may change from line to line.

\begin{proof}[Proof of Theorem~\ref{thm:coarse}]
    It remains to bound the one-step error and to plug this into Corollary~\ref{cor:obabco_multi_step}.

First, by Lemma~\ref{lem:errors},
\begin{align*}
    &\frac{1}{h^2}\,\norm{\widehat \Phi_h^\msx(X_{nh}, P_{nh}^\msO) - \Phi_h^\msx(X_{nh}, P_{nh}^\msO)}^2 + \norm{\widehat \Phi_h^\msp(X_{nh}, P_{nh}^\msO) - \Phi_h^\msp(X_{nh}, P_{nh}^\msO)}^2 \\
    &\qquad{} \lesssim \beta^2 h^4\, \bigl(\norm{P_{nh}}^2 + h^2\, \norm{\nabla V(X_{nh})}^2 + \gamma dh\bigr)\,.
\end{align*}
Hence, Corollary~\ref{cor:obabco_multi_step} yields
\begin{align*}
    &\exp\bigl(2\,(q-1)\,\Renyi_q(\bs\mu_0\widehat{\bs P}_h^N \mmid \bs\mu_0 \bs P_h^N)\bigr) \\
    &\qquad \le \frac{1}{N} \sum_{n=0}^{N-1} \E\exp\Bigl(\frac{C\beta^2 h^2 q^2 T}{\gamma}\,\bigl(\norm{P_{nh}}^2 + h^2\,\norm{\nabla V(X_{nh})}^2\bigr) + C\beta^2 dh^3 q^2 T\Bigr)
\end{align*}
for $T \deq Nh$.
Next, for $P \sim \cN(0, I)$, by Cauchy--Schwarz and the data-processing inequality,
\begin{align*}
    \E\exp\bigl(C\,(\norm{P_{nh}}-\sqrt d)^2\bigr)
    &\le \Bigl\{\bigl(1+\chi^2(\bs\mu_0 \bs P_h^n\mmid \bs \pi)\bigr)\,\E\exp\bigl(2C\,(\norm{P}-\sqrt d)^2\bigr)\Bigr\}^{1/2} \\
    &\le \exp\bigl(\bar C\,(1 + \Renyi_2(\bs\mu_0 \mmid \bs \pi))\bigr)\,,
\end{align*}
provided that $C$ is a sufficiently small absolute constant.
It readily yields
\begin{align*}
    \E\exp\Bigl(\frac{C\beta^2 h^2 q^2 T}{\gamma}\,\norm{P_{nh}}^2\Bigr)
    &\le \exp\Bigl(\frac{\bar C\beta^2 h^2 q^2 T}{\gamma}\,(d + \Renyi_2(\bs\mu_0\mmid \bs\pi))\Bigr)
\end{align*}
provided $h \ll \nfrac{\gamma^{1/2}}{\beta q T^{1/2}}$.
A similar argument for the $\norm{\nabla V(X_{nh})}$ term, using Lemma~\ref{lem:subgsn-score} instead, yields the conclusion of the theorem.
\end{proof}

\subsubsection{Phase II: Refined R\'enyi bound}

\begin{proof}[Proof of Theorem~\ref{thm:obabco-final-renyi}]
    Assume that $h \lesssim \nfrac{1}{\betH^{2/7}\gamma^{1/7} d^{1/7}}$ (recall that we have already met this condition in Theorem~\ref{thm:harnack}).
    By Corollary~\ref{cor:obabco_multi_step} (absorbing some terms via our condition on $h$),
    \begin{align*}
        &\exp\bigl((2q-1)\,\Renyi_q(\bs\mu_0 \widehat{\bs P}_h^N \mmid \bs\mu_0 \bs P_h^N)\bigr)
        \le \frac{1}{N} \sum_{n=0}^{N-1} \E\underbrace{\exp\Bigl(\frac{Cq^2 T}{\gamma h^2}\,\mc E_n^2 +  C\beta^3 dh^5 q^3 T \Bigr)}_{\rm I} \\
        &\qquad{} \times \underbrace{\exp\Bigl(\frac{Ch^3 q^2 T}{\gamma} \int_0^h \bigl(\norm{\nabla^3 V(X_{nh} + t P_{nh}^\msOB)\, P_{nh}^\msOB}^2_{\rm F} + \norm{\nabla^3 V(X_{nh,t}^\msOH)\, P_{nh,t}^\msOH}^2_{\rm F}\bigr) \, \D t\Bigr)}_{\rm II}\,,
    \end{align*}
    where
    \begin{align*}
        \mc E_n^2
        &\deq \frac{\norm{\widehat \Phi_h^\msx(X_{nh}, P_{nh}^\msO) - \Phi_h^\msx(X_{nh}, P_{nh}^\msO)}^2}{h^2} + \norm{\widehat \Phi_h^\msp(X_{nh}, P_{nh}^\msO) - \Phi_h^\msp(X_{nh}, P_{nh}^\msO)}^2
    \end{align*}
    is the squared local error at iteration $n$.
    We use Cauchy--Schwarz repeatedly to decouple the terms in the exponential, allowing us to bound them separately.
    
    For simplicity, we write down the bound for $n= 0$ and omit the iteration number subscript, and note that the bounds for $n \ge 1$ look identical, only with the initial point $(X_{(n-1)h}, P_{(n-1)h})$ drawn from the appropriate distribution ${\bs \mu}_{(n-1)h}$.

    \paragraph{Term II\@.}
    We begin by handling the easier term II\@.
    Consider
    \begin{align*}
        \E\exp\Bigl(\frac{Ch^3 q^2 T}{\gamma} \int_0^h \norm{\nabla^3 V(X_t^\msOH)\,P_t^\msOH}_{\rm F}^2\,\D t\Bigr)
        &\le \E\exp\Bigl(\frac{C\beta_\eff^2 h^4 q^2 T}{\gamma}\sup_{t\in [0,h]}{\norm{P_t^\msOH}^2}\Bigr)\,\D t \\[0.25em]
        &\le \E\exp\Bigl(\frac{\bar C\beta_\eff^2 h^4 q^2 T}{\gamma}\,\bigl(\norm{P^\msO}^2 + h^2\,\norm{\nabla V(X)}^2\bigr)\Bigr)\,,
    \end{align*}
    where we used Lemma~\ref{lem:ham-flow-grow}.
    Now, by the same change of measure argument as in the proof of Theorem~\ref{thm:coarse}, it follows that for $h \ll \nfrac{\gamma^{1/4}}{\beta_\eff^{1/2} q^{1/2} T^{1/4}}$,
    \begin{align*}
        \E\exp\Bigl(\frac{Ch^3 q^2 T}{\gamma} \int_0^h \norm{\nabla^3 V(X_t^\msOH)\,P_t^\msOH}_{\rm F}^2\,\D t\Bigr)
        &\le \exp\Bigl(\frac{\bar C\beta_\eff^2 h^4 q^2 T}{\gamma}\,\bigl(d+\Renyi_2(\bs\mu_0\mmid\bs\pi)\bigr)\Bigr)\,.
    \end{align*}
    The other term involving $\norm{\nabla^3 V(X + tP^\msOB)\,P^\msOB}_{\rm F}$ is completely analogous.

\paragraph{Term I\@.}
For this term, we would like to substitute in the refined bound from Lemma~\ref{lem:errors}, but this would lead to terms such as
\begin{align*}
    \E\exp\Bigl(\frac{C h^3 q^2 T}{\gamma} \int_0^h \norm{\nabla^3 V(X_t^\msOH)[P_t^\msOH,P_t^\msOH]}^2\,\D t\Bigr)\,.
\end{align*}
Under our assumptions, this term resembles $\E\exp(\lambda\,\norm{P}^4)$ for some $\lambda > 0$, which is infinite as $P$ is sub-Gaussian at best.
Thus, we use the following strategy:
\begin{align*}
    \E\exp\Bigl(\frac{Cq^2 T}{\gamma h^2}\,\mc E^2\Bigr)
    &\le \E\exp(C\min\{A,B\})\,,
\end{align*}
where $A$, $B$ are produced by the coarse and refined bounds in Lemma~\ref{lem:errors}:
\begin{align*}
    A
    &\deq \frac{\beta^2 h^2 q^2 T}{\gamma}\,\bigl(\norm P^2 + h^2\,\norm{\nabla V(X)}^2 + \gamma dh\bigr)\,, \\[0.25em]
    B
    &\deq \underbrace{\frac{\beta^2 h^4 q^2 T}{\gamma}\,\bigl(\norm{\nabla V(X)}^2 + \beta^2 h^2\,\norm P^2 + \beta^2 \gamma dh^{3}\bigr)}_{\eqqcolon B_1} \\
    &\qquad{} + \underbrace{\frac{h^3 q^2 T}{\gamma} \int_0^h\bigl(\norm{\nabla^3 V(X+tP^{\msOB})[P^{\msOB}, P^{\msOB}]}^2 + \norm{\nabla^3 V(X_t^\msOH)[P_t^\msOH, P_t^\msOH]}^2 \bigr) \, \D t}_{\eqqcolon B_2}\,.
\end{align*}
Moreover,
\begin{align*}
    \E\exp(C\min\{A,B\}) \le \E\exp(CB_1 + C\min\{A, B_2\})
\end{align*}
and the same change of measure arguments as above yield
\begin{align*}
    \E\exp(CB_1) \le \exp\Bigl(\frac{\bar C\beta^3 h^4 q^2 T}{\gamma}\,\bigl(d+\Renyi_2(\bs\mu_0\mmid\bs\pi)\bigr)\Bigr)\,,
\end{align*}
provided $h \ll \nfrac{\gamma^{1/4}}{\beta^{3/4} q^{1/2} T^{1/4}}$. 
The additional term $\exp(C\beta^3 dh^5 q^3 T)$ can be absorbed into this provided that $h \ll \nfrac{1}{\gamma q}$.

We will control the other term via
\begin{align*}
    \E\exp(C\min\{A,B_2\})
    &\le \E[\exp(CK)\one_{B_2 \ge K} + \exp(CA) \one_{B_2 \le K}] \\
    &\le \exp(CK) + \sqrt{\E\exp(2CA)\,\Pr(B_2 > K)}\,,
\end{align*}
for an appropriately chosen threshold $K > 0$.

By the proof of Theorem~\ref{thm:coarse}, we have
\begin{align*}
    \E\exp(2CA)
    \le \exp\Bigl(\frac{\bar C\beta^2 h^2 q^2 T}{\gamma}\,\bigl(d+\Renyi_2(\bs\mu_0\mmid\bs\pi)\bigr)\Bigr)
\end{align*}
for $h \ll \nfrac{\gamma^{1/2}}{\beta q T^{1/2}}$.
Moreover, in Lemma~\ref{lem:stoc-proc-bound} below, we establish the following high-probability bounds: for every $\delta \in (0,\frac{1}{2})$, with probability at least $1-\delta$,
\begin{align*}
    &\int_0^h\bigl(\norm{\nabla^3 V(X+tP^{\msOB})[P^{\msOB}, P^{\msOB}]}^2 + \norm{\nabla^3 V(X_t^\msOH)[P_t^\msOH, P_t^\msOH]}^2 \bigr) \, \D t \\
    &\qquad \lesssim \beta_\eff^2 h\,\bigl(d + \Renyi_2(\bs\mu_0\mmid\bs\pi)^2 + \log^2\frac{1}{\delta}\bigr)\,,
\end{align*}
provided $h \ll \nfrac{1}{\beta^{1/2} d^{1/4}}$.

This shows that
\begin{align*}
    \E\exp(C\min\{A,B_2\})
    &\le \exp\Bigl(\frac{\bar C\beta_\eff^2 h^4 q^2 T}{\gamma}\,\bigl(d + \Renyi_2(\bs\mu_0\mmid\bs\pi)^2 + \log^2\frac{1}{\delta}\bigr)\Bigr) \\
    &\qquad{} + \sqrt\delta \exp\Bigl(\frac{\bar C\beta^2 h^2 q^2 T}{\gamma}\,\bigl(d+\Renyi_2(\bs\mu_0\mmid\bs\pi)\bigr)\Bigr) \\
    &\eqqcolon e^b + \sqrt\delta e^a\,.
\end{align*}
Using, for $b > 0$,
\begin{align*}
    e^b + \sqrt\delta e^a
    &\le e^b\,(1+\sqrt\delta e^a)
    \le \exp(b + \sqrt\delta e^a)\,,
\end{align*}
we arrive at
\begin{align*}
    &\E\exp(C\min\{A,B_2\}) \\
    &\qquad \le \exp\biggl[\frac{\bar C\beta_\eff^2 h^4 q^2 T}{\gamma}\,\bigl(d + \Renyi_2(\bs\mu_0\mmid\bs\pi)^2 + \log^2\frac{1}{\delta}\bigr) + \sqrt\delta \exp\Bigl(\frac{\bar C\beta^2 h^2 q^2 T}{\gamma}\,\bigl(d+\Renyi_2(\bs\mu_0\mmid\bs\pi)\bigr)\Bigr)\biggr]\,.
\end{align*}
To (approximately) balance the terms, set
\begin{align*}
    \sqrt\delta = \frac{\bar C\beta_\eff^2 d h^4 q^2 T}{\gamma} \exp\Bigl(-\frac{\bar C\beta^2 h^2 q^2 T}{\gamma}\,\bigl(d+\Renyi_2(\bs\mu_0\mmid\bs\pi)\bigr)\Bigr)\,.
\end{align*}
This yields
\begin{align*}
    &\E\exp(C\min\{A,B_2\}) \\
    &\qquad \le \exp\Bigl[\frac{\bar C\beta_\eff^2 h^4 q^2 T}{\gamma}\,\Bigl(d + \Renyi_2(\bs\mu_0\mmid\bs\pi)^2 + \Bigl(\frac{\beta^2 h^2 q^2 T}{\gamma}\,\bigl(d+\Renyi_2(\bs\mu_0\mmid\bs\pi)\bigr) + \log\frac{\gamma}{\beta_\eff^2 dh^4 q^2 T}\Bigr)^2\Bigr)\Bigr] \\
    &\qquad \le \exp\Bigl[\frac{\bar C\beta_\eff^2 h^4 q^2 T}{\gamma}\,\Bigl(d + \Renyi_2(\bs\mu_0\mmid\bs\pi)^2 + \log^2\frac{\gamma}{\beta_\eff^2 dh^4 q^2 T}\Bigr)\Bigr] \\
\end{align*}
provided $h \lesssim \nfrac{\gamma^{1/2}}{\beta d^{1/4} qT^{1/2}}$.
Collecting together all of the terms finishes the proof.
\end{proof}

\subsubsection{Helper lemmas}

Finally, in our proof, we made use of the following Jacobian bounds.

\begin{lemma}[OBABCO Jacobian bounds]\label{lem:obabo-distortion}
    Assume that $-\beta I \preceq \nabla^2 V \preceq \beta I$ and $h \lesssim \nfrac{1}{\beta^{1/2}}$.
    Without Hessian Lipschitzness, we have
    \begin{align*}
        \sup_{x_0\in\R^d}{\norm{\nabla_{p_0} [\Psi_h(x_0, \widehat \Phi_h^\msx(x_0, p_0))] - I}_{\rm op}} \lesssim \beta h^2\,.
    \end{align*}
    Next, let $P^\msO$ denote the result of applying~\eqref{eq:O} to $p_0$, and let $(X_t^\msOH, P_t^\msOH)_{t\ge 0}$ denote the Hamiltonian flow~\eqref{eq:ham_ode} started from $(x_0, P^\msO)$.
    Let $\widehat P^\target \deq \Psi_h(x_0, \widehat \Phi_h^\msx(x_0, P^\msO))$.
    Then, under Assumption~\ref{as:hess-lip},
    \begin{align*}
        \norm{\nabla_p[\Psi_h(x_0, \widehat \Phi_h^\msx(x_0, \cdot))] \mid_{P^\msO} - I}_{\rm F}^2
        &\lesssim h^5 \int_0^h \bigl(\norm{\nabla^3 V(x_0 + t P^\msOB)\, P^\msOB}^2_{\rm F} + \norm{\nabla^3 V(X_t^\msOH)\, P_t^\msOH}^2_{\rm F}\bigr) \, \D t \\
        &\qquad + \betH^2 d h^6 \,\norm{P^\msO - \widehat P^\target}^2 + \beta^4 d h^8\,.
    \end{align*}
\end{lemma}
\begin{proof}
Define $K(\cdot) \deq \nabla_2 \Psi_h(x_0, \cdot)$ and $\widehat J(\cdot) \deq \nabla_2 \widehat \Phi_h^x(x_0, \cdot)$. We have already obtained bounds for the eigenvalues of $K$ in Lemma~\ref{lem:hmc-distortion} --- $\norm{K(\cdot) - h^{-1} I}_{\rm op} \lesssim \beta h$. 

Without Hessian smoothness, we can directly compute $\nabla_2 \widehat \Phi_h^\msx(x_0, \cdot)$ via~\eqref{eq:obabc-expansion} to obtain
\begin{align}\label{eq:jacob-bound}
    \nabla_2 \widehat \Phi_h^\msx(x_0, \cdot) = h I - \frac{h^3}{6}\, \nabla^2 V\Bigl(x_0 + h\,\bigl(\cdot - \frac{h}{2}\, \nabla V(x_0)\bigr)\Bigr) \,.
\end{align}
For $h \lesssim \nfrac{1}{\beta^{1/2}}$, this gives the first bound.

Under Hessian Lipschitzness, it is helpful to consider the specific point $P^\msO$ of instantiation, which will be useful when applying this lemma. Let $\widehat J$ correspond to the matrix $\widehat J(P^\msO) = \nabla_2 \widehat \Phi_h^\msx(x_0, P^\msO)$. Let $ X_t^\target \deq \Phi_t^\msx(x_0, \widehat P^\target) = x_t(\widehat P^\target)$. Using the expansion of the Hamiltonian flow, we write
\begin{align*}
    \nabla_2 \Phi_h^\msx(x_0, \widehat P^\target) - hI = -\int_0^h \int_0^{t_1} \nabla^2 V(X_{t_2}^\target)\, \nabla x_{t_2}(\widehat P^\target)\, \D t_2 \, \D t_1\,.
\end{align*}
Secondly, note that $\norm{\nabla x_t(\widehat P^\target) - tI}_{\rm op} \lesssim \beta h^3$ for $h \lesssim \nfrac{1}{\beta^{1/2}}$, so that this can be written as
\begin{align*}
    \nabla_2 \Phi_h^\msx(x_0, \widehat P^\target) - h I &= -\int_0^h \int_0^{t_1} t_2\, \nabla^2 V(X_{t_2}^\target) \, \D t_2 \, \D t_1 + O(\beta^2 h^5) \\
    &=-\int_0^h (h-t)\, t\, \nabla^2 V(X_t^\target) \, \D t + O(\beta^2 h^5)\,.
\end{align*}
It follows at last from~\eqref{eq:jacob-bound} that
\begin{align*}
    \widehat J(P^\msO) - \nabla_2 \Phi_h^\msx(x_0, \widehat P^\target) & = \int_0^h (h-t)\, t\, [\nabla^2 V(X^\target_t) - \nabla^2 V(x_0 + hP^\msOB)] \, \D t + O(\beta^2 h^5) \\
    &= \int_0^h (h-t)\, t\, [\nabla^2 V(X_t^\msOH) - \nabla^2 V(x_0 + h P^\msOB)] \, \D t \\
    &\qquad + O(\betH h^4\,\norm{P^\msO - \widehat P^\target}+ \beta^2 h^5)\,.
\end{align*}
Here, we use the Hessian Lipschitz condition and the fact that
\begin{align*}
    \norm{X_t^{\msOH} - X_t^\target} \lesssim h\, \norm{P^\msO - \widehat P^\target}\,,
\end{align*}
from Lemma~\ref{lem:hmc-expansion}.

Now, we can write the two components in the first term as, respectively,
\begin{align*}
    \Bigl\lVert\int_0^h (h-t)\, t\, [\nabla^2 V(X_t^\msOH) - \nabla^2 V(x_0)] \, \D t\Bigr\rVert^2_{\rm F} \lesssim h^7 \int_0^h \norm{\nabla^3 V(X_t^\msOH)\, P_t^\msOH}^2_{\rm F} \, \D t\,,
\end{align*}
while the other can be bounded as
\begin{align*}
    \Bigl\lVert \int_0^h (h-t)\, t\,[\nabla^2 V(x_0 + h P^\msOB)-\nabla^2 V(x_0)] \, \D t\Bigr\rVert^2_{\rm F} \lesssim h^7 \int_0^h \norm{\nabla^3 V(x_0 + t P^\msOB)\, P^\msOB}^2_{\rm F} \, \D t\,.
\end{align*}
We again combine this with the operator norm bounds $\norm{K}_{\op} \lesssim h^{-1}$ to conclude the bound under Hessian Lipschitzness; namely, in line with Lemma~\ref{lem:hmc-distortion-II}, we have that
\begin{align*}
    K(\widehat \Phi_h^\msx(x_0, P^\msO))= \bigl[(\nabla_2 \Phi_h^\msx)(x_0, \widehat P^\target)\bigr]^{-1}\,.
\end{align*}
As a result,
\begin{align*}
    \norm{K(\widehat \Phi_h^\msx(x_0, P^\msO))\,\widehat J(P^\msO) - I}_{\rm F}^2 \leq \norm{K(\widehat \Phi_h^\msx(x_0, P^\msO))}_{\op}^2\, \norm{\widehat J(P^\msO) - \nabla_p \Phi_h^\msx(x_0 \widehat P^\target)}_{\rm F}^2\,. 
\end{align*}
Combining the bounds yields the result.
\end{proof}

We will need an auxiliary lemma to handle some expectations even under the stationary measure.
\begin{lemma}\label{lem:perturb-renyi}
    Consider the transform $T_t: (x, p) \mapsto (x + t\,(p - \frac{h}{2}\, \nabla V(x))\,,\; p - \frac{h}{2}\, \nabla V(x))$ from the definition of the $\mathrm{OBABCO}$ algorithm. Then, we have for $t \in [0, h]$ and $h \ll \nfrac{1}{\beta^{1/2} q^{1/2}}$,
    \begin{align*}
        \Renyi_q((T_t)_{\#} \bs \pi \mmid \bs \pi) &\lesssim \beta dh^2 q\,.
    \end{align*}    
    A similar bound holds when pushing forward via $\bar T_t: (x, p) \mapsto (x + tp, p)$.
\end{lemma}
\begin{proof}
    The transform $T_t$ produces the Jacobian
    \begin{align*}
        J \deq \nabla T_t(x, p) = \begin{bmatrix}
            I - \frac{ht}{2}\, \nabla^2 V(x) & tI \\[0.25em]
            -\frac{h}{2}\, \nabla^2 V(x) & I
        \end{bmatrix}\,,
    \end{align*}
    which has $\det(J) = \det(J_{1,1} - J_{1,2} J_{2,1}) = 1$.\footnote{This is not a coincidence; the leapfrog integrator is symplectic.} Furthermore, the inverse of this transform is $T_t^{-1}(x', p') = (x'-tp',\, p' + \frac{h}{2}\, \nabla V(x' -tp'))$. As a result, we have the exact formula
    \begin{align*}
        \frac{\D (T_t)_\# {\bs \pi}}{\D {\bs \pi}}(x', p') = \exp\Bigl(V(x') - V(x'-tp') + \frac{1}{2}\,\bigl(\norm{p'}^2 - \norm{p' + \tfrac{h}{2}\, \nabla V(x'-tp')}^2\bigr) \Bigr)\,.
    \end{align*}
    Expanding, we find
    \begin{align*}
        V(x') - V(x'-tp') = \int_0^t \langle\nabla V(x' - sp'), p'\rangle \, \D s = t \,\langle\nabla V(x'), p'\rangle -\int_0^t (t-s)\, \langle p', \nabla^2 V(x'-sp')\, p'\rangle \, \D s\,.
    \end{align*}
    As a result, it follows that
    \begin{align*}
        &V(x') - V(x'-tp')+ \frac{1}{2}\,\bigl(\norm{p'}^2 - \norm{p' + \tfrac{h}{2}\, \nabla V(x'-tp')}^2\bigr) \\
        &\qquad= \bigl(t-\frac{h}{2}\bigr)\, \langle\nabla V(x'), p'\rangle + \frac{1}{2} \int_0^{t} (h-2\,(t-s)) \,\langle p', \nabla^2 V(x'-sp'), p'\rangle \, \D s - \frac{h^2}{8}\, \norm{\nabla V(x'-tp')}^2\,.
    \end{align*}
    Drop the final term. Note that the eigenvalues of the matrix in the quadratic form are bounded, so we can bound, using that the momentum and position coordinates are independent at stationarity and that $P$ is a standard Gaussian,
    \begin{align*}
        \E_{\bs \pi} \Bigl(\frac{\D (T_t)_\# {\bs \pi}}{\D {\bs \pi}}\Bigr)^q &\leq \E_{\bs \pi}\Bigl[\exp\Bigl(\bigl(t-\frac{h}{2}\bigr)\,q \,\langle\nabla V(X), P\rangle + \frac{\beta h^2q}{2}\, \norm{P}^2\Bigr)\Bigr] \\
        &= (1-\beta h^2 q)^{-d/2}\, \E_{\bs \pi}\Bigl[\exp\Bigl(\tfrac{(t-\frac{h}{2})^2\,q^2}{2\,(1-\beta h^2 q)}\, \norm{\nabla V(X)}^2\Bigr) \Bigr] \,.
    \end{align*}
    Recall that $\norm{\nabla V(X)} \lesssim \sqrt{\beta d} + \sqrt{\beta \log \frac{1}{\delta}}$ with probability $1-\delta$ under $\bs \pi$ (Lemma~\ref{lem:subgsn-score}). As a result, taking logarithms, we find that for $h \ll \nfrac{1}{\beta^{1/2} q^{1/2}}$,
    \begin{align*}
        \Renyi_{q}((T_t)_{\#} \bs \pi \mmid \bs \pi) &\lesssim \frac{d}{q-1} \log \frac{1}{1-\beta h^2 q} + \beta d h^2 q\lesssim \beta d h^2 q\,.
    \end{align*}
    This establishes the first bound.

    As for the second transform, the Jacobian again has unit determinant (it is upper triangular), while
    \begin{align*}
        \frac{\D(\bar T_t)_\# \bs \pi}{\D \bs \pi}(x_0, p_0) = \exp\bigl(V(x_0) - V(x_0-tp_0)\bigr) \leq \exp\bigl(t \inner{\nabla V(x_0), p_0} + \frac{\beta t^2}{2}\, \norm{p_0}^2\bigr)\,.
    \end{align*}
    The remainder of the bound proceeds identically to the analysis above.
\end{proof}

The following lemma is the main high-probability bound for the most problematic error term.

\begin{lemma}\label{lem:stoc-proc-bound}
    Suppose we draw $(X, P) \sim \bs \mu$ and produce $X^\msO, P^\msOB$, $(X^\msOH_t, P^\msOH_t)_{t \in [0,h]}$ from this starting point. Then, with probability $1-\delta$, under Assumption~\ref{as:higher-reg}, with $h \ll \nfrac{1}{\beta^{1/2} d^{1/4}}$, we can bound
    \begin{align*}
        \int_{0}^{h} \norm{\nabla^3 V(X^\msO+tP^{\msOB})[P^{\msOB}, P^{\msOB}]}^2 \, \D t
        &\lesssim \beta_\eff^2 h\,\bigl(d + \Renyi_2(\bs\mu\mmid\bs\pi)^2 + \log^2 \frac{1}{\delta}\bigr)\,.
    \end{align*}
    Likewise, for $h \ll \nfrac{1}{\beta^{1/2}}$,
    \begin{align*}
        \int_{0}^{h} \norm{\nabla^3 V(X^\msOH_t)[P^{\msOH}_t, P^{\msOH}_t]}^2 \, \D t
        &\lesssim \beta_\eff^2 h\,\bigl(d + \Renyi_2(\bs\mu\mmid\bs\pi)^2 + \log^2 \frac{1}{\delta}\bigr)\,.
    \end{align*}
\end{lemma}
\begin{proof}
    Let us begin by tackling the first integral.
    First, note that under $\bs \pi$, we have the following bound by Lemma~\ref{lem:chaos}:
    \begin{align*}
        \norm{\nabla^3 V(X)[P, P]}^2 \lesssim \beta_\eff^2\,\bigl( d + \log^2 \frac{1}{\delta}\bigr)\,,
    \end{align*}
    since under $\bs \pi$, $P$ is independent of $X$.
    
    Furthermore, let $\bs P^\msO$ be the~\eqref{eq:O} semigroup, and $T_t$ the transform from Lemma~\ref{lem:obabo-distortion}, for $t\in [0, h)$, $h \ll \nfrac{1}{\beta^{1/2}}$. We have, using the weak triangle inequality for R\'enyi divergences (Proposition~\ref{prop:renyi})
    \begin{align*}
        \Renyi_{3/2}\bigl((T_t)_\#(\bs \mu \bs P^\msO) \bigm\Vert \bs \pi \bigr) &\lesssim \Renyi_2\bigl((T_t)_\#(\bs \mu \bs P^\msO) \bigm\Vert (T_t)_\#(\bs \pi \bs P^\msO) \bigr) + \Renyi_3\bigl((T_t)_\#(\bs \mu \bs P^\msO) \bigm\Vert \bs \pi\bigr) \\
        &\le \Renyi_2(\bs \mu \mmid \bs \pi) + \Renyi_3\bigl((T_t)_\#(\bs \mu \bs P^\msO) \bigm\Vert \bs \pi\bigr)
        \lesssim \Renyi_2(\bs\mu\mmid\bs\pi)+ \beta dh^2\,,
    \end{align*}
    where the second inequality follows by the data-processing inequality on the first term, and the fact that $\bs \pi \bs P^\msO = \bs \pi$ in the second term. We then invoke Lemma~\ref{lem:obabo-distortion}.

    It follows that, with the initial iterate drawn from $\bs \mu$, via Lemma~\ref{lem:high-prob-change}, that with probability $1-\delta$, $\delta \in (0, \nfrac{1}{2})$, we have
    \begin{align}
        \norm{\nabla^3 V(X^\msO + tP^\msOB)[P^\msOB, P^\msOB]}^2 
        &\lesssim \beta_{\eff}^2\,\bigl( d + \log^2 \frac{1}{\delta} + \Renyi_2(\bs\mu\mmid\bs\pi)^2 + \beta^2 d^2 h^4\bigr) \nonumber\\[0.25em]
        &\lesssim \beta_{\eff}^2\,\bigl( d + \log^2 \frac{1}{\delta} + \Renyi_2(\bs\mu\mmid\bs\pi)^2\bigr)\label{eq:high-prob-mu}
    \end{align}
    provided $h\lesssim \nfrac{1}{\beta^{1/2} d^{1/4}}$.

    Next, by Minkowski's integral inequality,~\eqref{eq:high-prob-mu}, and the equivalence between tail bounds and growth of $L^p$ norms (see Lemma~\ref{lem:tail-lp-equivalence}),
    \begin{align*}
        \Bigl\lVert \int_0^h\norm{\nabla^3 V(X^\msO + tP^\msOB)[P^\msOB, P^\msOB]}^2 \, \D t\Bigr\rVert_{L^p}
        &\le \int_0^h \bigl\lVert \norm{\nabla^3 V(X^\msO + tP^\msOB)[P^\msOB, P^\msOB]}^2\bigr\rVert_{L^p}\,\D t \\
        &\lesssim \beta_\eff^2 h\,\bigl(d + \Renyi_2(\bs\mu\mmid\bs\pi)^2 + p^2\bigr)\,.
    \end{align*}
    Applying the equivalence again, we see that with probability $1-\delta$,
    \begin{align*}
        \int_0^h\norm{\nabla^3 V(X^\msO + tP^\msOB)[P^\msOB, P^\msOB]}^2 \, \D t
        &\lesssim \beta_\eff^2 h\,\bigl(d + \Renyi_2(\bs\mu\mmid\bs\pi)^2 + \log^2 \frac{1}{\delta}\bigr)\,.
    \end{align*}

    The second bound is deduced the exact same way, only we note that for $\bs P^\msOH_t$ corresponding to running the~\eqref{eq:O} semigroup for time $\nfrac h 2$ followed by~\eqref{eq:H} for any time $t \in [0, h)$, we have $\bs \pi \bs P^\msOH_t = \bs \pi$, so that
    \begin{align*}
        \Renyi_2\bigl(\bs\mu\bs P^\msOH_t \mmid \bs \pi )
        &\le \Renyi_2(\bs \mu \mmid \bs \pi)\,.
    \end{align*}
    The high-probability bound follows the same template.
\end{proof}

\subsubsection{Proof of Theorem~\ref{thm:main-warm}}\label{sssec:main_warm_pf}

\begin{proof}[Proof of Theorem~\ref{thm:main-warm}]
    We initialize OBABCO at $\bs\mu_0 \deq \cN(0, \beta^{-1} I) \otimes \cN(0, I)$ and $\gamma =\sqrt{32\beta}$.
    We assume that the step size condition~\eqref{eq:big_stepsize_cond} holds (which will indeed occur for the choices we make below).

    By the R\'enyi triangle inequality (Proposition~\ref{prop:renyi}),
    \begin{align}\label{eq:renyi_decomp}
        \Renyi_q(\bs\mu_0\widehat{\bs P}_h^n \mmid \bs\pi)
        &\lesssim \Renyi_{2q}(\bs\mu_0\widehat{\bs P}_h^n \mmid \bs\mu_0\bs P_h^n) + \Renyi_{2q-1}(\bs\mu_0\bs P_h^n\mmid\bs \pi)\,.
    \end{align}
    By Theorem~\ref{thm:harnack}, provided $nh \gtrsim (\beta^{1/2}/\alpha)\log\kappa$ and $h \ll \nfrac{1}{\beta^{1/2} \kappa^{1/2}}$,
    \begin{align*}
        \Renyi_{2q-1}(\delta_{x,p}\bs P_h^n \mmid \delta_{\bar x,\bar p}\bs P_h^n)
        &\lesssim \alpha q \exp\bigl(-\Omega\bigl(\frac{\alpha nh}{\beta^{1/2}}\bigr)\bigr) \,\bigl\{\norm{x-\bar x}^2 + \frac{1}{\beta}\,\norm{p-\bar p}^2\bigr\}\,.
    \end{align*}
    By the convexity principle (see~\cite[Theorem 4]{scr1}), it implies, for $\Coup(\bs \mu, \bs \nu)$ representing the set of couplings, $\gamma \in \mc P(\R^{2d} \times \R^{2d})$ such that its first marginal is $\bs \mu$ and its second marginal is $\bs \nu$,
    \begin{align*}
        &\exp(2\,(q-1)\,\Renyi_{2q-1}(\bs\mu_0\bs P_h^n \mmid \bs\pi)) \\
        &\qquad \le \inf_{\gamma \in \Coup(\bs\mu_0,\bs\pi)} \int \exp\Bigl(O\Bigl(\alpha q^2 \exp\bigl(-\Omega\bigl(\frac{\alpha nh}{\beta^{1/2}}\bigr)\bigr)\,\bigl\{\norm{x-\bar x}^2 + \frac{1}{\beta}\,\norm{p-\bar p}^2\bigr\}\Bigr)\Bigr)\,\gamma(\D(x,p),\D(\bar x,\bar p)) \\
        &\qquad \le \inf_{\gamma \in \Coup(\cN(0,\beta^{-1} I),\pi)} \int \exp\Bigl(O\Bigl(\alpha q^2 \exp\bigl(-\Omega\bigl(\frac{\alpha nh}{\beta^{1/2}}\bigr)\bigr)\,\norm{x-\bar x}^2 \Bigr)\Bigr)\,\gamma(\D x,\D\bar x)\,.
    \end{align*}
    By taking, e.g., an independent coupling, for $nh\gtrsim (\beta^{1/2}/\alpha) \log q$,
    \begin{align}\label{eq:oho_renyi_bd}
        \Renyi_{2q-1}(\bs\mu_0 \bs P_h^n\mmid \bs\pi)
        &\lesssim dq \exp\bigl(-\Omega\bigl(\frac{\alpha nh}{\beta^{1/2}}\bigr)\bigr)\,.
    \end{align}

    In phase 1, we choose a number of iterations $N_1$ such that $N_1 h \asymp (\beta^{1/2}/\alpha) \log(\kappa d)$, so that from~\eqref{eq:oho_renyi_bd} we have $\Renyi_3(\bs\mu_0\bs P^{N_1}\mmid\bs\pi) \lesssim d^{1/2}$.
    Then, by~\eqref{eq:renyi_decomp} and Theorem~\ref{thm:coarse}, since $\Renyi_2(\bs\mu_0\mmid\bs \pi) \le \frac{d}{2}\log \kappa$ by Lemma~\ref{lem:initialization}, we can ensure that $\Renyi_2(\bs\mu_0\widehat{\bs P}_h^{N_1} \mmid \bs\pi) \lesssim d^{1/2}$, provided
    \begin{align*}
        h \ll \frac{1}{\beta^{1/2} \kappa^{1/2} d^{1/4} \log^{1/2}(\kappa d)}\,.
    \end{align*}

    Next, for phase 2, we use~\eqref{eq:renyi_decomp} but with $\bs\mu_0$ replaced by $\bs\mu_0 \widehat{\bs P}_h^{N_1}$.
    We choose a number of iterations $N_2$ such that $N_2 h \asymp (\beta^{1/2}/\alpha)\log(\kappa dq/\varepsilon^2)$. Using a similar argument as the proof of~\eqref{eq:oho_renyi_bd} but with an extra change of measure step to argue that the coupling cost is bounded under $\bs\mu_0\widehat{\bs P}_h^{N_1}$, we can ensure that $\Renyi_{2q-1}(\bs\mu_0 \widehat{\bs P}_h^{N_1}\bs P_h^{N_2}\mmid \bs\pi) \ll \varepsilon^2$.
    From Theorem~\ref{thm:obabco-final-renyi}, it follows that $\Renyi_{2q}(\bs\mu_0\widehat{\bs P}_h^{N_1+N_2} \mmid \bs\mu_0\widehat{\bs P}_h^{N_1} \bs P_h^{N_2}) \ll \varepsilon^2$ as well, provided that
    \begin{align*}
        \beta^{1/2} d^{1/4} h \ll_{\log} \frac{\varepsilon^{1/2}}{\kappa^{1/4} q^{1/4}} \wedge \frac{\kappa^{1/2} \varepsilon^{1/2}}{\kapH^{1/2} q^{1/4}} \wedge \frac{1}{\kappa^{1/2} q} \wedge \frac{\kappa^{3/7} d^{1/4-1/7}}{\kapH^{2/7}} \wedge \frac{\kappa^{1/2} d^{1/4}}{\kapH^{1/2} q^{1/2}}\,,
    \end{align*}
    where $\ll_{\log}$ hides logarithmic factors.
    If we choose $h$ to be the largest allowable value, then the total number of iterations is
    \begin{align*}
        N_1 + N_2 = \widetilde O\Bigl(\kappa d^{1/4}\,\Bigl(\kappa^{1/2} q + \frac{\kapH^{2/7} }{\kappa^{3/7} d^{1/4-1/7}} + \frac{\kappa^{1/4} q^{1/4}}{\varepsilon^{1/2}} + \frac{\kapH^{1/2} q^{1/4}}{\kappa^{1/2} \varepsilon^{1/2}} + \frac{\kapH^{1/2} q^{1/2}}{\kappa^{1/2} d^{1/4}} \Bigr)\Bigr)\,.
    \end{align*}
    Setting $q = O(1)$ and $\varepsilon = 1$ yields the result, noting that $\kappa^{4/7} \kapH^{2/7} \le \kappa^{3/2} + \kappa^{1/2} \kapH^{1/2}$ by Young's inequality.
\end{proof}

    \paragraph*{Acknowledgements.} JMA acknowledges funding from a Sloan Research Fellowship and a Seed Grant Award from Apple. MSZ is funded by an NSERC CGS-D award. 

    \newpage
    \appendix

    \section{Helper lemmas}

\begin{lemma}[Hamiltonian flow growth]\label{lem:ham-flow-grow}
    Assume that $-\beta I \preceq \nabla^2 V \preceq \beta I$ and that $(x_t, p_t)_{t\ge 0}$ evolve according to the Hamiltonian flow~\eqref{eq:ham_ode}.
    If $h \ll \nfrac{1}{\beta^{1/2}}$, then
    \begin{align*}
        &\sup_{t \in [0, h]} \norm{p_t} \lesssim \norm{p_0} + h\, \norm{\nabla V(x_0)}\,, \\
        &\sup_{t \in [0, h]} \norm{\nabla V(x_t)} \lesssim \norm{\nabla V(x_0)} + \beta h\, \norm{p_0}\,.
    \end{align*}
\end{lemma}
\begin{proof}
    By definition of the Hamiltonian flow~\eqref{eq:ham_ode},
    \begin{align*}
        \sup_{t\in [0,h]}\norm{x_t - x_0} \le h\,\norm{p_0} + \frac{h^2}{2}\sup_{t\in [0,h]}{\norm{\nabla V(x_t)}}
        \le h\,\norm{p_0} + \frac{h^2}{2}\,\norm{\nabla V(x_0)} + \frac{\beta h^2}{2} \sup_{t\in [0,h]}{\norm{x_t - x_0}}\,.
    \end{align*}
    Thus, for $h \ll \nfrac{1}{\beta^{1/2}}$, $\sup_{t\in [0,h]}{\norm{x_t - x_0}} \lesssim h\,\norm{p_0} + h^2\,\norm{\nabla V(x_0)}$.
    The two statements in the lemma follow from this and the Lipschitzness of $\nabla V$.
\end{proof}

We also make use of the following suite of helper lemmas from~\cite{zhang2025analysis}, which we recall here for the convenience of the reader. 
\begin{lemma}[{Score concentration, adapted from~\cite[Lemma 6.2.7]{chewibook}}]\label{lem:subgsn-score}
    Let $\pi \propto \exp(-V)$ and assume that $\nabla V$ is $\beta$-Lipschitz. Then $\nabla V$ is $\sqrt{\beta}$ sub-Gaussian under $\pi$, in that for all $v \in \R^d$,
    \begin{align*}
        \E_\pi \exp \inner{\nabla V, v} \leq \exp\Bigl(\frac{\beta \norm{v}^2}{2} \Bigr)\,.
    \end{align*}
    In particular, for $0 < \delta < 1/2$, with probability at least $1-\delta$ under $\pi$,
    \begin{align*}
        \norm{\nabla V}^2 \lesssim \beta d + \beta \log \frac{1}{\delta}\,.
    \end{align*}
\end{lemma}

\begin{lemma}[{R\'enyi change of measure, adapted from~\cite[Lemma 6.2.8]{chewibook}}]\label{lem:renyi-change-measure}
    Consider two probability measures $\mu,\pi$ and a test function $\phi : \Omega \to \R$. Suppose that 
    \begin{align*}
        \pi\{\phi \geq {h}\} \leq \psi({h})\,,
    \end{align*}
    holds for some growth function $\psi$. Then
    \begin{align*}
        \mu\{\phi \geq {h}\} \leq \psi({h})^{(q-1)/q} \exp\Bigl(\frac{q-1}{q}\, \Renyi_q(\mu \mmid \pi) \Bigr)\,.
    \end{align*}
\end{lemma}

The following result easily follows from the preceding lemma.

\begin{lemma}[Sub-Weibull change of measure]\label{lem:high-prob-change}
    Suppose that under $(X, P) \sim \bs \pi$, we have for some measurable function $f: \R^{2d} \to \R$ with probability $1-\delta$,
    \begin{align*}
        f(X, P) \lesssim \msf C\,\bigl( d + \log^2 \frac{1}{\delta}\bigr)\,.
    \end{align*}
    Then, for any $q > 1$, we have with probability at least $1-\delta$ for $(\widehat X, \widehat P) \sim \widehat{\bs \mu}$,
    \begin{align*}
        f(\widehat X, \widehat P) \lesssim_q \msf C\,\bigl(d+ \log^2 \frac{1}{\delta} + \Renyi_q(\widehat{\bs\mu} \mmid \bs\pi)^2\bigr)\,.
    \end{align*}
\end{lemma}

\begin{lemma}[Linearization of $\log \det$]\label{lem:log-det}
    Consider $M \in \R^{d\times d}$ with $\norm{M}_{\op} \leq c < 1$. Then
    \begin{align*}
        \log \det(I + M) = \tr(M) - \frac{1}{2} \tr(M^2) + O_c(\norm{M}^3_{\op}\, d)\,.
    \end{align*}
\end{lemma}
\begin{proof}
    Note that $I + tM$ is invertible for $|t| < 1/\norm{M}_{\op}$. Thus we can compute $\frac{\D }{\D t} \log \det (I + tM) = \tr((I + tM)^{-1} M) = \sum_{k \geq 0} (-t)^k \Tr(M^{k+1})$. Integrating $t$ from $0$ to $1$ gives the power series expansion
    \begin{align*}
        \log \det (I + M) = \sum_{k \geq 1} \frac{(-1)^{k+1}}{k} \Tr(M^k)\,.
    \end{align*}
    Since $\Tr(M^k)| \leq d \norm{M}_{\op}^k$, we can bound the remainder by $|\sum_{k \geq 3} (-1)^{k+1} \Tr(M^k) / k| \leq d \sum_{k \geq 3} \norm{M}_{\op}^k / k \leq d\norm{M}_{\op}^3 / (3(1 - \norm{M}_{\op}))$.
\end{proof}

\begin{lemma}[{Equivalence between $L^p$ bounds and tail decay; adapted from~\cite[Exercise 2.43]{vershynin2018high}}]\label{lem:tail-lp-equivalence}
    Suppose that $X \sim \pi \in \mc P(\R_+)$ obeys, with some constants $\alpha > 0$, $\msf C_1, \msf C_2 \in \R_+$,
    \begin{align}\label{eq:tail-bound-generic}
        X \lesssim \msf C_1 + \msf C_2 \log^\alpha \frac{1}{\delta}\qquad\text{with probability}\ge 1-\delta\,,~\text{for all}~0 < \delta < \frac{1}{2}\,.
    \end{align}
    Then,
    \begin{align}\label{eq:lp-bound-generic}
        \norm{X}_{L^p} \lesssim \msf C_1 + \msf C_2 p^\alpha\qquad\text{for all}~p \ge 1\,.
    \end{align}
    Furthermore,~\eqref{eq:lp-bound-generic} implies~\eqref{eq:tail-bound-generic} with possibly different absolute constants.
\end{lemma}

    \section{Proof of Lemma~\ref{lem:chaos}}\label{app:chaos}

    We start by noting that
    \begin{align*}
        \norm{T[\xi,\xi]}
        &= \Bigl\lVert \sum_{i,j=1}^d T_{i,j,\bullet} \xi_i \xi_j\Bigr\rVert
        \le \underbrace{\Bigl\lVert \sum_{i,j=1}^d T_{i,j,\bullet} \one_{i=j}\Bigr\rVert}_{\le \bar\beta d^{1/2}}
        + \underbrace{\Bigl\lVert \sum_{i,j=1}^d T_{i,j,\bullet} (\xi_i \xi_j -\one_{i=j})\Bigr\rVert}_{\eqqcolon Z}\,,
    \end{align*}
    where we use $\norm{\sum_{i,j=1}^d T_{i,j,\bullet} \one_{i=j}} \le \bar\beta\,\norm I_{\rm F} = \bar\beta d^{1/2}$.

    Next, we use the Banach-valued Hanson--Wright inequality in~\cite[Theorem 6]{AdaLatMel20HWBanach} via Remark 9 therein, which yields with probability at least $1-\delta$,
    \begin{align*}
        Z
        &\lesssim \E Z + R + U\sqrt{\log\frac{1}{\delta}} + V \log\frac{1}{\delta}
    \end{align*}
    where the terms $R,U,V$ are defined (and bounded) as follows. First, letting $G_{i,j} \sim \cN(0,1)$ be i.i.d.,
    \begin{align*}
        R
        \deq
        \E \Bigl\lVert \sum_{i \neq j} T_{i,j,\bullet} G_{i,j} \Bigr\lVert
        \leq
        \Bigl( \E \Bigl \lVert \sum_{i \neq j} T_{i,j,\bullet} G_{i,j} \Bigr\lVert^2 \Bigr)^{1/2}
        =
        \Bigl(  \sum_{i \neq j}  \lVert T_{i,j,\bullet} \lVert^2 \Bigr)^{1/2}
        \leq 
        \Bigl(  \sum_{k=1}^d \norm{T_{\bullet,\bullet,k}}_{\rm F}^2 \Bigr)^{1/2}
        \leq
        \bar \beta \sqrt{d}\,.
    \end{align*}
    Next,
    \begin{align*}
        U
        &\deq \sup_{\norm x\le 1} \E\Bigl\lVert \sum_{i,j=1}^d T_{i,j,\bullet} x_i \xi_j\Bigr\rVert + \sup_{\norm x_{\rm F} \le 1}{\Bigl\lVert \sum_{i,j=1}^d T_{i,j,\bullet} x_{i,j}\Bigr\rVert}
        \le \sup_{\norm x \le 1} \biggl\{\E\Bigl[\Bigl\lVert \sum_{i,j = 1}^d T_{i,j,\bullet} x_i \xi_j\Bigr\rVert^2\Bigr]\biggr\}^{1/2} + \bar\beta \\[0.25em]
        &= \sup_{\norm x \le 1} \biggl\{\sum_{j=1}^d \Bigl\lVert \sum_{i=1}^d T_{i,j,\bullet} x_i\Bigr\rVert^2\biggr\}^{1/2} + \bar\beta
        = \sup_{\norm x \le 1} \biggl\{\sum_{j,k=1}^d \Bigl(\sum_{i=1}^d T_{i,j,k} x_i\Bigr)^2\biggr\}^{1/2} + \bar\beta
        \le 2\bar\beta\,,
    \end{align*}
    where we note that $\sum_{i=1}^d T_{i, \bullet, \bullet} x_i$ yields a matrix with Frobenius norm bounded by $\bar\beta$, by the definition of the $\{1, 2\},\! \{3\}$ norm and the symmetry of the tensor.
    Also,
    \begin{align*}
        V \deq \sup_{\norm x\le 1,\,\norm y\le 1}{\Bigl\lVert \sum_{i,j=1}^d T_{i,j,\bullet} x_i y_j \Bigr\rVert}
        \le \sup_{\norm A_{\rm F} \le 1}{\Bigl\lVert \sum_{i,j=1}^d T_{i,j,\bullet} A_{i,j}\Bigr\rVert}
        \le \bar\beta\,.
    \end{align*}
    Finally, for the expectation, use $\E[Z] \leq (\E[Z^2])^{1/2}$ and the upper bound
    \begin{align*}
        \E[Z^2]
        &= \sum_{k=1}^d \E\Bigl[\Bigl\lvert\sum_{i,j=1}^d T_{i,j,k}(\xi_i \xi_j - \one_{i = j})\Bigr\rvert^2\Bigr]
        = \sum_{k=1}^d \var{\langle \xi, T_{\bullet, \bullet, k}\,\xi\rangle}
        = 2\sum_{k=1}^d \norm{T_{\bullet, \bullet, k}}_{\rm F}^2 
        \le 2\bar\beta^2 d\,.
    \end{align*}
    Combining the above displays concludes the proof.

    \section{OHO dynamics in general settings}\label{sec:app-oho}

In \S\ref{sec:oho} we established Harnack inequalities for OHO (Theorem~\ref{thm:harnack}) in the setting of strongly convex potentials $(\alpha > 0)$ and high friction $(\gamma \gtrsim \sqrt{\beta})$, since that is the setting of relevance for our application to sampling algorithms in \S\ref{sec:app}. Here we briefly mention that, with very little change to the analysis, these Harnack inequalities extend more generally to potentials that are weakly convex or semi-convex, and to settings with low friction. Although we do not use these Harnack inequalities in our applications, we state them here as they may be of independent interest, since  
Harnack inequalities provide detailed information about long-term regularity and convergence of the OHO dynamics. 

Just like the Harnack inequalities for underdamped Langevin dynamics (which we recall is the limit of OHO as the step size $h \searrow 0$, see \S\ref{ssec:oho_overview}), the strength of the Harnack inequalities for OHO are governed by the following parameter:
\begin{align}\label{eq:omeg-app}
    \omega \deq \begin{cases}
        \alpha/\gamma\,, & \text{if $\gamma \geq \sqrt{32\beta}$ \quad (high friction setting)}\,, \\
         -\sqrt{\beta}\,, & \text{if $\gamma < \sqrt{32\beta}$ \quad (low friction setting)}\,.
    \end{cases}
\end{align}

\begin{theorem}[Harnack inequalities for OHO in settings beyond Theorem~\ref{thm:harnack}]\label{thm:harnack-general}
    Assume that $-\beta I \preceq \alpha I \preceq \nabla^2 V \preceq \beta I$. Suppose that the following step size condition holds:
    \begin{align*}
        h \ll \frac{1}{\beta^{1/2} q^{1/2}} \wedge \begin{cases}
            \nfrac{1}{\gamma^{3/2} T^{1/2}}\,, & \text{if}~\alpha = 0~\text{and}~\gamma \ge \sqrt{32\beta} \quad(\text{weakly convex, high friction})\,, \\
            \nfrac{\abs{\omega}}{\gamma^{3/2}}\,, & \text{if}~\alpha < 0~\text{or}~\gamma < \sqrt{32\beta}\quad (\text{semi-convex or low friction})\,.
        \end{cases}
    \end{align*}
    Then the conclusions of Theorem~\ref{thm:harnack} hold.
\end{theorem}

Theorem~\ref{thm:harnack-general} covers all settings of the convexity parameter $\alpha$ and the friction parameter $\beta$ not covered by Theorem~\ref{thm:harnack}. As in Remark~\ref{rem:harnack-recover} (which discusses the setting in Theorem~\ref{thm:harnack}),
the Harnack inequalities in the settings of Theorem~\ref{thm:harnack-general} exactly match the analogous bounds for underdamped Langevin in~\cite[Theorem 3.2]{scr4} up to absolute constants, except for an additive term $\beta^2 d h^3 q T$ without Assumption~\ref{as:hess-lip}. We are unsure if this additive term is fundamental, but it vanishes as $h\searrow 0$, thereby recovering the bounds for underdamped Langevin in the continuous-time limit. 

The proof of Theorem~\ref{thm:harnack-general} is identical to Theorem~\ref{thm:harnack} once we show that the key lemma about the auxiliary distance recursion (Lemma~\ref{lem:shift-contraction}) extends to these settings. We do this below. 

\begin{lemma}[Auxiliary distance recursion in settings beyond Lemma~\ref{lem:shift-contraction}]\label{lem:shift-contraction-general}
    Suppose that $-\beta I \preceq \alpha I \preceq \nabla^2 V \preceq \beta I$ and that the step size $h$ satisfies
     \begin{align*}
        h \ll \begin{cases}
            \nfrac{1}{\gamma^{3/2} T^{1/2}}\,, & \text{if}~\alpha = 0~\text{and}~\gamma \ge \sqrt{32\beta} \quad(\text{weakly convex, high friction})\,, \\
            \nfrac{\abs{\omega}^{1/2}}{\gamma^{3/2}}\,, & \text{if}~\alpha < 0~\text{or}~\gamma < \sqrt{32\beta}\quad (\text{semi-convex or low friction})\,.
        \end{cases}
    \end{align*}
    Then the conclusion of Lemma~\ref{lem:shift-contraction} holds.
\end{lemma}
\begin{proof}
    The proof is nearly identical to that of Lemma~\ref{lem:shift-contraction}; we only describe the differences in verifying the subclaims (i), (ii), and (iii).
    \begin{enumerate}[label=(\roman*)]
        \item In the low friction setting, we have $c_t \lambda \leq \frac{1}{12} (\eta_t^{\msp})^2$ if we take $c_0$ large enough. Thus this still gives the desired conclusion $\frac{c_t \lambda}{\gamma_t} \leq \frac{\gamma_t}{12}$.
        \item We need only justify~\eqref{eq:gam_bd_blah}. This follows from the step size assumption. 
        \item We need only adjust the final step in the proof of this claim, regarding $c_t \alpha + \frac{\gamma_t \eta_t^{\msp}}{2}$. In the semi-convex, high friction case, the second term is $\gamma_t \eta_t^{\msp} \geq c_0 \gamma \abs{\omega} \geq c_0 \abs{\alpha}$, which allows us to absorb $c_t \alpha$ into it if $c_0 \geq 24$ is large enough. In the low friction case, using $\alpha \geq -\beta$, $\gamma_t \eta_t^{\msp} \geq (\eta_t^{\msp})^2 \geq c_0^2 \omega^2 \geq c_0^2 \beta$. If $c_0 \geq 24$, we can again absorb the $c_t \alpha$ term into this. In any case, this gives the desired conclusion.
    \end{enumerate}
\end{proof}

    \addcontentsline{toc}{section}{References}
    \printbibliography{}
    
\end{document}